\documentclass[12pt, reqno]{amsart}
\usepackage{mathrsfs}
\usepackage{amsmath}
\usepackage[height=22cm, width=16.5cm, hmarginratio={1:1}]{geometry}
\usepackage{diagbox}
\usepackage{tikz}
\usetikzlibrary{positioning, shapes.geometric}
\usepackage{newclude}
\usepackage{appendix}
\usepackage{extarrows}
\usepackage{ulem}

\def\sspan{\mathrm{span}}
\usepackage[hidelinks,hyperindex,breaklinks]{hyperref}
\usepackage{tikz-cd}
\usetikzlibrary{positioning, arrows,decorations.pathmorphing,shapes}
\usetikzlibrary{shapes.geometric}

\usepackage{array}
\usepackage{fullpage}

\author{Shuai Guo}
\address{School of Mathematical Sciences, 
     Peking University, 
     Beijing, 100871, China}
\email{guoshuai@math.pku.edu.cn}

\author{Qingsheng Zhang}
\address{School of Sciences,
	Great Bay University,
     Great Bay Institute for Advanced Study,
     Dongguan, 523000, China}
\email{zhangqingsheng@gbu.edu.cn}

\newcommand{\corr}[1]{\langle\!\langle {#1} \rangle\!\rangle}

\newcommand{\<}{\langle}
\renewcommand{\>}{\rangle}

\def\str{\mathop{\rm str}}
\def\Gdd{\mathrm{Gdd}}
\def\eGdd{\mathrm{eGdd}}

\newcolumntype{C}[1]{>{\centering\arraybackslash$}m{#1}<{$}}
\newlength{\mycolwd}                                         % array column width
\newlength{\mycolwdm}                                         % array column width
\newlength{\mycolwda}                                         % array column width
\newlength{\mycolwdb}                                         % array column width
\newlength{\mycolwdc}                                         % array column width
\newlength{\mycolwdd}                                         % array column width
\newlength{\mycolwddd}                                         % array column width
\newlength{\mycolwddc}                                         % array column width
\newcommand{\Frob}{H}
\newcommand{\anc}{\tau}
\newcommand{\cc}{c_{1}}
\newcommand{\cdes}{c_{\dvac}}
\newcommand{\virm}{m_{\nu}}

\newcommand{\vac}{\mathbf{v}}
\newcommand{\dvac}{\nu}
\newcommand{\dil}{\mathbf{u}}

\newcommand{\id}{ \hspace{-2pt}\text{I}}
\usepackage {color,graphicx,psfrag, verbatim,amssymb,amscd,enumerate,subfigure}
\def\End{\mathrm{End}}

\def\beq{\begin{equation}}
\def\eeq{\end{equation}}
\newcommand{\be}{\begin{equation*}}
\newcommand{\ee}{\end{equation*}}

\DeclareMathOperator{\Cont}{Cont}

\DeclareMathOperator{\ad}{ad}

\DeclareMathOperator{\Res}{Res}
\DeclareMathOperator{\Aut}{Aut}

\DeclareMathOperator{\ev}{ev}

\newtheorem{dummy}{}[section]
\newtheorem{lemma}[dummy]{Lemma}
\newtheorem{proposition}[dummy]{Proposition}
\newtheorem{theorem}[dummy]{Theorem}
\newtheorem{corollary}[dummy]{Corollary}

\theoremstyle{definition}

\newtheorem{definition}[dummy]{Definition}

\newtheorem{remark}[dummy]{Remark}

\usepackage{graphicx}

\newcommand{\M}{\overline{\mathcal M}}

\newcommand{\tr}{\mathrm{tr}}

\newcommand{\pd}{\partial}

\newcommand{\cF}{\mathcal{F}}

\newcommand{\cA}{\mathcal{A}}
\newcommand{\cD}{\mathcal{D}}
\newcommand{\cL}{\mathcal{L}}

\newcommand{\E}{\mathcal{E}}

\newcommand{\Mbar}{\overline{\mathcal M}}

\newtheorem{manualtheoreminner}{Theorem}
\newenvironment{manualtheorem}[1]{%
	\renewcommand\themanualtheoreminner{#1}%
	\manualtheoreminner
}{\endmanualtheoreminner}

\newtheorem{manualpropositioninner}{Proposition}
\newenvironment{manualproposition}[1]{%
    \renewcommand\themanualpropositioninner{#1}%
    \manualpropositioninner
}{\endmanualpropositioninner}

\newtheorem{manualconjectureinner}{Conjecture}
\newenvironment{manualconjecture}[1]{%
    \renewcommand\themanualconjectureinner{#1}%
    \manualconjectureinner
}{\endmanualconjectureinner}

\def\bs{\mathbf s}

\def\bs{\mathbf s}

\def\cH{\mathcal{H}}
\def\cG{\mathcal{G}}

\begin{document}
\title[Virasoro constraints for CohFT]
{Cohomological  Field Theory with vacuum and its Virasoro constraints}

\maketitle
%\makeindex

\raggedbottom

\begin{abstract}
  This is the first part of a series of papers on {\it Virasoro constraints for Cohomological  Field Theory (CohFT)}. 
For a CohFT with vacuum, we introduce the concepts of \(S\)-calibration and \(\dvac\)-calibration. Then, we define the (formal) total descendent potential corresponding to a given calibration. Finally, we  introduce an additional structure, namely homogeneity, for both the CohFT and the calibrations.

After these preliminary introductions, we propose two crucial conjectures: (1) the ancestor version of the Virasoro conjecture for the homogeneous CohFT with vacuum; and (2) the generalized Virasoro conjecture for the (formal) total descendent potential of a calibrated homogeneous CohFT.
We verify the genus-0 part of these conjectures and deduce a simplified form of the genus-1 part of these conjectures for arbitrary CohFTs. Additionally, we prove the full conjectures for semisimple CohFTs.

As applications, our results yield the Virasoro constraints for the deformed negative \(r\)-spin theory. Moreover, by applying the Virasoro constraints, we discover an extension of Grothendieck's dessins d'enfants theory which is widely  studied in the literature.
\end{abstract}

\setcounter{section}{-1}
\setcounter{tocdepth}{1}
\tableofcontents

\section{Introduction}
The Virasoro conjecture originally proposed  by Eguchi--Hori--Xiong and Katz~\cite{EHX97,EJX98} and generalized by Dubrovin--Zhang~\cite{DZ99}
is one of the most important (conjectural) structures in enumerative geometries such like the Gromov--Witten (GW) theory and the Fan--Jarvis--Ruan--Witten (FJRW) theory. 
It predicts that the generating series $\cD({\bf t};\hbar)$, called the total descendent potential, of the enumerative geometry is annihilated by a sequence of differential operators $\{L_m\}_{m\geq -1}$
satisfying the commutation relation $[L_m,L_n]=(m-n)L_{m+n}$ of (half of) Virasoro algebra.

For the simplest case, the GW theory of a point, the Virasoro conjecture is equivalent~\cite{DVV91} to the celebrated Witten conjecture~\cite{Wit91} first proved by Kontsevich~\cite{Kont92}.
Generally, if the theory satisfies the so called semisimplicity, the Virasoro conjecture is proved in lower genera~\cite{DZ98,DZ99,Get99,LT98,Liu01,Liu02} by Liu--Tian, Liu, Dubrovin--Zhang and Getzler,
and completely solved by Givental and Teleman's works~\cite{Giv01a,Tel12}.
Without the assumption of semisimplicity, the genus-0 part of the Virasoro constraints is proved by Liu--Tian~\cite{LT98}, see also~\cite{Liu01,Liu02,Lin19} for some lower genus explorations.
In~\cite{OP06c}, Okounkov--Pandharipande proved the Virasoro conjecture for the GW theory of curves, see also~\cite{HS21} for some analogous results for quantum singularity theories.

In~\cite{KM94}, Kontsevich and Manin introduced the concept of CohFT to capture the axiomatic properties of the GW theory (see \S \ref{sec:CohFT} for definitions). 
Later, it was discovered that the CohFT properties hold generally in various kinds of enumerative geometries,
 and the Virasoro conjecture can be naturally generalized to the homogeneous CohFTs with flat unit. 
Despite the progress on proving the Virasoro conjectures, we find that the Virasoro constraints are not fully established in various theories such as CohFTs without flat unit and the Virasoro constraints for the ancestor generating series. 
In this paper, we will study the Virasoro constraints in the following directions:
\begin{enumerate}
	\item {\bf Ancestor Virasoro conjecture}: We propose an ancestor version of the Virasoro conjecture for any homogeneous CohFT with vacuum. Compared to the original (descendent) Virasoro conjecture, we believe that the ancestor version is more general and should   hold true for the entire set of Virasoro operators.
	 We  prove  its genus-0 part and deduce simplified equations for the genus-1 part. 
	 We also prove the ancestor  Virasoro constraints for semisimple cases,  essentially following the work of Givental and Teleman. 
	\item {\bf Generalized Virasoro conjecture}: We generalize the (descendent) Virasoro conjecture to the cases of homogeneous CohFTs calibrated by a homogeneous $S$-matrix and $\dvac$-vector. Such generalized Virasoro constraints has been studied for specific models in the literature, but the precise form of the Virasoro conjecture is not established in general. 
	We prove genus-0 part of the conjecture, generalizing works of Liu--Tian~\cite{LT98} and Givental~\cite{Giv04}. We also deduce simplified equations for genus-1 and prove the conjecture for semisimple cases by establishing its relation with the ancestor Virasoro conjecture.
\end{enumerate}

\subsection{Ancestor Virasoro conjecture}
To describe the ancestor Virasoro constraints for CohFTs, we briefly introduce several notations. The details can be found in \S \ref{sec:CohFT}.

\begin{itemize}
	\item The CohFT $\Omega$ is defined over a finite dimensional vector space, called the state space.
	\item The state space  $\Frob$ is equipped  with a non-degenerate (super-) symmetric bilinear form $\eta$. Let \(\{\phi_{a}\}\) be a basis of \(\Frob\), \(\{\phi^a\}\) its dual w.r.t. \(\eta\), and define \(\eta^{ab}=\eta(\phi^b,\phi^a)\). 
	\item The genus zero part of the CohFT defines a family of  associative (super-) commutative products  $*_{\tau}$  over $\Frob$, called  the {\it quantum product},  for $\tau=\sum_a \tau^a \phi_a \in \Frob$.
	\item The vector field \({\bf 1}\) of unities of the family of quantum products is called the {\it flat unit} when it is covariantly flat with respect to \(\eta\).
\end{itemize}

In arbitrary CohFTs, the flat unit is extended to the vacuum vector field \(\vac(z)\), satisfying \(\vac(z) = {\bf 1}\) when it is flat. The unit axiom introduced by Kontsevich--Manin can be generalized to the vacuum axiom (Definition~\ref{vacuumaxiom}).
Throughout this paper, we focus on CohFTs that satisfy the vacuum axiom and refer to them as CohFTs with vacuum\footnote{For semisimple CohFTs, Teleman~\cite{Tel12} has proved that the vacuum axiom always holds for classes with insertions.}.

Furthermore, we introduce an additional mathematical structure for CohFTs, namely the homogeneity property  (Definition~\ref{homogeneouscohft}). This property is defined in terms of two key concepts: the conformal dimension \(\delta\) and the Euler vector field \(E=\sum_{a}(1 - d_a)\tau^a\partial_{\tau^a}+\sum_{a}r^a\partial_{\tau^a}\), where \(r^a\) are constants. Based on these concepts, we can define the grading operator \(\mu\in\End(\Frob)\) such that \(\mu(\phi_a)=(d_a-\frac{\delta}{2})\phi_a\). For further details, we direct the reader to \S \ref{sec:hom-condition}.

We consider the generating series \(\cA^{\tau}({\bf s};\hbar)\) of a given CohFT, with formal variables \({\bf s}(z)=\sum s^a_k\phi_a z^k\in\Frob[[z]]\), and this generating series is known as the total ancestor potential.
Subsequently, we introduce a set of quadratic differential operators \(L^{\anc}_{m}\), which can act on \(\cA^{\tau}({\bf s};\hbar)\).
\begin{definition}
The ancestor Virasoro operators $L^{\anc}_{m}$, $m\geq -1$, are defined as follows:
\begin{align}
		L^{\anc}_m=&\, \frac{1}{2\hbar^2}\eta(\E^{m+1}s_0,  s_0)
		-\frac{1}{4} \sum_{i+j=m} \tr(\E^i(\mu+\tfrac{1}{2})\E^j(\mu-\tfrac{1}{2}))
		+\sum_{k \geq 0}\sum_{l=0}^{m+k}\sum_{a,b=0}^{N-1}(C^{\tau}_{m})_{k-1,a}^{l,b} \tilde s_k^a  \frac{\pd }{\pd s_{l}^{b}}    \nonumber \\
		&\, +\frac{\hbar^2}{2}\sum_{k=0}^{m-1}\sum_{l=0}^{m-k-1}\sum_{a,b,c=0}^{N-1}(-1)^{k+1}(C^{\tau}_{m})_{-k-2,b}^{l,c}\eta^{ab}\frac{\pd^2 }{\pd s_{k}^{a}\pd s_{l}^{c}},\label{def:Lmtau}
\end{align}
where $\E=E*_{\tau}$, $\tilde {\bf s}(z)={\bf s}(z)-z\vac^{\tau}(z)$, 
and for $m\geq -1$, $k,l\in \mathbb Z$, $a,b\in\{0,\cdots,N-1\}$, terms $(C^{\tau}_{m})_{k,a}^{l,b}$ are polynomials of entries of $\E$ defined by
$$\textstyle
\big(\E+  ( \mu+\tfrac{3}{2})z\!+z^2\partial_z\big)^{m+1}\phi_a z^k
=\sum_{l=k}^{k+m+1}\sum_{b=0}^{N-1}(C^{\tau}_{m})_{k,a}^{l,b}\phi_bz^{l}.
$$
\end{definition}
It's easy to verify these operators satisfy the Virasoro commutation relation:
$$
[L^{\tau}_m,L^{\anc}_n]=(m-n) L^{\anc}_{m+n},\qquad  m,n \geq -1.
$$
\begin{manualconjecture}{1}[Ancestor Virasoro conjecture]
For any homogeneous CohFT with vacuum, 
its total ancestor potential $\cA^{\tau}({\bf s};\hbar)$   satisfies the following ancestor Virasoro constraints:
$$
	L^{\anc}_m\cA^{\tau}({\bf s};\hbar)=0,\qquad m\geq -1.
$$
\end{manualconjecture}

\begin{remark}
 {More precisely, we can just require that the vacuum axiom holds for classes with insertions. It is straightforward to observe that the ancestor Virasoro constraints are equivalent for two CohFTs that differ only by classes without insertions. In particular, for semisimple cases, we do not need to assume the vacuum axiom.}
\end{remark}

The ancestor Virasoro constraints admit a genus expansion form as follows:
\beq\label{def:Lgm-A}\textstyle
(L^{\anc}_m\cA^{\tau}({\bf s};\hbar))/\cA^{\tau}({\bf s};\hbar)=\sum_{g\geq 0}\hbar^{2g-2}\mathscr L^{\anc}_{g,m}({\bf s}).
\eeq
We refer to \(\mathscr{L}^{\anc}_{g,m}({\bf t}) = 0\) as the \textit{genus-\(g\) \(L^{\anc}_m\)-constraint}. For fixed \(m\) and all \(g\), they're called the \(\textit{L}^{\anc}_m\)-constraint; for fixed \(g\) and all \(m\), they're called the \textit{genus-\(g\) ancestor Virasoro constraints}. 
We will see that the $L^{\anc}_{-1}$-constraint follows from the vacuum axiom (Proposition~\ref{prop:string-dilaton}).
Our main result is the following theorem:

\begin{manualtheorem}{1}[=Theorem~\ref{thm:anc-vira}]\label{thm:vira-anc-intro}
	(1). The genus-0 ancestor Virasoro conjecture always holds.

(2). For each $m\geq 0$, the genus-1 $L_m^{\anc}$-constraint is equivalent to the following equation: 
	\beq\label{eqn:vira-genusone-0}
	\<E^{m+1}\>^{\tau}_{1,1}
	=-\frac{1}{4}\sum_{i+j=m}\str(\E^i\mu \E^j\mu)
	-\frac{1}{24}\sum_{i+j=m}\str((\E^i\mu\E^j{\bf 1})*_{\tau}) 
+\frac{1}{24}\str((\E^{m}(\mu+\tfrac{\delta}{2}){\bf 1})*_{\tau}),
	\eeq
where `$\str$' stands for the super-trace defined by $\str(A)=\sum_a\eta(\phi_a,A\phi^a)$ for $A\in\End(\Frob)$.

	(3). The ancestor Virasoro constraints hold for semisimple homogeneous CohFTs.
\end{manualtheorem}

\subsection{Generalized Virasoro conjecture}
The {\it generalized Virasoro conjecture}  extends the original Virasoro conjecture for the descendent GW theory to an arbitrary  CohFT with vacuum. To define  the (formal) {\it descendent potential} (Definition~\ref{def:formaldes})
of a given CohFT, we need to pick an  $S$-calibration and $\dvac$-calibration   (see \S \ref{sec:calibration}).
For a specific geometric theory, we expect that there always exist certain calibrations such that the formal  descendent potential  introduced   matches  the   naturally defined geometric one, as it does in GW theory.

Then we select the  $S$- and $\dvac$-calibration that meet a specific homogeneity condition,  as  the CohFT does. The homogeneity condition is defined by   an additional operator \(\rho\in\End(\Frob)\), which serves as the descendent counterpart of \(\E\) \footnote{In general, the operator \(\rho\) should be generalized to an operator-valued polynomial \(\rho(z)\in \End(\Frob)[z^{-1}]\). All the results for \(\rho\) have a parallel version for \(\rho(z)\). See Section \ref{subsec:hom-des} for details.}.

We introduce the generalized Virasoro constraints as follows. 
Recall the flat unit ${\bf 1}$ is involved in the standard Virasoro operators via the dilaton shift $\tilde {\bf t}(z)={\bf t}(z)-{\bf 1}z$ (\cite{EHX97,EJX98}).
For an arbitrary CohFT, we generalize the dilaton shift by $\tilde{\bf t}(z)={\bf t}(z)-\tau_0-z\dil(z)$, 
where $\tau_0$ is a point on $\Frob$ and $\dil(z)\in \Frob[z]$ is a constant vector valued polynomial in $z$ \footnote{It is determined by the $S$-calibration and $\dvac$-calibration, see equation~\eqref{eqn:dvac-u}}.

\begin{definition}
The Virasoro operators $L_{m}$, $m\geq -1$, are defined as follows: 
	\begin{align}
		L_m=&\, \frac{1}{2\hbar^2}\, \eta(\rho^{m+1} \tilde t_0, \tilde t_0)
		-\delta_{m,0}\cdot \frac{1}{4}\, \str\Big(\mu^2-\frac{1}{4}\Big)
		+\sum_{k \geq 0}\sum_{l=0}^{m+k}\sum_{a,b=0}^{N-1}  (C_{m})_{k-1,a}^{l,b} \tilde t_k^a \frac{\pd }{\pd t_{l}^{b}}    \nonumber \\
		&\,+\frac{\hbar^2}{2}\sum_{k=0}^{m-1}\sum_{l=0}^{m-k-1}\sum_{a,b,c=0}^{N-1}(-1)^{k+1}(C_{m})_{-k-2,b}^{l,c}\eta^{ab}\frac{\pd^2 }{\pd t_{k}^{a}\pd t_{l}^{c}},\label{def:Lm}
	\end{align}
	where for $m\geq -1$, $k,l\in \mathbb Z$, $a,b\in\{0,\cdots,N-1\}$, the $(C_{m})_{k,a}^{l,b}$ are constants defined by
	$$
	\big(\rho+  ( \mu+\tfrac{3}{2})z\!+z^2\partial_z\big)^{m+1}\phi_a z^k
	=\textstyle \sum_{l=k}^{k+m+1}\sum_{b=0}^{N-1}(C_{m})_{k,a}^{l,b}\phi_bz^{l}.
	$$
\end{definition}
It's easy to verify these operators satisfy the Virasoro commutation relation:
$$
[L_m,L_n]=(m-n) L_{m+n} ,\qquad   m,n \geq -1.
$$

The following two constants $m_\nu$ and $c_\nu$ are needed in the generalized Virasoro conjecture:
\begin{definition} \label{defnofm0}
	For an \(S\)- and \(\dvac\)-calibrated homogeneous CohFT, we say it has the {\it Virasoro-index} \(\virm \in \mathbb{Z}\) and  {\it Virasoro-constant} \(\cdes\)   under the following conditions:
	\begin{enumerate}
		\item If the $\dvac$-vector \(\dvac^{\tau}(z)\) is a polynomial in \(z\), then we define \(\virm = - 1\) and \(\cdes = 0\).
		\item If the $\dvac$-vector \(\dvac^{\tau}(z)\) is not a polynomial in $z$, but the conformal dimension \(\delta\) satisfies \(\frac{\delta - 3}{2}\in \mathbb{Z}_{\geq 0}\), then we define \(\virm=\frac{\delta - 3}{2}\) and $\cdes$ as the   coefficient of $ z^{1 - \delta} $ in	\begin{equation}\label{def:cm}
		 \tfrac{1}{2}\cdot (-1)^{\virm}\cdot(\virm!)^2\cdot\eta(E,\dvac^{\tau}(z)),
		\end{equation}
		We will prove in Lemma \ref{lem:cm}  that \(\cdes\) is a constant (independent of $\tau$).
	\end{enumerate}
	 
\end{definition}

\begin{manualconjecture}{2}[Generalized Virasoro conjecture]\label{conj:des-vira}
When calibrated by a homogeneous $S$-matrix and $\dvac$-vector, any homogeneous CohFT with the Virasoro-index $\virm$ has the total descendent potential $\cD({\bf t};\hbar)$ that satisfies the following generalized Virasoro constraints:
\beq\label{eqn:vira-D}
\big(L_m+\tfrac{\delta_{m,2 \virm}}{ \hbar^{2}}\cdot\cdes\big)\cD({\bf t};\hbar)=0, \qquad  m\geq \virm.
\eeq
\end{manualconjecture}

Similarly to the ancestor version, the generalized Virasoro constraints also admit a genus expansion form:
\beq\label{def:Lgm-D}\textstyle
(L_m\cD({\bf t};\hbar))/\cD({\bf t};\hbar)=\sum_{g\geq 0}\hbar^{2g-2}\mathscr L_{g,m}({\bf t}).
\eeq
We have the concepts the {\it genus-$g$ $L_m$-constraint}, the {\it genus-$g$ (generalized) Virasoro constraints} and {\it $L_{m}$-constraint} in parallel.

For $m\geq \virm$, we will prove that the $L_m$-constraint is equivalent to the $L_m^{\anc}$-constraint in Proposition~\ref{prop:anc-des-vira}.
In particular, if $\virm=-1$, the $L_{-1}$-constraint always holds.
We have the following results for the generalized Virasoro conjecture.
\begin{manualtheorem}{2}[$=$Theorem~\ref{thm:vira-des-0} $+$ Theorem~\ref{thm:vira-des-g} ]\label{thm:vira-des-intro}
(1). For any homogeneous CohFT with homogeneous $S$- and $\dvac$-calibrations and a Virasoro index $\virm$, the genus-$0$ generalized Virasoro constraints $\mathscr L_{0,m}({\bf t}) + {\delta_{m,2 \virm}} \! \cdot   c_\nu =0$ hold for $m \geq \virm$. 

(2). For each $m\geq \max\{\virm,0\}$, the genus-1 $L_m$-constraint is equivalent to equation~\eqref{eqn:vira-genusone-0}.

(3). The full generalized Virasoro conjecture holds for any semisimple homogeneous CohFT with homogeneous $S$- and $\dvac$-calibrations and a Virasoro index $\virm$.
\end{manualtheorem}
\begin{remark}
When the unit ${\bf 1}$ is flat, we have $\mu({\bf 1})=-\frac{\delta}{2}\cdot {\bf 1}$, and the simplified genus-1 Virasoro constraints (equation~\eqref{eqn:vira-genusone-0}) are equivalent to the ones in~\cite[Theorem 4.4]{Liu01} (for GW theories) and~\cite[\S 3.10.7]{DZ01} (for semisimple cases).
\end{remark}

\subsection{Applications}
We present two applications of our main theorems.

First, we investigate an enumerative geometric theory named the $\epsilon$-deformed negative \(r\)-spin theory, whose descendent potential is denoted by \(\cD^{r,\epsilon}({\bf t};\hbar)\)   (see \S \ref{sec:negative-r-spin} for the definition). This theory has the Virasoro index $m_\dvac=0$. By directly applying Theorem~\ref{thm:vira-des-intro}, we derive the Virasoro constraints for this theory.
\begin{manualproposition}{1}[$=$Proposition~\ref{prop:vira-negative-rspin}]
	\label{prop:vira-negative-r}
	The generalized Virasoro conjecture holds for the deformed negative r-spin theory. 
Namely, $L^{r,\epsilon}_m\mathcal D^{r,\epsilon}({\bf t};\hbar)=0$ for $m\geq  0$.
\end{manualproposition}

Second, we identify a two-dimensional semisimple homogeneous CohFT $\Omega^{\eGdd}$, calibrated with a specific \(S\)-matrix and \(\dvac\)-vector.   Its descendent correlators extend the counting invariants of Grothendieck's dessins d'enfants, a topic that is widely studied in the literature \cite{KZ15, Zhou19}.
To be precise, let \(\cD^{\eGdd}(\epsilon_1,\epsilon_2,{\bf t};\hbar)\)  denote the total descendent potential of the calibrated $\Omega^{\eGdd}$,  and let $Z^{\Gdd}(u,v,{\bf p};\hbar)$ denote the exponential of the generating function of Grothendieck's dessins d'enfants.
Using the generalized Virasoro constraints,   we prove:
\begin{manualproposition} {2}[$=$Proposition~\ref{prop:Gdd-eGdd}]
By taking $t^1_k=0$ and $t_k^0=k!\, p_{k+1}$ for $k\geq 0$,
we have
$$
\cD^{\eGdd}(u,v,{\bf t};\hbar)/\cD^{\eGdd}(u,v, {\bf 0};\hbar)=Z^{\Gdd}(u,v,{\bf p};\hbar).
$$
\end{manualproposition}
Since \(\cD^{\eGdd}(\epsilon_1,\epsilon_2,{\bf t};\hbar)\) encompasses the generating series of Grothendieck's dessins d'enfants as a sub-series, we term \(\Omega^{\eGdd}\) the  CohFT  of extended Grothendieck's dessins d'enfants. 
 \subsection{Plan of the paper} This paper is organized as follows. 
We first introduce the CohFTs with vacuum in \S \ref{sec:CohFT} and then define their descendent potentials in \S \ref{sec:CohFT-des}.
We prove our main results on Virasoro constraints in \S \ref{sec:vira-proof} and give two applications of the generalized Virasoro constraints in \S \ref{sec:negative-r-spin} and \S \ref{sec:vira-dessin} respectively.

\textbf{Acknowledgments.} 
The authors want to thank Ce Ji for reading this paper carefully.
The second author would like to thank Jian Zhou and Di Yang for their encouragements and valuable discussions.
The work was supported in part by National Key Research and Development Program of China  No. 2023YFA1009802
and NSFC 12225101.

\section{CohFT with vacuum}
\label{sec:CohFT}
In this section, we begin by reviewing the preliminaries of GW  theories. Subsequently, we extend the concepts from GW theory to  CohFTs  with   vacuum and establish the necessary properties.
Next, we consider the generalized Frobenius manifold that naturally arises from these CohFTs. We introduce the homogeneity condition for both the Frobenius manifold and the underlying CohFT.
In the end, we recall the Givental–Teleman reconstruction theorem for semisimple homogeneous CohFTs with a vacuum.

\subsection{Gromov--Witten theory}
Let $X$ be a non-singular projective variety of complex dimension $\delta$, 
we call the cohomology $H^*(X,\mathbb C)$ of $X$ the {\it state space} and denote it by $\Frob$.
There is a natural bilinear form $\eta$ on $\Frob$ defined by the Poincar\'e pairing,
and we fix a homogeneous basis $\{\phi_a\}_{a=0}^{N-1}$ of $\Frob$ and denote by $\{\phi^a\}_{a=0}^{N-1}$ its dual basis with respect to $\eta$, i.e., $\eta(\phi^a,\phi_b)=\delta^a_b$. 
In the follows, we also call $\Frob$ the {\it small phase space} and the space of $\Frob$-valued power series $\Frob[[z]]$ the {\it big phase space}.

Let $\Mbar_{g,n}(X,\beta)$ be the moduli space of  degree $\beta\in H_2(X,\mathbb Z)$ stable maps from a genus $g$ curve with $n$ marked points to the target $X$.
It was proved~\cite{LiT98, BF97} that this moduli space is compact and can be equipped with a virtual fundamental class $[\Mbar_{g,n}(X,\beta)]^{\rm vir}$ of complex dimension $(3-\delta)(g-1)+n+\int_{\beta}c_1(X)$.

The genus-$g$ descendent potential $\cF_g({\bf t})$ of the GW theory of $X$ is defined by
$$
\cF_g({\bf t}):=\sum_{n\geq 0,\beta\in H_2(X,\mathbb Z)}\frac{Q^{\beta}}{n!}
\int_{[\Mbar_{g,n}(X,\beta)]^{\rm vir}}\prod_{i=1}^{n}\bigg(\sum_{k_i\geq 0}\ev_i^{*}(t_{k_i})\psi_i^{k_j}\bigg),
$$
where $t_k=t_{k}^a\phi_a\in\Frob$, $\psi_i$ is the $1$-st Chern class of the universal cotangent line bundle over $\Mbar_{g,n}(X,\beta)$ corresponding to the $i$-th marked point, $\ev_i^*$ is the pull-back by the evaluation map $\ev_i:\Mbar_{g,n}(X,\beta)\to X$ at $i$-th marked point
and $Q^{\beta}=\prod_i Q_{i}^{d_i}$ is a monomial in the Novikov ring $\mathbb A:=\mathbb C[[Q ]]$ with $\beta$ having expansion $\sum_{i=1}^{b} d_i\beta_i$ in the basis $\{\beta_i\}_{i=1}^{b}$ of $H_2(X,\mathbb Z)$.
The total descendent potential $\cD({\bf t};\hbar)$ is defined by
$$
\cD({\bf t};\hbar):=e^{\sum_{g\geq 0}\hbar^{2g-2}\cF_g({\bf t})}.
$$

Let $\Mbar_{g,n}$ be the  moduli space of stable curves of genus $g$ with $n$ marked points. Consider the forgetful map 
$f: \Mbar_{g,n+m}(X,\beta)\to \Mbar_{g,n+m}\to \Mbar_{g,n}$ and
denote by $\bar\psi_i:=f^{*}(\psi_i)$   the pull back of the class $\psi_i$ on $\Mbar_{g,n}$. The genus-$g$ ancestor potential $\bar\cF^{\tau}_g({\bf s})$ of the GW theory of $X$ is defined by:
	\beq\label{def:anc-Fg-GW}
	\bar\cF^{\tau}_g({\bf s})
	:=\sum_{n,m\geq 0,\beta\in H_2(X)}\frac{Q^{\beta}}{n!\, m!}\int_{[\Mbar_{g,n+m}(X,\beta)]^{\rm vir}}\prod_{i=1}^{n}\bigg(\sum_{k_i\geq 0}\ev_i^{*}(s_{k_i})\bar\psi_i^{k_i}\bigg)\prod_{i=n+1}^{n+m}\ev_i^{*}(\tau),
	\eeq
	where $\tau=\tau^a\phi_a\in\Frob$ and $s_k=s_k^a\phi_a\in\Frob$. 
The total ancestor potential $\cA^{\tau}({\bf s};\hbar)$ is defined by
$$
\cA^{\tau}({\bf s};\hbar):=e^{\sum_{g\geq 0}\hbar^{2g-2}\bar\cF^{\tau}_g({\bf s})}.
$$

Following Givental~\cite{Giv01a}, introduce the $S$-matrix $S^{\tau}(z)\in \End(\Frob)[[z^{-1}]]$ defined by:
$$
\eta(\phi_a,S^{\tau}(z)\phi_b)=\eta(\phi_a,\phi_b)+\sum_{k\geq 0}\frac{\pd^2 \cF_0({\bf t})}{\pd t_0^a\pd t_k^b}\Big|_{\bf t=\tau}z^{-k-1}.
$$
where by ${\bf t=\tau}$,  we mean $t_0=\tau$ and $t_k=0$ for $k\geq 1$.
The $S$-matrix $S^{\tau}(z)$ satisfies the symplectic condition $S^{\tau,*}(-z)S^{\tau}(z)=\id$ 
and the quantum differential equation (QDE): 
$$
z \pd_{\tau^a}S^{\tau}(z)= \phi_a*_{\tau}S^{\tau}(z), \qquad a=0,\cdots,N-1.
$$
Here $S^{\tau,*}(z)$ is the dual of $S^{\tau}(z)$ with respect to $\eta$ and $*_{\tau}$ is the quantum product defined by
$$
\eta(\phi_a*_{\tau}\phi_b,\phi_c)=\frac{\pd^3 \cF_{0}(\bf t)}{\pd t_0^a\pd t_0^b\pd t_0^c}\bigg|_{\bf t=\tau}.
$$

According to Kontsevich--Manin~\cite{KM98} and Givental~\cite{Giv01a}, the total descendent potential $\cD({\bf t};\hbar)$ and the total ancestor potential $\cA^{\tau}({\bf s};\hbar)$ are related via the $S$-matrix by the following formula called the Kontsevich--Manin formula:
\beq\label{eqn:DA-correspondence}
\cD({\bf t};\hbar)=e^{F_1(\tau)+\frac{1}{2\hbar^2}W^{\tau}(\tilde{\bf t},\tilde{\bf t})}\cA^{\tau}({\bf s(t)};\hbar).
\eeq
Here $F_1(\tau)=\cF_1({\bf t})|_{\bf t=\tau}$, 
$\tilde{\bf t}$ is defined by $\tilde t_k^a=t_k^a-\delta_{k,1}\delta_{a,1}$ and is called the {\it dilaton shift},
$W^{\tau}(\tilde{\bf t},\tilde{\bf t})=\sum_{k,l\geq 0}\eta(\tilde{t}_k,W^{\tau}_{k,l}\tilde{t}_{l})$ is a quadratic function defined by
\beq\label{eqn:W-S}
W^{\tau}(z,w)=\sum_{k,l\geq 0}W^{\tau}_{k,l}z^{-k}w^{-l}=\frac{S^{\tau,*}(z)S^{\tau}(w)-\id}{z^{-1}+w^{-1}};
\eeq 
and the coordinate transformation ${\bf s}={\bf s(t)}$ is given by 
\beq\label{eqn:s-t}
{\bf s}(z)=[S^{\tau}(z){\bf t}(z)]_{+}-\tau,
\eeq 
where $[S^{\tau}(z){\bf t}(z)]_{+}$ stands for the part of $S^{\tau}(z){\bf t}(z)$ that contains only non-negative powers of $z$.

\subsection{Cohomological field theory}\label{subsec:CohFT}
Let $\varphi \in \Mbar_{g,n}(X,\beta)$ be a stable map from a genus $g$ curve $\Sigma_g$ with $n$ distinguished marked points  $x_1,\cdots,x_n$ to $X$ such that $[\varphi(\Sigma_g)]=\beta$, we consider two maps $\ev:\Mbar_{g,n}(X,\beta)\to X^n$ and $p: \Mbar_{g,n}(X,\beta)\to\Mbar_{g,n}$ defined by $\ev(\varphi)=(\varphi(x_1),\cdots,\varphi(x_n))$ and $p(\varphi)=\widetilde \Sigma_g$ (the stable curve defined by contracting the non-stable components of $\Sigma_g$) respectively.
The GW theory induces the GW class $I_{g, n}:\Frob ^{\otimes n}  \rightarrow   \mathbb Q[[Q]] \otimes H^*(\Mbar_{g,n},\mathbb{Q})$ defined by 
	$$\textstyle
	I_{g,n}(v_1\otimes\cdots \otimes v_n):=\sum_{\beta}Q^{\beta}\cdot p_{*}\big( \ev^{*}(v_1\otimes\cdots \otimes v_n)\big).
	$$

Kontsevich and Manin~\cite{KM94} introduced the CohFT to capture the axiomatic properties of the GW class $I_{g,n}$. 
Let  $\Frob$ be a complex vector space of dimension $N$  with a  non-degenerate (super-) symmetric bilinear form $\eta$.
Let $\mathbb A$ be an $\mathbb C$-algebra, a CohFT $\Omega=\{\Omega_{g,n}\}_{2g-2+n>0}$ on $(\Frob, \eta)$ is defined to be a set of maps to the cohomological classes of $\Mbar_{g,n}$
$$
\Omega_{g, n} : \Frob ^{\otimes n}  \rightarrow   \mathbb A \otimes H^*(\Mbar_{g,n},\mathbb{Q})
$$
satisfying (1) {\it  $S_n$-invariance axiom}: $\Omega_{g,n}(v_1\otimes \cdots \otimes v_n)$ is invariant under permutation between $v_1, \ldots, v_n\in \Frob$ 
and (2) {\it gluing axiom}: the pull-backs $q^*\Omega_{g,n}$ and $r^*\Omega_{g,n}$ of the gluing maps
\begin{equation*}
\begin{aligned}
q&:&&\Mbar_{g-1, n+2}\to \Mbar_{g,n}\\
r&:&& \Mbar_{g_1, n_1+1}\times \Mbar_{g_2, n_2+1} \to \Mbar_{g_1+g_2, n_1+n_2}
\end{aligned}
\end{equation*}
are equal to the contraction of $\Omega_{g-1, n+2}$ and $\Omega_{g_1, n_1+1} \otimes \Omega_{g_2, n_2+1}$ by the bivector $\sum_a \phi_a\otimes \phi^a$. Here $\{\phi_a\}_{a=0}^{N-1}$ is   a flat basis of $H$, { and $\{\phi^a\}_{a=0}^{N-1}$ is its dual basis with respect to $\eta$}.
According to the $S_n$-invariance axiom, we also denote $\Omega_{g,n}(v_1\otimes\cdots \otimes v_n)$ by $\Omega_{g,n}(v_1,\cdots , v_n)$.

Let $\tau$ be a point in the neighborhood $U$ of ${\bf 0}\in\Frob$, we define 
	\beq\label{def:shifted-Omega}
	\Omega_{g, n}^{\tau}(-):=\sum_{m\geq 0}\frac{1}{m!}(\pi^{m})_{*}\Omega_{g,n+m}(-,\tau,\cdots,\tau) \in \mathbb A[[\tau]] \otimes H^*(\Mbar_{g,n},\mathbb{Q}),
	\eeq
	where $(\pi^{m})_{*}$ is the push-forward via forgetful map $\pi^m:\Mbar_{g,n+m}\to \Mbar_{g,n}$ forgetting the last $m$ markings.
It is proved~\cite[Proposition 7.1]{Tel12} that $\Omega^{\tau}$ gives another (formal) CohFT, called the shifted CohFT,  on $\Frob$ with the bilinear form $\eta$ remains unchanged  and with $\mathbb C$-algebraic $\mathbb A^{\tau} := \mathbb A[[\tau]]$.

Given a CohFT $\Omega$, the quantum product is defined by $\eta(\phi_a*\phi_b,\phi_c):=\Omega_{0,3}(\phi_a,\phi_b,\phi_c)$. 
We denote by $*_{\tau}$ the quantum product for shifted CohFT $\Omega^{\tau}$.
The commutativity and the associativity of the quantum product is given by the $S_n$-invariance axiom and the gluing axiom  respectively. 

Introduce the ancestor correlators $\<-\>^{\tau}_{g,n}$, $2g-2+n>0$, for the CohFT $\Omega^{\tau}$:
\beq\label{def:bracket-anc}
\<v_1\bar\psi^{k_1},\cdots,v_n\bar\psi^{k_n}\>^{\tau}_{g,n}
:=\int_{\Mbar_{g,n}}\Omega^{\tau}_{g,n}(v_1, \cdots , v_n)\psi_1^{k_1}\cdots\psi_n^{k_n}.
\eeq
The genus $g$ ancestor potential $\bar\cF^{\tau}_g({\bf s})$ of the CohFT $\Omega^{\tau}$ is defined by
\beq\label{def:anc-Fg-choft}
\bar\cF^{\tau}_g({\bf s})
:=\sum_{n\geq 0}\frac{1}{n!}\<{\bf s}(\bar\psi_1),\cdots,{\bf s}(\bar\psi_{n})\>^{\tau}_{g,n}.
\eeq
where ${\bf s}(z)=\sum_{k\geq 0}s_kz^k\in\Frob[[z]]$. 
The total ancestor potential $\cA^{\tau}({\bf t};\hbar)$ is defined by
\beq\label{def:barF-A}
\cA^{\tau}({\bf s};\hbar):=e^{\sum_{g\geq 0}\hbar^{2g-2}\bar\cF^{\tau}_g({\bf s})}.
\eeq
For the shifted GW class $I_{g,n}^{\tau}$ defined by equation~\eqref{def:shifted-Omega} from the GW class $I_{g,n}$, 
one can see the definition equation~\eqref{def:anc-Fg-GW} of $\bar\cF^{\tau}_g$ is equivalent to equation \eqref{def:anc-Fg-choft}.

\begin{proposition}\label{prop:diff-anc-corr}
The ancestor correlators $\<-\>^{\tau}_{g,n}$, $2g-2+n>0$, satisfy the following differential equations~\footnote{For simplicity, we assume all the insertions here are even classes, when there are odd classes, the equation~\eqref{eqn:diff-anc-corr} should be written as follows:
$$\textstyle
\pd_{\tau^b}(\<\phi_{a_{[n]}}\bar\psi^{k_{[n]}}\>^{\tau}_{g,n})
=\<\phi_b,\phi_{a_{[n]}}\bar\psi^{k_{[n]}}\>^{\tau}_{g,n+1}
-\sum_{i=1}^{n}(-1)^{\sigma(a_i)\cdot \sum_{j=1}^{i-1}\sigma(a_j)}\cdot\<\phi_{b}*_{\tau}\phi_{a_i}\bar\psi^{k_i-1}, \phi_{a_{[n]\setminus\{i\}}}\bar\psi^{k_{[n]\setminus\{i\}}}\>^{\tau}_{g,n},
$$
where $\sigma(a)=0$ if $\phi_a$ is even and $1$ otherwise, and the proof is similar.}: for $b=0,\cdots,N-1$,
\beq\label{eqn:diff-anc-corr}
\pd_{\tau^b}(\<\phi_{a_{[n]}}\bar\psi^{k_{[n]}}\>^{\tau}_{g,n})
=\<\phi_b,\phi_{a_{[n]}}\bar\psi^{k_{[n]}}\>^{\tau}_{g,n+1}
-\sum_{i=1}^{n}\<\phi_{b}*_{\tau}\phi_{a_i}\bar\psi^{k_i-1}, \phi_{a_{[n]\setminus\{i\}}}\bar\psi^{k_{[n]\setminus\{i\}}}\>^{\tau}_{g,n},
\eeq
where $[n]=\{1,\cdots,n\}$ and for any $I\subset [n]$, $\phi_{a_I}\bar\psi^{k_{I}}=\otimes_{i\in I}\phi_{a_i}\bar\psi^{k_i}$.
\end{proposition}
\begin{proof}
	We first recall some properties of the $\psi$-classes under the pull-back of the forgetful map, one can see e.g.~\cite{Wit91} for details.
		Let $\psi_1,\cdots,\psi_{n}$ be the psi-classes of $ \M_{g,n+1}$  at the first $n$-marked points or  psi-classes of of $ \M_{g,n}$ (by abuse of notation)  and let $\pi_{n+1}:  \M_{g,n+1} \rightarrow  \M_{g,n}$ be the forgetful map which forgets the last marked point. For $k=1,\cdots,n$, we have
		$$
		\pi_{n+1}^* \psi_k = \psi_k  -  D_{k,n+1},
		$$
		where $D_{k,n+1}$ is the boundary divisor  representing the nodal curve which has two irreducible components: a genus $0$ curve with two marked points $x_k$ and $x_{n+1}$ connected with a genus $g$ curve with the rest $n-1$-marked points. 
		Moreover, the boundary divisors satisfy
		$$
		D_{i,n+1} \cdot D_{j,,n+1}=0  \text{   for $i\neq j$ } \quad   \text{   and }   \quad   D_{i,n+1}^2=-  \pi_{n+1}^* \psi_{i } \cdot D_{i,n+1}. $$
		It is straightforward to deduce the following equation: for $k=1,\cdots,n$ and $m\geq 0$,
		\beq\label{eqn:pi-psi}
		\pi_{n+1}^{*}\psi^{m}_{k}=\psi_{k}^m-\pi_{n+1}^{*}\psi^{m-1}_{k} \cdot D_{k,n+1}.
		\eeq
	
Now we return to the CohFT side. By equation~\eqref{eqn:pi-psi}, we have
	\begin{align*}
\Omega^{\tau}_{g,n+1} \big( \phi_b ,   \phi_{a_{[n]}}  \big) \cdot   \pi_{n+1}^* \big( \textstyle \prod_{i=1}^n   \psi_i^{k_{i}}  \big)
 =   & \,\Omega^{\tau}_{g,n+1} \big( \phi_b ,   \phi_{a_{[n]}}  \big) \cdot \textstyle \prod_{i=1}^n\psi_i^{k_{i}} \  - \\
	&\, \textstyle \Omega^{\tau}_{g,n+1} \big( \phi_b ,   \phi_{a_{[n]}}  \big) \cdot    \sum_{i=1}^n      D_{i,n+1}\cdot  \pi_{n+1}^*  \big(    \psi_i^{k_{i}-1} \textstyle \prod_{j\neq i}    \psi_j^{k_{j}}      \big).
	\end{align*}
	Hence for intersection numbers
	\begin{align*}
		& \  \int_{\M_{g,n}}  ( \pi_{n+1} )_* \Omega^{\tau}_{g,n+1} \big( \phi_b ,   \phi_{a_{[n]}}  \big) \cdot \prod_{i=1}^n   \psi_i^{k_{i}}
         = \int_{\M_{g,n+1}}  \Omega^{\tau}_{g,n+1} \big( \phi_b ,   \phi_{a_{[n]}}  \big)   \prod_{i=1}^n \psi_i^{k_{i}}  \  - \\
		& \  \qquad   \sum_{i=1}^n   {\Omega^{\tau}_{0,\{x_i,x_{n+1},{\bullet}\}}}  \big( \phi_{a_{i}} , \phi_b,  \phi^c \big)  \cdot  \int_{\M_{g,n}}  \Omega^{\tau}_{g,n} \big( \phi_c ,   \phi_{a_{[n]\setminus \{i\}}}  \big)   \psi_\bullet^{k_{i}-1}   \prod_{j\neq i}   \psi_j^{k_{j}} ,
	\end{align*}
where we have used $\pi_{n+1,*} (\alpha \cdot  \pi_{n+1}^*\beta) =  \pi_{n+1,*} (\alpha) \cdot \beta $ and the second  gluing axiom, and the $\bullet$ is used to trace the marking. 
	This is exactly equation~\eqref{eqn:diff-anc-corr}, and the proof is finished.
\end{proof}

\subsection{Vacuum axiom}
In GW theory, there is an distinguished element ${\bf 1}\in \Frob$ (indeed, {\bf 1} is exactly the generator of the space $H^{0}(X,\mathbb C)$ ) satisfying (1)
${\bf 1}$ is the unity of the quantum product $*_\tau$ and (2) ${\bf 1}$ is covariantly constant with respect to $\eta$. 
Axiomatically, an distinguished element ${\bf 1}\in \Frob$ is called a {\it flat unit} if it satisfies the following {\it flat unit axiom}: 
$$
\pi_{\bullet}^*\Omega^{\tau}_{g,n}(v_1, \cdots , v_n)=\Omega^{\tau}_{g,{\bullet}+n}({\bf 1}, v_1, \cdots , v_n),
$$
where $\pi_{\bullet}:\Mbar_{g,{\bullet}+n}\to \Mbar_{g,n}$ is the forgetful map.
Here we have used $\bullet$ in $\Omega^{\tau}_{g,{\bullet}+n}$, as well as in $\Mbar_{g,{\bullet}+n}$, to trace the marking point we are referring to.
It is shown in~\cite[\S 5.3.3]{Man99} the flat unit axiom implies ${\bf 1}$ is covariantly constant with respect to $\eta$.
In general, the flat unit does not necessarily exist and the flat unit axiom is generalized  to the vacuum axiom  by Teleman \cite{Tel12}:
\begin{definition}[Vacuum]   \label{vacuumaxiom}
A distinguished element $\vac^{\tau}(z)=\sum_{k\geq 0}\vac_k^{\tau}\cdot z^k\in\Frob[[z]]$ is called the {\it vacuum vector}, and $\Omega^{\tau}$ is called the CohFT with vacuum,  if it satisfies the following 
{\it vacuum axiom}:
\beq\label{eqn:vacuum-axiom}
\pi_{\bullet}^*\Omega^{\tau}_{g,n}(v_1, \cdots , v_n)=\Omega^{\tau}_{g,\bullet+n}(\vac^{\tau}(\psi_{\bullet}), v_1, \cdots , v_n),
\eeq
where $\pi_{\bullet}:\Mbar_{g,{\bullet}+n}\to \Mbar_{g,n}$ is the forgetful map.
\end{definition}

The following Proposition is a direct consequence of the vacuum axiom.
\begin{proposition}\label{prop:string-dilaton}
The total ancestor potential $\cA^{\tau}({\bf s};\hbar)$ of a CohFT with vacuum $\vac^{\tau}(z)$ satisfies the ancestor string equation (the $L^{\anc}_{-1}$-constraint) 
\beq\label{eqn:string}
\bigg(\sum_{k,a}\tilde s_{k+1}^a\frac{\pd}{\pd s_k^a}+\frac{1}{2\hbar^2}\, \eta(s_0,s_0)\bigg)\cA^{\tau}({\bf s};\hbar)=0,
\eeq
and the ancestor dilaton equation
\beq\label{eqn:dilaton}
	\bigg(\sum_{k,a}\tilde s^{a}_k\frac{\pd}{\pd s^a_k}+\hbar\frac{\pd}{\pd{\hbar}}+\frac{1}{24}\str({\bf 1}*) \bigg)\cA^{\tau}({\bf s};\hbar)=0,
\eeq
where $\tilde s_k^{a}$ is defined by $\tilde {\bf s}(z)={\bf s}(z)-z\vac^{\tau}(z)$.
\end{proposition}
\begin{proof}
Firstly, by coupling with the psi-classes and then taking integration on the space $\Mbar_{g,n}$, equation~\eqref{eqn:vacuum-axiom} gives immediately the ancestor string equation~\eqref{eqn:string}.
Secondly, together with the equation $\pi_{\bullet,*}\psi_{\bullet}=(2g-2+n)$ on $\Mbar_{g,n}$, the vacuum axiom gives
\beq\label{eqn:cohft-dilaton}
\pi_{\bullet,*}\Omega^{\tau}_{g,n+1}(\psi_{\bullet}\cdot \vac^{\tau}(\psi_{\bullet}), v_1, \cdots , v_n)
=(2g-2+n)\ \Omega^{\tau}_{g,n}(v_1, \cdots , v_n),
\eeq
where we have used $\pi_{\bullet, *}(\alpha\cdot \pi_{\bullet}^{*}\beta)=\pi_{\bullet,*}(\alpha)\cdot \beta$.
By dimension reason and the genus-1 topological recursion relation: $\<\bar\psi\cdot\vac(\bar\psi)\>^{\tau}_{1,1}=\<{\bf 1}\bar\psi\>^{\tau}_{1,1}=\frac{1}{24}\sum\eta(\phi_a,\phi^a)$.
Then, similar as how we prove the ancestor string equation, equation~\eqref{eqn:cohft-dilaton} gives us the ancestor dilaton equation~\eqref{eqn:dilaton}.
\end{proof}

Now we consider the differential equation for the vacuum vector, we prove the following Proposition generalizing Teleman's result~\cite[Proposition 7.3]{Tel12} for semisimple CohFTs to arbitrary cases.
\begin{proposition}\label{prop:QDE-vacuum}
The vacuum vector satisfies the following QDE: 
	\beq\label{eqn:QDE-vacuum}
	z\pd_{\tau^a}\vac^{\tau}(z)=\phi_a*_{\tau}\vac^{\tau}(z)-\phi_a,\qquad a=0,\cdots,N-1.
	\eeq
\end{proposition}
\begin{proof}
	Take derivative with respect to $\tau^a$ on equation~\eqref{eqn:vacuum-axiom} (without losing generality, we assume $v_i$ are flat vectors), we have
\begin{align*}
\pi_{\bullet}^*\pi_{n+1, *}\Omega^{\tau}_{g,n+1}(v_1, \cdots , v_n,\phi_{a})
=&\, \textstyle
 \sum_{k\geq 0}\psi_\bullet^k\cdot \pi_{n+1,*}\Omega^{\tau}_{g,\bullet+n+1}(\vac_k, v_1, \cdots , v_n, \phi_{a})\\
&\, \textstyle +\sum_{k\geq 0}\psi_\bullet^k\cdot \Omega^{\tau}_{g,\bullet+n}(\pd_{\tau^a}\vac_k, v_1, \cdots , v_n).
\end{align*}
By commuting $\pi_{\bullet}^*$ with $\pi_{n+1,*}$ and by Definition~\ref{vacuumaxiom}, the left-hand side of this equation equals to $\pi_{n+1, *}\Omega^{\tau}_{g,\bullet+n+1}(\vac(\psi_{\bullet}), v_1, \cdots , v_n,\phi_{a})$. 
Similar as the method used in the proof of Proposition~\ref{prop:diff-anc-corr}, by equation~\eqref{eqn:pi-psi}, we have
$$\textstyle
\sum_{k\geq 1}\psi_{\bullet}^{k-1}\cdot \Omega^{\tau}_{g,\bullet+n}(\phi_a*\vac_k, v_1, \cdots , v_n)
=\sum_{k\geq 0}\psi_{\bullet}^k\cdot \Omega^{\tau}_{g,\bullet+n}(\pd_{\tau^a}\vac_k, v_1, \cdots , v_n).
$$
Then the non-degeneration of $\Omega^{\tau}$ gives $\pd_{\tau^a}\vac_k=\phi_a*_{\tau}\vac_{k+1}$ for all $k\geq 0$.
\end{proof}
The QDE~\eqref{eqn:QDE-vacuum} for the vacuum vector gives immediately ${\bf 1}^{\tau}:=\vac^{\tau}_0$ is the unity of the quantum product $*_{\tau}$, and furthermore
\beq\label{eqn:formula-vacuum}
\vac^{\tau}(z)=\textstyle \sum_{k\geq 0}\pd_{\bf 1^{\tau}}^k({\bf 1^{\tau}})z^k.
\eeq
In particular, if the vacuum vector $\vac^{\tau}(z)$ is flat (i.e., covariantly constant with respect to $\eta$), then we see $\vac^{\tau}(z)={\bf 1}^{\tau}$ is a flat unit of the CohFT $\Omega^{\tau}$. 
For non-flat vacuum vector $\vac^{\tau}(z)$, the vector field ${\bf 1}$ of the unities ${\bf 1}^{\tau}$ of the quantum product $*_{\tau}$ is called the {\it non-flat unit}.

\subsection{Generalized Frobenius manifold}
We recall the notion of Frobenius manifold, which was introduced by  Dubrovin~\cite{Dub96} to capture the structure of genus-$0$ GW theory, and consider its generalizations.

Let $U$ be a complex manifold, a Frobenius structure over $U$ consists of 
\begin{enumerate}[\,\,\,\, (i)]
	\item a flat holomorphic bilinear form $\eta$ on the complex tangent space $T  U$; 
	\item a commutative and associative product $*_{\tau}: T _\tau U \times T _\tau U \rightarrow T _\tau U$, called the {\it quantum product}, which are holomorphic on $\tau\in U$;
	\item  a holomorphic vector field   ${\bf 1} $, such that ${\bf 1}^{\tau}$ is the unity of the product $*_\tau$.
\end{enumerate}
In the original definition of a Frobenius manifold, the vector field  ${\bf 1}$ of unities is required to be covariantly constant.
We will call such a structure a generalized Frobenius manifold with non-flat unit if this property fails (c.f.  ~\cite{LQZ22}).

According to the flatness of $\eta$, we have $T_{0}U \cong T_{\tau}U$, by abuse of the notation, we denote this by $\Frob$. 
Furthermore, the neighborhood of $0\in U$ can be identified with a neighborhood of $0\in \Frob$ (via the exponential map), thus we also view $\tau$ as a point in $\Frob$.

Furthermore, the structure constants $\eta(\phi_a*_{\tau}\phi_b,\phi_c)$ of the quantum product are given by the $3$-rd derivatives $\pd_{\tau^a}\pd_{\tau^b}\pd_{\tau^c}\Phi(\tau)$ of a holomorphic function $\Phi(\tau)$, which is called the {\it potential} of the (generalized) Frobenius manifold. 
We see the potential $\Phi$ is uniquely specified modulo quadratic terms.
The (super-) commutativity of the quantum product comes from the (super-) commutativity of partial derivatives of $\Phi(\tau)$ with respect to $\tau$ 
and the associativity is equivalent to the following Witten-Dijkgraaf-Verlinde-Verlinde (WDVV) equation: 
$$
\frac{\pd^3\Phi(\tau)}{\pd\tau^a\pd \tau^b \pd \tau^c}\eta^{cd}\frac{\pd^3 \Phi(\tau)}{\pd\tau^d\pd \tau^e \pd \tau^f}
=(-1)^{\sigma(b)\cdot \sigma(e)}\frac{\pd^3 \Phi(\tau)}{\pd\tau^a\pd \tau^e \pd \tau^c}\eta^{cd}\frac{\pd^3 \Phi(\tau)}{\pd\tau^d\pd \tau^b \pd \tau^f},
$$
where $\sigma(a)=0$ if $\phi_a$ is even and $1$ otherwise.

\medskip
	In \S \ref{subsec:CohFT},  we observe that
	given a CohFT $\Omega$, we can define the shifted CohFT $\Omega^\tau$ and thus the quantum product $*_\tau$, with  structure constants lying in  $\mathbb A[[\tau]]$. If we assume that the structure constants  converge  on $\tau$   and we can choose the set of certain analytic functions as our $\mathbb C$-algebra $\mathbb A^\tau$, then  we can see a generalized Frobenius structure arises naturally on a neighborhood  $U$ containing $\tau=0 $. In general, even without assuming the convergence property, we can still define a {\it formal Frobenius manifold} via the quantum product of the shifted CohFT. We refer the reader to \cite{LP04} for more details about the definition of a formal Frobenius manifold.

\subsection{Homogeneity conditions}
\label{sec:hom-condition}
A homogeneous Frobenius manifold is defined by the existence of a coordinate system $\{\tau^a\}_{a=0}^{N-1}$ with a flat basis $\{\phi_a\}_{a=0}^{N-1}$  and a  vector field $E$, known as the  {\it Euler vector field}, which
has  the following form~\footnote{This is equivalent to the form $E=\sum_a(1-d_a)\tau^a\phi_a+\sum_a r^a\phi_a$ with $r^a=-(1-d_a)\tau^a_0$ for $d_a\ne 1$.}
\beq\label{def:euler}
E=\sum_{a}(1-d_a)\tilde\tau^a\pd_{\tau^a}+\sum_{d_a=1}r^a\pd_{\tau^a},\qquad  \tilde\tau^a=\tau^a-\tau^a_0
\eeq 
for some constants $\tau^a_0$ and satisfies the following equations:
\begin{align}
E\big(\eta(v_1,v_2)\big)-\eta([E,v_1],v_2)-\eta(v_1,[E,v_2])=&\, (2-\delta)\eta(v_1,v_2), \label{eqn:euler-eta}\\
[E,v_1*_{\tau}v_2]-[E,v_1]*_{\tau}v_2-v_1*_{\tau}[E,v_2]=&\, v_1*_{\tau}v_2, \label{eqn:euler-qp}	
\end{align}
for some constant $\delta$, which is called the {\it conformal dimension} of the Frobenius manifold.
We always identify $\pd_{\tau^a}$ with $\phi_a$ on the Frobenius manifold.

Introduce the grading operator $\mu$:
\beq\label{eqn:mu-E}
\mu(v)=\big(1-\tfrac{\delta}{2}\big)v-\nabla_{v}E,
\eeq
where $\nabla$ is the Levi--Civita connection of $\eta$.  
Then equation~\eqref{eqn:euler-eta} is equivalent to $\mu^{*}=-\mu$ ($\mu^*$ is the adjoint operator of $\mu$ with respect to $\eta$)
and equation~\eqref{eqn:euler-qp} is equivalent to the following homogeneity condition for the potential $\Phi(\tau)$:
\beq\label{eqn:hom-F0}
E \Phi(\tau)=(3-\delta) \Phi(\tau)+ \text{quadratic}+ \text{linear}+ \text{constant}.
\eeq
Introduce the operator $\E:=E*_{\tau}$, then for a homogeneous Frobenius manifold, one can see 
\beq\label{eqn:diff-E}
\pd_{\tau^a}\E=\phi_{a}*_{\tau}+[\phi_a*_\tau,\mu].
\eeq

We define a linear operator $\deg$ on the cohomology ring $H^{*}(\Mbar_{g,n}, \mathbb{Q})$, 
which acts on a homogeneous class $\alpha \in H^{2k}(\Mbar_{g,n}, \mathbb{Q})$ by $\deg \alpha := k\cdot \alpha$,
where $k$ denotes the complex degree of the class $\alpha$.

\begin{definition}
\label{homogeneouscohft} 
A  CohFT $\Omega$ is called {\it homogeneous} with respect to $E$, if for $2g-2+n>0$,
\begin{align}
	&\, \textstyle \pi_{\bullet,*}\Omega_{g,\bullet+n}(E|_{\tau=0},v_1,\cdots,v_n)-\sum_{i=1}^{n}\Omega_{g,n}(v_1,\cdots,(\tfrac{\delta}{2}+\mu)v_i,\cdots,v_n)\nonumber\\
	=&\, \big((g-1)\delta-\deg\big)\Omega_{g,n}(v_1,\cdots,v_n). \label{def:hom-CohFT-0}
\end{align} 
	A shifted CohFT $\Omega^{\tau}$ is called {\it homogeneous} with respect to $E$, if for $2g-2+n>0$,
	\begin{align}
		&\, \textstyle E \big(\Omega^{\tau}_{g,n}(v_1,\cdots,v_n)\big) -\sum_{i=1}^{n}\Omega^{\tau}_{g,n}(v_1,\cdots,[E,v_i],\cdots,v_n)\nonumber\\
		=&\, \big((g-1)\delta+n-\deg\big)\Omega^{\tau}_{g,n}(v_1,\cdots,v_n).\label{def:hom-CohFT}
	\end{align}
\end{definition}
\begin{remark}
We see that $\Omega$ is homogeneous if   equation~\eqref{def:hom-CohFT} holds  at $\tau=0$. 
	Conversely, by the definition of $\Omega^{\tau}$, one can deduce equation~\eqref{def:hom-CohFT} from equation~\eqref{def:hom-CohFT-0}.
	Therefore, $\Omega^{\tau}$ is homogeneous if and only if $\Omega$ is homogeneous.
\end{remark}

Consider $\pi_\bullet^{*}$ acting on both sides of equation~\eqref{def:hom-CohFT}, then by using similar method that we used to deduce the QDE~\eqref{eqn:QDE-vacuum} for $\vac^{\tau}(z)$, we have for homogeneous CohFT the vacuum vector satisfies the following homogeneity condition:
\beq\label{eqn:hom-vacuum}
(z\pd_z+E)\vac^{\tau}(z)=-\big(\mu+\tfrac{\delta}{2}\big)\vac^{\tau}(z).
\eeq
Furthermore, by the QDE~\eqref{eqn:QDE-vacuum}, the homogeneity condition can be rewritten as follows:
\beq\label{eqn:hom-vacuum-2}
(z\pd_z+\tfrac{1}{z}\E+\mu+\tfrac{\delta}{2})\vac^{\tau}(z)-\tfrac{1}{z}E=0.
\eeq
Let $v(\tau,z)=(z\pd_z+\tfrac{1}{z}\E+\mu+\tfrac{\delta}{2})\vac^{\tau}(z)-\tfrac{1}{z}E$, then by equations~\eqref{eqn:QDE-vacuum} and \eqref{eqn:diff-E}, one can see
$$
\pd_{\tau^a}v(\tau,z)=\tfrac{1}{z}\phi_a*_{\tau}v(\tau,z).
$$
Therefore, the homogeneity condition~\eqref{eqn:hom-vacuum-2} for $\vac^{\tau}(z)$ holds if and only if it holds at $\tau=0$.

Notice that the non-vanished correlator $\<\phi_{a_1}\bar\psi^{k_1},\cdots,\phi_{a_n}\bar\psi^{k_n}\>_{g,n}^{\tau}$ 
is obtained by taking integration of degree $(3g-3+n-\sum_{i}k_i)$ part of  $\Omega^{\tau}_{g,n}(\phi_{a_1},\cdots,\phi_{a_n})$, coupling with psi-classes $\prod_i \psi_i^{k_i}$ on $\Mbar_{g,n}$, the homogeneity condition~\eqref{def:hom-CohFT} for $\Omega^{\tau}$ gives the  following homogeneity condition for ancestor correlators: 
$$\textstyle
E\big(\<\phi_{a_1}\bar\psi^{k_1},\cdots,\phi_{a_n}\bar\psi^{k_n}\>_{g,n}^{\tau}\big)
=\big((\delta-3)(g-1)-\sum_{i=1}^{n}(1-d_{a_i}-k_i)\big)\<\phi_{a_1}\bar\psi^{k_1},\cdots,\phi_{a_n}\bar\psi^{k_n}\>_{g,n}^{\tau}.
$$
This is equivalent to the following homogeneity condition for potentials $\bar\cF_{g}^{\tau}({\bf s})$:
\beq\label{eqn:hom-Fg-anc-E}\textstyle
	\big(E+\sum_{k,a}(1-d_a-k)s_k^{a}\frac{\pd}{\pd s_k^{a}}\big)\bar\cF_g^{\tau}({\bf s})
	=(\delta-3)(g-1)\bar\cF_g^{\tau}({\bf s}).
\eeq

We introduce the {\it ancestor Euler vector field} $\mathscr E^{\tau}$ on the big phase space:
\beq\label{eqn:euler-anc-big}\textstyle
\mathscr E^{\tau}=\sum_{k,a}(1-d_a-k)\tilde s_k^{a}\frac{\pd}{\pd s_k^a}-\sum_{k,a,b}\E_b^a\tilde s_{k+1}^{b}\frac{\pd }{\pd s_{k}^a}.
\eeq
\begin{proposition}\label{prop:hom-anc}
The ancestor potentials $\bar\cF_g^{\tau}({\bf s})$, $g\geq 0$, for a homogeneous CohFT with vacuum satisfy
\beq\label{eqn:hom-Fg-anc}
\mathscr E^{\tau}\bar\cF_{g}^{\tau}({\bf s}) =(\delta-3)(g-1)\bar\cF_g^{\tau}({\bf s})
+\tfrac{\delta_{g,0}}{2}\eta(\E s_0,s_0)+\delta_{g,1}\<E\>_{1,1}^{\tau}.
\eeq
Moreover, the correlator $\<E\>_{1,1}^{\tau}$ is a constant and we denote it by $\cc$.
\end{proposition}
\begin{proof}
To prove equation~\eqref{eqn:hom-Fg-anc}, we first notice that by the QDE~\eqref{eqn:QDE-vacuum} and homogeneity condition~\eqref{eqn:hom-vacuum} for the vacuum vector $\vac^{\tau}(z)$, the operator $\mathscr E^{\tau}$ can be rewritten as follows:
$$\textstyle
\mathscr E^{\tau}
=\big(\sum_a\eta(E,\phi^a)\frac{\pd}{\pd s_0^a}-\sum_{k,a,b}\E_b^a s_{k+1}^{b}\frac{\pd }{\pd s_{k}^a}\big)
+\sum_{k,a}(1-d_a-k) s_k^{a}\frac{\pd}{\pd s_k^a}.
$$
By equation~\eqref{eqn:diff-anc-corr}, we have
$$\textstyle
\big(\sum_a\eta(E,\phi^a)\frac{\pd}{\pd s_0^a}-\sum_{k,a,b}\E_b^a s_{k+1}^{b}\frac{\pd }{\pd s_{k}^a}\big)\bar\cF^{\tau}_{g}({\bf s})
=E\bar\cF^{\tau}_{g}({\bf s})+\frac{\delta_{g,0}}{2}\eta(\E s_0,s_0)+\<E\>^{\tau}_{1,1}.
$$
We see equation~\eqref{eqn:hom-Fg-anc} follows from equation~\eqref{eqn:hom-Fg-anc-E}.

Now we prove the correlator $\<E\>^{\tau}_{1,1}$ is a constant. Consider
	$$\textstyle
     \pd_{\tau^a}(\<E\>^{\tau}_{1,1})
	=\int_{\Mbar_{1,1}}\pd_{\tau^a}\Omega^{\tau}(E)=\int_{\Mbar_{1,1}}\Omega^{\tau}(\nabla_{\phi_a}E)+\int_{\Mbar_{1,2}}\Omega^{\tau}(E, \phi_a),
	$$
then by homogeneity condition of CohFT $\Omega^{\tau}$, the degree one part of $\Omega^{\tau}_{1,1}$ communicates with $\ad_{E}$, hence
	$$\textstyle
	\int_{\Mbar_{1,1}}\Omega^{\tau}(\nabla_{\phi_a}E)=-\int_{\Mbar_{1,1}}\Omega^{\tau}([E,\phi_a])
	=-\int_{\Mbar_{1,1}}E\big(\Omega^{\tau}(\phi_a)\big)
	=-\int_{\Mbar_{1,2}}\Omega^{\tau}(E, \phi_a).
	$$
This proves $\pd_{\tau^a}(\<E\>^{\tau}_{1,1})=0$ for each $a=0,\cdots, N-1$, and thus $\<E\>^{\tau}_{1,1}$ is a constant.
\end{proof}

\subsection{Semisimple CohFT and reconstruction theorem}
A CohFT $\Omega$ is called semisimple if the algebra $(H,\eta,*)$ is semisimple.
If the shifted CohFT $\Omega^{\tau}$ is semisimple, then we say $\tau$ is a semisimple point of the corresponding Frobenius manifold.
It is clear that if $\tau$ is a semisimple point, then so is the point in its (formal) neighborhood.

For a semisimple homogeneous CohFT $\Omega^{\tau}$ with vacuum $\vac^{\tau}(z)$, based on Givental's reconstruction procedure~\cite{Giv01a,Giv01b}, Teleman~\cite{Tel12} proved that $\Omega^{\tau}$ is uniquely and explicitly constructed from the Frobenius structure of $\Omega^{\tau}$.
In the follows, we briefly describe Givental and Teleman's result, we refer readers to~\cite{PPZ15} for more details and~\cite{CGL18} for generalizations of the setting.

Let $R^{\tau}(z)\in \mathbb E^{\tau}\otimes\End(\Frob)[[z]]$, called the $R$-matrix, be a formal power series
$$
R^{\tau}(z)=\id +R^{\tau}_1z+R^{\tau}_2z^2+\cdots,
$$
satisfying the symplectic condition $R^{\tau,*}(-z)R^{\tau}(z)=\id $, where $R^{\tau,*}(z)$ is the adjoint of $R^{\tau}(z)$ with respect to $\eta$.
Here $\mathbb E^{\tau}$ is an algebraic extension of the fractional field of $\mathbb A^{\tau}$.
We also introduce the matrix $V^{\tau}(z,w)$ associated with $R$-matrix:
$$
V^{\tau}(z,w)=\sum_{k,l\geq 0}V^{\tau}_{k,l}z^kw^l=\frac{\id -R^{\tau,*}(-z)R^{\tau}(-w)}{z+w}.
$$
An $R$-matrix defines an action, called the $R$-action, on a CohFT $\omega^{\tau}$ on $(\Frob,\eta)$ in the following way.
Let $\cG_{g,n}$ be the set of stable graphs of genus-$g$ with $n$-legs. For $\Gamma\in \cG_{g,n}$, define $\Cont_{\Gamma}:\Frob^{\otimes n}\to \mathbb E^{\tau}\otimes H^{*}(\Mbar_{g,n},\mathbb Q)$ by the following construction:
\begin{enumerate}[\qquad 1).]
\item place $\omega^{\tau}_{g(v),n(v)}$ at each vertex $v\in \Gamma$,
\item place $R^{\tau,*}(-\psi_i)\cdot $ at each leg $l_i\in \Gamma$ labeled by $i=1,\cdots,n$,
\item place $V^{\tau}(\psi_{v'},\psi_{v''})\phi_a\otimes \phi^a$ at each edge $e\in \Gamma$ connecting vertexes $v'$ and $v''$.
\end{enumerate}
The $R$-action on the CohFT $\omega^{\tau}$ is defined by
$$
(R^{\tau}\cdot \omega^{\tau})_{g,n}:=\sum_{\Gamma\in \cG_{g,n}}\frac{1}{|\Aut(\Gamma)|}\xi_{\Gamma,*}\Cont_{\Gamma},
$$
where $\xi_{\Gamma}:=\prod_{v\in \Gamma}\Mbar_{g(v),n(v)}\to \Mbar_{g,n}$ is the canonical map with image equal to boundary stratum associated to the graph $\Gamma$, and its push-forward $\xi_{\Gamma,*}$ induces a homomorphism from the strata algebra on  $\prod_{v\in \Gamma}\Mbar_{g(v),n(v)}$ to the cohomology ring (see~\cite{CGL18, PPZ15} for details).

Let $T^{\tau}(z)\in \mathbb E^{\tau}\otimes\Frob[[z]]$, called the $T$-vector, be a power series starting from degree $1$:
$$
T^{\tau}(z)=T^{\tau}_1z+T^{\tau}_2z^2+\cdots.
$$
A $T$-vector defines an action, called the $T$-action, on a CohFT $\omega^{\tau}$ on $(\Frob,\eta)$ by
$$
(T^{\tau}\cdot \omega^{\tau})_{g,n}(-)=\sum_{m\geq 0}\frac{1}{m!}\pi^{m}_{*}\omega^{\tau}_{g,n+m}(-\otimes T^{\tau}(\psi_{n+1})\otimes\cdots \otimes T^{\tau}(\psi_{n+m})),
$$
where $\pi^m:\Mbar_{g,n+m}\to \Mbar_{g,n}$ is the forgetful map forgetting the last $m$ marked points.

\begin{theorem}[Givental--Teleman reconstruction theorem]\label{thm:reconstruction}
	For a semisimple homogeneous CohFT $\Omega^{\tau}$ on $(\Frob,\eta)$ with the vacuum vector $\vac^{\tau}(z)$, 
let $\Frob=\oplus_{\alpha=1}^{N}\mathbb C\{\bar e_\alpha\}$ be the decomposition of algebra $(\Frob,\eta,*_{\tau})$ with $\eta(\bar e_\alpha,\bar e_\beta)=\delta_{\alpha,\beta}$, and $\bar e_\alpha*_{\tau}\bar e_\beta=\delta_{\alpha,\beta}\Delta_{\alpha}^{\frac{1}{2}} \bar e_\alpha$ for some functions $\Delta_{\alpha}=\Delta_{\alpha}(\tau)$,
then the CohFT $\Omega^{\tau}$ is uniquely reconstructed by the following formula:
	$$
	\Omega^{\tau}=R^{\tau}\cdot T^{\tau}\cdot (\oplus_{\alpha=1}^{N}\Omega^{\rm KW_\alpha}).
	$$
Here 
$\Omega^{\rm KW_\alpha}$ is the trivial CohFT on 1-dimensional space $\mathbb C\{\bar e_\alpha\}$
(i.e., $\Omega^{\rm KW_\alpha}_{g,n}(\bar e_\alpha,\dots,\bar e_\alpha)=1$),
$R^{\tau}$ is the $R$-matrix uniquely determined by
\beq\label{eqn:hom-R}
[R^{\tau}_{m+1},\E]=(m+\mu)R^{\tau}_m,\qquad m\geq 0,
\eeq
and $T^{\tau}$ is the $T$-vector given by
$$
T^{\tau}(z)=z\cdot {\bar {\bf 1}}-zR^{\tau}(z)^{-1}\vac^{\tau}(z),
$$
where $\bar {\bf 1}=\bar e_1+\cdots+\bar e_N$.
\end{theorem}
\begin{remark}\label{rmk:semicohft-vacuum}
For an arbitrary semisimple homogeneous CohFT $\Omega^{\tau}$ (without assuming the vacuum axiom in advance), Teleman~\cite{Tel12} has proved that there always exists a unique vacuum vector $\vac^{\tau}(z)$ such that the vacuum axiom~\eqref{eqn:vacuum-axiom} holds for $n\geq 1$ (but may fail for $n=0$, see~\cite[Remark 3.2]{Tel12}) and one always has
$$
\Omega^{\tau}_{g,n}=(R^{\tau}\cdot T^{\tau}\cdot (\oplus_{\alpha=1}^{N}\Omega^{\rm KW_\alpha}))_{g,n},\qquad n\geq 1.
$$
\end{remark}

\section{Descendents for CohFT}
\label{sec:CohFT-des}
 In this section, we first introduce the descendent potentials for a calibrated CohFT with vacuum, drawing parallels to the  GW theory. Our aim is to formally define the descendents in a manner independent of any specific geometric theory.
Next, we extend Givental's results regarding his Lagrangian cone to the context of the generalized Frobenius manifold.
	At the end of this section, we introduce the homogeneity condition for   both the calibrations and the descendent potentials, and subsequently define two types of Virasoro operators.
 
 \subsection{$S$- and $\dvac$-calibration of generalized Frobenius manifold} \label{sec:calibration}
For a CohFT with flat unit, the Kontsevich--Manin formula~\eqref{eqn:DA-correspondence} provides an approach to define a (formal) total descendent potential from the total ancestor potential via the $S$-matrix.
For a general CohFT (not necessarily contains a flat unit), the $S$-matrix can always be defined to be an operator-valued series which has the form:
$$S^{\tau}(z)= \id+\textstyle \sum_{n\geqslant 1} S^\tau_n z^{-n} \in \End H [[z^{-1}]]$$ 
and  satisfies the symplectic condition 
\beq\label{eqn:S-symp}
S^{\tau,*}(-z)S^{\tau}(z)=\id,
\eeq 
and the QDE: 
\beq\label{eqn:QDE-S}
z \pd_{\tau^a}S^{\tau}(z)= \phi_a*_{\tau}S^{\tau}(z), \qquad a=0,\cdots,N-1.
\eeq
It is clear that the $S$-matrix is not uniquely determined by the above conditions.
A choice of $S^{\tau}(z)$ is called an $S$-calibration of the (generalized) Frobenius manifold (of the CohFT).

Note that the flat unit plays a crucial role in the original Kontsevich–Manin formula~\eqref{eqn:DA-correspondence}.
To apply this formula for defining a (formal) total descendent potential in arbitrary CohFTs, we need to extend the notion of the flat unit on the descendent side, analogous to the vacuum vector on the ancestor side.
\begin{definition}[$\dvac$-calibration]
The $\dvac$-vector $\dvac^{\tau}(z)$ is an vector-valued series in $\Frob((z^{-1}))$ which satisfies the QDE: 
\beq\label{eqn:QDE-vacuum-des}
	z\pd_{\tau^a}\dvac^{\tau}(z)=\phi_a*_{\tau}\dvac^{\tau}(z)-\phi_a,\qquad
a=0,\cdots,N-1.
\eeq
Note the $\dvac$-vector is not uniquely determined by the above conditions. A choice of $\dvac^{\tau}(z)$ is called a $\dvac$-calibration of the (generalized) Frobenius manifold (of the CohFT).
\end{definition}

We note here the QDE~\eqref{eqn:QDE-vacuum-des} for the $\dvac$-vector $\dvac^{\tau}(z)$ is the same as the QDE~\eqref{eqn:QDE-vacuum} for the vacuum vector $\vac^{\tau}(z)$.
However, these two vectors belong to different spaces, $\dvac^{\tau}(z)$ resides in $\Frob((z^{-1}))$ while $\vac^{\tau}(z)$ lies in $\Frob[[z]]$ (in which space the solution to the QDE~\eqref{eqn:QDE-vacuum} is unique). 
The intersection of these two spaces is $\Frob[z]$, where the solution, if exists, is unique. 
In other words, if the $\dvac$-vector $\dvac^{\tau}(z)$ is a polynomial in $z$, then the vacuum vector $\vac^{\tau}(z)$ must also be a polynomial, and we have $\dvac^{\tau}(z)=\vac^{\tau}(z)$. 

Given an  $S$-matrix $S^{\tau}(z)$ and a $\dvac$-vector $\dvac^{\tau}(z)$, we define the $J$-function by:
\beq\label{def:J-function-small}
J^{\tau}(-z):=-zS^{\tau,*}(-z)\dvac^{\tau}(z).
\eeq
By equations~\eqref{eqn:QDE-S} and~\eqref{eqn:QDE-vacuum-des}, we see $\pd_{\tau^a}J^{\tau}(-z)=S^{\tau,*}(-z)\phi_a$ and thus $J^{\tau}(-z)$ has form
\beq\label{eqn:form-J}
J^{\tau}(-z)=\tau-\tau_0-z\dil(z)+  \textstyle \sum_{k\geq 0}J^{\tau}_{k,a}\phi^a(-z)^{-k-1}   \  \in \   H((z^{-1}))
\eeq 
for some constant vectors $\tau_0\in\Frob$, $\dil(z)\in\Frob[z]$.
Especially, we have
 \beq\label{eqn:dvac-u}
\dil(z)=[S^{\tau,*}(-z)\dvac^{\tau}(z)]_{+}.
\eeq
This is a constant vector and it will be used to generalize the dilaton shift.
\begin{remark}\label{rmk:J-nu}
Rather than introducing the \(\dvac\)-vector, if one begins by defining the \(J\)-function through  certain equations, the \(S\)-matrix can be derived from the equation \(\frac{\partial}{\partial\tau^a}J^{\tau}(-z)=S^{\tau,*}(-z)\phi_a\), and the \(\dvac\)-vector is determined from \(\dvac^{\tau}(z)=-z^{-1}S^{\tau}(z)J^{\tau}(-z)\). This alternative approach offers a different path to define these fundamental elements, highlighting a connection between the \(J\)-function and the \(\dvac\)-vector.
\end{remark}

Furthermore, note that $\eta(\phi_a,\pd_{\tau^c}S^{\tau}_1\phi_b)=\pd_{\tau^a}\pd_{\tau^c}\pd_{\tau^b}\Phi(\tau)$,
where $\Phi(\tau)$ is the potential of the Frobenius manifold induced by $\Omega^{\tau}$.
By taking integration of this equation, one has $\eta(\phi_a,S^{\tau}_1\phi_b)=\pd_{\tau^a}\pd_{\tau^b}\Phi(\tau)+C_{a,b}$ for some constants $C_{a,b}$. 
Recall the function $\Phi$ is defined up to a at most quadratic function of $\tau$, we require the function $\Phi(\tau)$ to satisfy
$\eta(\phi_a,S^{\tau}_1\phi_b)=\pd_{\tau^a}\pd_{\tau^b}\Phi(\tau)$.
Similarly, we require the function $\Phi(\tau)$ to satisfy $J_{0,a}^{\tau}=\pd_{\tau^a}\Phi(\tau)$.
The function $\Phi(\tau)$ is now determined up to a constant $\Phi(0)$.

\subsection{Descendent potentials for a  CohFT} 
Now we are ready to define the (formal) descendent potentials by applying a generalized version of the Kontsevich--Manin formula of an $S$- and $\dvac$-calibrated CohFT.

Firstly, for an arbitrary CohFT $\Omega^{\tau}$ (without requiring $S$- and $\dvac$-calibration), we define $F_1(\tau)$ by 
requiring $\pd_{\tau^a}F_1(\tau)=\int_{\M_{1,1}} \Omega^\tau_{1,1}( \phi_a)$. It is clear that $F_1(\tau)$ is determined by this condition up to a constant $F_1(0)$:
$$
F_1(\tau)=F_1(0)+\sum_{n\geq 1}\frac{1}{n!}\int_{\Mbar_{1,n}}\Omega_{1,n}(\tau,\cdots,\tau).
$$
Secondly, after choosing an $S$- and $\dvac$-calibration for a CohFT $\Omega^{\tau}$, we introduce
\beq\label{def:tildeW}\textstyle
\widetilde W^{\tau}({\bf t}-\tau,{\bf t}-\tau) :=2\Phi(\tau)+2\sum_{k,a} J^{\tau}_{k,a}\cdot (t_k^a-\delta_{k,0}\tau^a)+W^\tau({\bf t}-\tau,{\bf t}-\tau),
\eeq
where the bilinear form $W^{\tau}$, is defined by using $S$-matrix as in GW theory (see  \eqref{eqn:W-S}).
We note here for the CohFT with flat unit ${\bf 1}$, one has $\widetilde W^{\tau}({\bf t}-\tau,{\bf t}-\tau)=W^\tau(\tilde{\bf t},\tilde{\bf t}) +2c$ for some constant $c$ 
which can be taken to be vanished by changing $\Phi(\tau)\to \Phi(\tau)-c$.
Now we can define the (formal) total descendent potential $\cD({\bf t};\hbar)$ for an $S$- and $\dvac$-calibrated CohFT:

\begin{definition}
\label{def:formaldes}
	For an $S$- and $\dvac$-calibrated CohFT, its total descendent potential $\cD({\bf t};\hbar)$ 
	is defined  via the total ancestor potential by the following formula:
	\beq\label{eqn:des-anc}
	\cD({\bf t};\hbar)=e^{F_1(\tau)+\frac{1}{2\hbar^2}\widetilde W^{\tau}({\bf t}-\tau,{\bf t}-\tau)}\cdot\cA^{\tau}({\bf s(t)};\hbar),
	\eeq
	where  the coordinate transformation ${\bf s}={\bf s(t)}$  is defined by ${\bf s}(z)=[S^{\tau}(z){\bf t}(z)]_{+}-\tau$.
\end{definition}

Given the total descendent potential $\cD({\bf t};\hbar)$, the genus-$g$ descendent potential $\cF_g({\bf t})$ is defined by
$$
\log(\cD({\bf t};\hbar))=  \sum_{g\geq 0}\hbar^{2g-2}\cF_{g}({\bf t}),
$$
and the descendent invariants / correlators $\<-\>_{g,n}$ are defined by
$$
\<\phi_{a_1}\psi^{k_1},\cdots,\phi_{a_n}\psi^{k_n}\>_{g,n}
:=\pd_{t^{a_1}_{k_1}}\cdots\pd_{t^{a_n}_{k_n}}\cF_g({\bf t})\big|_{\bf t=0}.
$$

Introduce the $\tau$-shifted descendent correlators $\<-\>^{\tau}_{g,n}$:
\beq\label{def:bracket-des}
\<\phi_{a_1}\psi^{k_1},\cdots,\phi_{a_n}\psi^{k_n}\>^{\tau}_{g,n}
:=\pd_{t^{a_1}_{k_1}}\cdots\pd_{t^{a_n}_{k_n}}\cF_g({\bf t})\big|_{\bf t=\tau},
\eeq
then formula~\eqref{eqn:des-anc} can be rewritten in the correlator form as follows: for $2g-2+n>0$,
	\beq\label{eqn:des-anc-S}
	\<\phi_{a_1}\psi^{k_1},\cdots,\phi_{a_n}\psi^{k_n}\>^{\tau}_{g,n}
	=\<S^{\tau}(\bar\psi)\phi_{a_1}\bar\psi^{k_1},\cdots,S^{\tau}(\bar\psi)\phi_{a_n}\bar\psi^{k_n}\>^{\tau}_{g,n},
	\eeq
	where the correlators with insertion $\phi_{a}\bar\psi^{k}$ of negative $k$ are set to be $0$.
\begin{remark}When there are only the primary insertions, two definitions (equations~\eqref{def:bracket-des} and \eqref{def:bracket-anc}) of the correlator $\<\phi_{a_1},\cdots,\phi_{a_n}\>_{g,n}^{\tau}$ coincide with each other for $2g-2+n>0$; 
for $2g-2+n\leq 0$, the ancestor correlators are defined to be $0$ and we only consider the descendent correlators defined by equation~\eqref{def:bracket-des}.
\end{remark}

We have the following theorem generalizing Theorem 7.1 in~\cite{Giv01a}.
\begin{theorem}\label{thm:des-tau}
	The total descendent potential $\cD({\bf t};\hbar)$ does not depend on $\tau$.
\end{theorem}
\begin{proof}
	The Theorem is equivalent to the following equations: for $g\geq 0$,
	$$
	\pd_{\tau^b}\cF_{g}({\bf t})=0,\qquad b=0,\cdots, N-1.
	$$
	We prove these equations by proving the vanishing of coefficient of $\prod_{i=1}^{n} t_{k_i}^{a_i}$ on the left-hand side of these equations for all $n\geq 0$ and all $k_i\geq 0$, $a_i=0,\cdots, N-1$. Precisely, we prove
\beq\label{eqn:diff-des-corr}
\pd_{\tau^b}\big(\<\phi_{a_1}\psi^{k_1},\cdots,\phi_{a_n}\psi^{k_n}\>^{\tau}_{g,n}\big)=\<\phi_b,\phi_{a_1}\psi^{k_1},\cdots,\phi_{a_n}\psi^{k_n}\>^{\tau}_{g,n}.
\eeq
For $(g,n)=(0,0), (0,1), (0,2), (1,0)$, the results hold by definition. 
For the cases with $2g-2+n>0$, by using equation~\eqref{eqn:des-anc-S} and the QDE~\eqref{eqn:QDE-S} for $S$-matrix,
equation~\eqref{eqn:diff-des-corr} follows from equation~\eqref{eqn:diff-anc-corr}.
\end{proof}

\subsection{Givental's Lagrangian submanifold for generalized Frobenius manifold}
Following Givental~\cite{Giv01a,Giv04}, we introduce the loop space $\cH:=\Frob((z^{-1}))$, which by definition is the space of Laurent series in $z^{-1}$  with vector coefficients from $\Frob$:
$$ \textstyle
\Frob((z^{-1}))=\Big\{\sum_{k\geq 0} q_k z^k+ p_{k } (-z)^{-k-1} :  q _k, p_l \in \Frob  \Big\}.
$$ 
There is a symplectic bilinear form in $\cH$: for $f,g\in\cH$,
$$
\omega(f,g):=\frac{1}{2\pi {\bf i}}\oint\eta(f(-z), g(z)) dz=-\omega(g,f),
$$
where the integration is formally defined by $\frac{1}{2\pi {\bf i}}\oint \sum_{k} a_k z^{k}dz:=a_{-1}$.

By the polarization $\cH=\cH_{+}\oplus\cH_{-}=\Frob[z]\oplus z^{-1}\Frob[[z^{-1}]]$, 
the loop space $\cH$ can be identified with the cotangent bundle $T^{*}\cH_{+}$ of space $\cH_{+}$.
The genus-0 descendent potential $\cF_0({\bf t})$ can be viewed as a formal function on $\cH_{+}$ via the shift ${\bf q}(z)={\bf t}(z)-\tau_0-z\dil(z)$, where the shift is determined by equation~\eqref{eqn:form-J}.

Denote by $\cL$ the graph of the differential $d\cF_0$:
$$
\cL:=\big\{({\bf p,q})\in T^{*}\cH_{+}:{\bf p}=d_{\bf q}\cF_0({\bf t})\big\}.
$$
We see $\mathcal L$ is a Lagrangian submanifold of the loop space $\cH$.
We denote by $J({\bf t},-z)$ the point $({\bf q,p})$ in $\cL$ via the shift ${\bf q}(z)={\bf t}(z)-\tau_0-z\dil(z)$
and we call it the $J$-function on the big phase space.
Recall the $J$-function $J^{\tau}(-z)$ defined by~\eqref{def:J-function-small} and its relation with the genus-0 potential $\cF_0({\bf t})$~\eqref{def:tildeW}, we see $J^{\tau}(-z)=J({\bf t},-z)|_{\bf t=\tau}$. 
We introduce the tangent space of $\cL$:
$$
T\cL:=\sspan_{\mathbb C}\{\pd_{t_k^a}J({\bf t},-z): k\geq 0, a=0,\cdots, N-1\}.
$$

For the theory containing a flat unit, Givental proved that the system of ``string equation + dilaton equation + topological recursion relation" of the function $\cF_0({\bf t})$ is equivalent to $\cL$ is a Lagrangian cone satisfying $J({\bf t},-z)\in z\cdot T\cL \subset T\cL$ (see~\cite[Theorem 1]{Giv04} for details).
For arbitrary CohFTs, the string equation and the dilaton equation may no-longer hold,
we have only the topological recursion relation which is equivalent to the {\it constitutive relation} introduced by Dijkgraaf and Witten~\cite{DW90}. 
Introduce the DW map defined by the (formal) solution $\tau=\tau_{\rm DW}$ of the following equation:
\beq\label{eqn:DWmap}
\tau=[z^{0}](S^{\tau}(z){\bf t}(z)),
\eeq
where $[z^0]f(z)$ stands for the coefficient of $z^0$ in series $f(z)$,
then the constitutive relation reads
\beq\label{eqn:corr02-DW}
\frac{\pd^2\cF_0({\bf t})}{\pd t_k^a \pd t_{l}^b}=\<\phi_a\psi^k,\phi_b\psi^l\>^{\tau_{\rm DW}}_{0,2}.
\eeq
For the formal genus-$0$ potential $\cF_0({\bf t})$ defined in this paper, the constitutive relation can also be proved as follows. 
Firstly, by definition, we have
$$\textstyle
\frac{\pd^2\cF_0({\bf t})}{\pd t_k^a \pd t_{l}^b}=\sum_{m\geq 0}\frac{1}{m!}\<\phi_a\psi^k,\phi_b\psi^l,{\bf t}(\psi)-\tau,\cdots,{\bf t}(\psi)-\tau\>^{\tau}_{0,2+m}.
$$
	Secondly, by the relation of $\tau$-shifted descendent correlators and ancestor correlators~\eqref{eqn:des-anc-S}, 
	when taking $\tau$ to be the DW map (notice that $[z^0](S^{\tau}(z)({\bf t}(z)-\tau))=0$),
	we see for $m\geq 1$, the dimensional reason (the total degree of insertions is at least $m$ while $\dim \Mbar_{0,2+m}=m-1$) makes the correlators vanishing and one gets equation~\eqref{eqn:corr02-DW}.

\begin{proposition}\label{prop:Lagrangian}
	We have the following results:
	
	(1). $T\cL=S^{\tau_{\rm DW}, *}(-z) H[z]$;
	
	(2). $J({\bf t},-z)-J^{\tau_{\rm DW}}(-z)\in z\cdot T\cL $.
	 
\end{proposition}
\begin{proof}
	By definition and the constitutive relation~\eqref{eqn:corr02-DW},
$$\textstyle
	\pd_{t_k^a}(J({\bf t},-z))=\phi_a z^k+\sum_{l\geq 0}\phi_b\cdot \<\phi^b\psi^l,\phi_a\psi^k\>^{\tau_{\rm DW}}_{0,2}\cdot (-z)^{-l-1}.
$$
	Since $\<\phi^b\psi^{l},\phi_a\psi^k\>^{\tau}_{0,2}=\sum_{i=0}^{l}(-1)^{i}(S^{\tau,*}_{l-i}\cdot S^{\tau}_{k+i+1})_{a}^{b}$,
	we have
$$\textstyle
\pd_{t_k^a}(J({\bf t},-z))=	\phi_a z^{k}-\sum_{l\geq 0}\sum_{i=0}^{l}S^{\tau_{\rm DW},*}_{l-i}\cdot S^{\tau_{\rm DW}}_{k+i+1} \cdot (-z)^{-l+i}\cdot z^{-i-k-1}\cdot \phi_a z^k.
$$
Notice that the right hand side of above equation equals $\phi_a z^{k}-S^{\tau_{\rm DW},*}(-z)[S^{\tau_{\rm DW}}(z)\phi_a z^{k}]_{-}$, where $[f(z)]_{-}$ stands for the negative part of a Laurent series $f(z)$ in $z^{-1}$, then by equation~\eqref{eqn:S-symp}, we have
\beq\label{eqn:diff-W-S}\textstyle
\pd_{t_k^a}(J({\bf t},-z))=S^{\tau_{\rm DW},*}(-z)[S^{\tau_{\rm DW}}(z)\phi_a z^{k}]_{+}.
\eeq
	This proves the first part of the Proposition.
	
	Similarly, by noticing 
	$\frac{\pd \cF_0({\bf t})}{\pd t_k^a}=\<\phi_a\psi^k\>^{\tau_{\rm DW}}_{0,1}+\<\phi_a\psi^k,{\bf t}(\psi)-\tau_{\rm DW}\>^{\tau_{\rm DW}}_{0,2}$,
	we have
	$$
	J({\bf t},-z)-J^{\tau_{\rm DW}}(-z)=S^{\tau_{\rm DW},*}(-z)[S^{\tau_{\rm DW}}(z)({\bf t}(z)-\tau_{\rm DW})]_{+}.
	$$
	The second part of the Proposition follows from $[z^0]S^{\tau_{\rm DW}}(z)({\bf t}(z)-\tau_{\rm DW})=0$.
\end{proof}

\subsection{Homogeneity conditions for descendent potentials} 
\label{subsec:hom-des}
For a homogeneous CohFT with vacuum $\vac^{\tau}(z)$ and calibrated with $S$-matrix $S^{\tau}(z)$ and $\dvac$-vector $\dvac^{\tau}(z)$,
we introduce homogeneity condition for $\dvac^{\tau}(z)$ by directly applying the homogeneity condition~\eqref{eqn:hom-vacuum} for the vacuum $\vac^{\tau}(z)$, precisely,
\beq\label{eqn:hom-vacuum-des}
(z\pd_z+E)\dvac^{\tau}(z)=-\big(\mu+\tfrac{\delta}{2}\big)\dvac^{\tau}(z).
\eeq
We call such a vector $\dvac^{\tau}(z)$ is homogeneous.
Similarly as the vacuum vector,  the homogeneity condition~\eqref{eqn:hom-vacuum} for $\dvac^{\tau}(z)$ can be rewritten as
\beq\label{eqn:hom-vacuum-des-2}
(z\pd_z+\tfrac{1}{z}\E+\mu+\tfrac{\delta}{2})\dvac^{\tau}(z)-\tfrac{1}{z}E=0,
\eeq
and this equation holds if and only if it holds at $\tau=0$.
\begin{lemma}\label{lem:cm}
Given a homogeneous $\dvac$-vector, the   $\cdes$   in Definition~\ref{defnofm0} (2) is a constant.
\end{lemma}
\begin{proof}
By the QDE~\eqref{eqn:QDE-vacuum-des} and the homogeneity condition~\eqref{eqn:hom-vacuum-des}, 
for each $a=0,\cdots,N-1$,
$$
\pd_{\tau^a}\big(\eta(E,\dvac^{\tau}(z))\big)=(1-\delta)\cdot \eta(E,\dvac^{\tau}(z))-\eta(\phi_a,z\pd_z\dvac^{\tau}(z)).
$$
Therefore, $[z^{1-\delta}]\pd_{\tau^a}\big(\eta(E,\dvac^{\tau}(z))\big)=0$ and this proves $\cdes$ is a constant.
\end{proof}
To introduce the homogeneity condition for the $S$-matrix, we first recall in GW theory of $X$, the homogeneity condition for the $S$-matrix is given by
\beq\label{eqn:hom-S-GW}
(z\pd_z+E)S^{\tau}(z)=[S^{\tau}(z),\mu]+S^{\tau}(z)\rho/z,
\eeq
where $\rho=c_1(X)\cup$. 
For an arbitrary $S$-calibrated homogeneous CohFT, there is no geometric definition of the operator $\rho$,
and there may not have such operator $\rho$ such that the $S$-matrix satisfies equation~\eqref{eqn:hom-S-GW}.
Nonetheless, we can always have the following definition:
\begin{definition}
Given an $S$-calibrated homogeneous CohFT, the $\rho$-matrix takes the form $\rho(z)=\sum_{k\geq 0}\rho_kz^{-k}\in \End(\Frob)[[z^{-1}]]$ and is defined by
\beq\label{eqn:hom-S}
(z\pd_z+E)S^{\tau}(z)=[S^{\tau}(z),\mu]+S^{\tau}(z)\rho(z)/z.
\eeq
The $S$-matrix is called homogeneous if the operator $\rho(z)$ satisfies
\beq\label{eqn:hom-rho}
z\pd_z\rho(z)=\rho(z)+[\rho(z),\mu].
\eeq
For such case, we also call $\rho(z)$ is homogeneous and call equation~\eqref{eqn:hom-S} (resp. equation~\eqref{eqn:hom-rho}) the homogeneity condition for $S$-matrix (resp. $\rho$-matrix).
\end{definition}
By the QDE~\eqref{eqn:QDE-S}, the homogeneity condition for $S$-matrix can be rewritten as
\beq\label{eqn:hom-S-2}
(z\pd_z+\tfrac{1}{z}\E+\mu)S^{\tau}(z)-S^{\tau}(z)(\mu+\rho(z)/z)=0.
\eeq
Denote by $A(\tau,z)$ the left-hand side of above equation, then $\pd_{\tau^a}A(\tau,z)=\frac{1}{z}\phi_a*_{\tau}A(\tau,z)$.
Therefore, the homogeneity condition~\eqref{eqn:hom-S-2} for $S$-matrix holds if and only if it holds at $\tau=0$.

Based on the definition, we highlight three important results regarding the $\rho$-matrix and the homogeneous $S$-matrix.
Firstly, $\rho(z)$ does not depend on $\tau$ and satisfies $\rho^*(z)=\rho(-z)$.
This can be easily derived by using equations~\eqref{eqn:S-symp}, \eqref{eqn:QDE-S} and \eqref{eqn:euler-qp}.
Secondly, since $\mu$ is a diagonal matrix, one can see from equation~\eqref{eqn:hom-rho} a homogeneous $\rho$-matrix must be a polynomial in $z^{-1}$ and must be nilpotent
(see Appendix~\ref{sec:cons-hom-J} for a proof).
Finally, given a homogeneous CohFTs, the homogeneous $S$-matrix always exists (equivalently, the homogeneous $\rho$-matrix always exists).
Indeed, when the vacuum vector is flat (thus is the flat unit), this is equivalent to~\cite[Lemma 2.5]{Dub99} introduced by Dubrovin.
The general cases can be proved by using similar method as the proof of~\cite[Lemma 2.5]{Dub99}. 
A detailed proof is put in Appendix~\ref{sec:hom-S}.

By the definition of the $J$-function~\eqref{def:J-function-small}, 
the homogeneity conditions for the $\dvac$-vector and the $S$-matrix imply the following homogeneity condition for the $J$-function:
\beq\label{eqn:hom-J}
	(z\pd_z+E)J^{\tau}(-z)=\big(1-\tfrac{\delta}{2}-\mu-\rho(z)/z\big)J^{\tau}(-z).
\eeq
We show some consequences of this equation.
Firstly, the vector $\dil(z)$ satisfies 
\beq\label{eqn:hom-J-positive}
\big(z\pd_z+\mu+\tfrac{\delta}{2}\big)\dil(z)=0,\qquad
\rho(z)\dil(z)\in \Frob[z^{-1}].
\eeq
Secondly, the Euler vector field $E$ (see equation~\eqref{def:euler}) relates with $\rho$ and $\dil$ by
\beq\textstyle\label{eqn:hom-J-0}
\sum_{k\geq 0}\rho_k\dil_k=\sum_{d_a=1}r^a\phi_a.
\eeq
Thirdly, the potential $\Phi(z)$ for the Frobenius manifold satisfies
\begin{align}
E \Phi(\tau)=(3-\delta)\Phi(\tau) \textstyle +\frac{1}{2}\,\eta(\tilde\tau,\rho_0\tilde\tau)-\sum_{i\geq 0}\eta(\tilde\tau,\rho_{i+1}\dil_i)
 +\frac{1}{2}\sum_{i,j\geq 0}\eta(\dil_i,\rho_{i+j+2}\dil_j)+c_0. \label{eqn:hom-Phi}
\end{align}
for some constant $c_0$.
We put the explanation of these consequences in Appendix~\ref{sec:cons-hom-J}
\begin{remark}\label{rmk:cdes}
If $\delta\ne 3$, one can take $c_0=0$ by taking $\Phi(\tau)\to \Phi(\tau)-\frac{c_0}{3-\delta}$, but one can not do this if $\delta=3$.
For such case, if one further has $\dvac^{\tau}(z)\in \Frob[[z^{-1}]]$,
then the Virasoro-index $\virm=0$ and one can deduces $c_0=\cdes$ (the Virasoro constant defined in Definition~\ref{defnofm0} (2)) from equation~\eqref{eqn:hom-Phi}. 
\end{remark}

Given an homogeneous $S$-matrix with operator $\rho(z)$, we introduce the {\it descendent Euler vector field} $\mathscr E$ on the big phase space
(compare this with equation~\eqref{eqn:euler-anc-big}):
\beq\label{eqn:euler-des-big}
\mathscr E=\sum_{k,a}(1-d_a-k)\tilde t_k^{a}\frac{\pd}{\pd t_k^{a}}-\sum_{k,i,a,b}(\rho_i)_{b}^{a}\tilde t_{k+i+1}^{b}\frac{\pd}{\pd t_k^a},
\eeq
where $\tilde{\bf t}(z)={\bf t}(z)-\tau_0-z\dil(z)$.
We see equations~\eqref{eqn:hom-J-positive} and~\eqref{eqn:hom-J-0} imply $\mathscr E|_{{\bf t}=\tau}=E$.
\begin{proposition} \label{prop:hom-des}
The descendent potentials $\cF_g({\bf t})$, $g\geq 0$, of homogeneous CohFT with vacuum and calibrated by homogeneous $S$-matrix and $\dvac$-vector satisfy
	\beq\label{eqn:hom-Fg-des}
	\mathscr E \cF_g({\bf t})
	=(\delta-3)(g-1)\cF_g({\bf t}) +\frac{\delta_{g,0}}{2}\sum_{i,j\geq 0}\eta(\tilde t_i,\rho_{i+j}\tilde t_j)+\delta_{g,0}c_0 + \delta_{g,1}\cc ,
	\eeq
where $c_0$, $\cc$ are two constants given by equation~\eqref{eqn:hom-Phi} and Proposition~\ref{prop:hom-anc} respectively.
\end{proposition}
\begin{proof}
We prove equation~\eqref{eqn:hom-Fg-des} by the homogeneity conditions for $J$-function~\eqref{eqn:hom-J} and for $\bar\cF_g({\bf s})$~\eqref{eqn:hom-Fg-anc-E}, as well as equation~\eqref{eqn:hom-S} with homogeneous $\rho(z)$.
Firstly, by the coordinate transformation ${\bf s}(z)=[S^{\tau}(z){\bf t}(z)]_{+}-\tau$ and equation~\eqref{eqn:hom-S}, we have
$$
(E+\mathscr E)(s_k)=(1-\tfrac{\delta}{2}-\mu-k)s_k,
$$
Together with equation~\eqref{eqn:hom-Fg-anc-E}, we obtain
$$
(E+\mathscr E)\bar\cF^{\tau}_{g}({\bf s(t)})=(\delta-3)(g-1)\bar\cF^{\tau}_{g}({\bf s(t)}).
$$
Secondly, $(E+\mathscr E)F_1(\tau)=\<E\>^{\tau}_{1,1}$ and as we have shown in the proof of Proposition~\ref{prop:hom-anc}, $\<E\>^{\tau}_{1,1}=\cc$ is a constant.
Thirdly, equation~\eqref{eqn:hom-S} gives us the following homogeneity condition for the $W$-matrix (see equation~\eqref{eqn:W-S} for the definition):
$$
(z\pd_z+w\pd_{w}+E)W^{\tau}(z,w)=[W^{\tau}(z,w),\mu]+\tfrac{\rho^{*}(z)}{z} \cdot W^{\tau}(z,w)+ W^{\tau}(z,w)\cdot  \tfrac{\rho(w)}{w} +W^{\tau}(z,w).
$$
Then by equation~\eqref{eqn:hom-Phi} and by direct computations, we have
$$\textstyle
(E+\mathscr E)\widetilde W^{\tau}({\bf t-\tau};{\bf t-\tau})=(3-\delta)\widetilde W^{\tau}({\bf t-\tau};{\bf t-\tau})
+\sum_{i,j\geq 0}\eta(\tilde t_i,\rho_{i+j}\tilde t_j)+2c_0.
$$ 
The Proposition follows from these computations and $E(\cF_g({\bf t}))=0$ (by Theorem~\ref{thm:des-tau}, $\cF_g({\bf t})$ does not depend on $\tau$).
\end{proof}

\begin{remark}
A descendent potential that satisfies \eqref{eqn:hom-Fg-des} is called homogeneous. Then, the generalized Virasoro conjecture claims that the homogeneous total descendent potential $\cD({\bf t};\hbar)$ satisfies the   {\it generalized Virasoro constraints} under certain conditions. 
\end{remark}

\subsection{Generalized  and ancestor Virasoro operators  via quantization}
For a infinitesimal transformation $A\in\End(\Frob((z^{-1})))$, i.e., $\omega(Af,g)=-\omega(f,Ag)$, let
\beq\label{eqn:quantization-A}
h_A(f):=\tfrac{1}{2}\cdot \omega(Af,f).
\eeq
and  we define  the quantization of $A$ by
$\widehat{A}:=\widehat{h_A(f)}$, where $f = \sum_{k\geq 0} q_k z^k + p_k (-z)^{-k-1}$   for $p_k=p_{k,a}\phi^a,q_k=q^a_k\phi_a \in \Frob$   and the quantization of quadratic Hamiltonians are given by
$$
(q_k^a q_{l}^b)^{\hat{}}=q_k^a q_{l}^b/\hbar^2,\qquad 
(q_k^a p_{l,b})^{\hat{}}=q_k^a \pd_{q_l^b},\qquad
(p_{k,a} p_{l,b})^{\hat{}}=\hbar^2 \pd_{q_k^a}\pd_{q_l^b}.
$$ 

For the case with a homogeneous $\rho$-matrix $\rho(z)=\rho\in\End(\Frob)$, following Givental~\cite{Giv01a}, we introduce the infinitesimal transformations $\ell_m$ for $m\geq -1$:
$$
\ell_m=-z^{-1/2}(z\pd_z z+\mu z+\rho)^{m+1}z^{-1/2} .
$$
Then the Virasoro operator $L_m$ (see equation~\eqref{def:Lm}) can be expressed as
\beq\label{def:Lm-quantum}
L_m=\widehat{\ell_m}+\tfrac{\delta_{m,0}}{4}\, \str\big(\tfrac{1}{4}-\mu^2\big),
\eeq
where the coordinate ${\bf q}(z) =\sum_{k} q_k z^k$ is related with the coordinate ${\bf t}(z)$ by ${\bf q}(z)=\tilde{\bf t}(z)$. 
In general, for the case with a homogeneous $\rho$-matrix $\rho(z)=\rho_0+\rho_1z^{-1}+\cdots\in \End(\Frob)[z^{-1}]$, we generalize the definition of $\ell_m$ by
$$
\ell'_m=-z^{-1/2}(z\pd_z z+\mu z+\rho(z))^{m+1}z^{-1/2},
$$
and define $L_m$ by substituting $\ell'_m$ for $\ell_m$ in equation~\eqref{def:Lm-quantum}.

Similarly, we introduce the infinitesimal transformations for the ancestors in parallel:
\begin{lemma}
For $m\geq -1$, let $\ell^{\anc}_m$ be the infinitesimal transformations defined by
$$
\ell^{\anc}_m=-z^{-1/2}(z\pd_z z+\mu z+\E)^{m+1}z^{-1/2},
$$
then
the Virasoro operator $L^{\anc}_m$ (see equation~\eqref{def:Lmtau}) can be expressed as
$$
L^{\anc}_m=\widehat{\ell^{\anc}_m}+\frac{1}{4}\sum_{i+j=m}\str\big(\E^i(\tfrac{1}{2}+\mu)\E^j(\tfrac{1}{2}-\mu)\big),
$$
where the coordinate ${\bf q}(z)$ is related with the coordinate ${\bf s}(z)$ by ${\bf q}(z)={\bf s}(z)-z\vac^{\tau}(z)$.
\end{lemma}
\begin{proof} 
Compare $\ell^{\anc}_m$ with $\ell_m$, by replacing $\rho$ with $\E$, the proof for the ancestor case closely parallels that of the descendent case.
\end{proof}
\begin{remark}
Let $\ell_{m}^{\circ}:=-z^{-1/2}(z\pd_z z)^{m+1}z^{-1/2}$
then we have 
$$
\ell'_m=z^{-\mu}z^{-\sum_i\rho_{i}}\ell_m^{\circ}z^{\sum_{i}\rho_i}z^{\mu},\qquad
\ell^{\anc}_{m}=S^{\tau}(z)z^{-\mu}z^{-\sum_i\rho_{i}}\ell_{m}^{\circ}z^{\sum_{i}\rho_i}z^{\mu}S^{\tau}(z)^{-1}.
$$
For either case, $\rho(z)\in \End(\Frob)[z^{-1}]$ or $\rho(z)=\rho\in \End(\Frob)$, we have 
firstly, $\ell^{\anc}_{m}$ does not depend on the precise form of $\rho(z)$ and 
secondly, the definitions of $\ell'_m$ and $\ell_m$ share the same form.
Since the results for arbitrary $\rho(z)$ are parallel to those for $\rho(z)=\rho\in \End(\Frob)$, in the follows of this paper, we will focus on the case that $\rho(z)=\rho$ for simplicity.
\end{remark}

Introduce two operators $D_{\rho,z}$ and $D_{\E,z}$ by
\beq\label{def:DAz}
D_{A,z}:=A+(\mu+\tfrac{3}{2})z+z^2\pd_z, \qquad A=\rho,\,  \E,
\eeq
then by the QDE~\eqref{eqn:QDE-S} and the homogeneity condition~\eqref{eqn:hom-S} for $S$-matrix,
one can see 
\beq\label{eqn:D-rho-E}
S^{\tau}(z)D_{\rho,z}=D_{\E,z}S^{\tau}(z).
\eeq
The infinitesimal transformations $\ell_m$ and $\ell^{\anc}_m$ can be expressed as follows:
$$
\ell_m=-D^{m+1}_{\rho,z}\cdot z^{-1},\qquad
\ell^{\anc}_m=-D^{m+1}_{\E,z}\cdot z^{-1}.
$$
By identifying $\phi_a\psi^k$ with $\pd_{t^a_k}$,
the Virasoro operators $L_m$, $m\geq -1$, can be formally rewritten as follows
	\begin{align}
		L_m=&\, \textstyle \frac{1}{2\hbar^2}\,\eta(\rho^{m+1} \tilde {t}_0, \tilde {t}_0)
		-\delta_{m,0}\cdot \frac{1}{4}\, \str\big(\mu^2-\frac{1}{4}\big)
		+[D^{m+1}_{\rho,\psi}\psi^{-1}\tilde{\bf t}(\psi)]_{+}\nonumber\\
		&\, \textstyle +\frac{\hbar^2}{2}\sum_{k=1}^{m}(-1)^k\phi_a\psi^{k-1}\circ
		[D^{m+1}_{\rho,\psi}\phi^a\psi^{-k-1}]_{+}, \label{eqn:Lm-des}
	\end{align}
where the symbol ``$\circ$" represents the composition of two operators.
Similarly, by identifying $\phi_a\bar\psi^k$ with $\pd_{s^a_k}$, we have
	\begin{align}
		L^{\anc}_m=&\, \textstyle\frac{1}{2\hbar^2}\eta(\E^{m+1}s_0,  s_0)
		-\frac{1}{4} \sum_{i+j=m} \str(\E^i(\mu+\tfrac{1}{2})\E^j(\mu-\tfrac{1}{2}))
		+[D^{m+1}_{\E,\bar\psi}\bar\psi^{-1}\tilde{\bf s}(\bar\psi)]_{+}\nonumber\\
		&\, \textstyle+\frac{\hbar^2}{2}\sum_{k=1}^{m}(-1)^k\phi_a\bar\psi^{k-1}\circ
		[D^{m+1}_{\E,\bar\psi}\phi^a\bar{\psi}^{-k-1}]_{+}. \label{eqn:Lm-anc}
	\end{align}

By the quantization formulation of the Virasoro operators $L_m$, one can see for $m\geq -1$,
\beq\label{eqn:L0m}
\mathscr L_{0,m}({\bf t})=-\tfrac{1}{2}\, \omega(D^{m+1}_{\rho,z}z^{-1}J({\bf t},-z),J({\bf t},-z)).
\eeq
See equation~\eqref{def:Lgm-D} for  the definition of $\mathscr L_{g,m}({\bf t})$.

\begin{proposition}\label{prop:Lagrangian-vira}
For $m\geq -1$, the function $\mathscr L_{0,m}({\bf t})$ is a constant if and only if  
	$$
	D_{\rho,z}^{m+1}z^{-1}J^{\tau}(-z)\in T\cL.
	$$
\end{proposition}

\begin{proof}
We first notice that $v(z)\in T\cL$ if and only if $v(z)=S^{\tau,*}(-z)[S^{\tau}(z)v(z)]_{+}$.
This is because if $v(z)\in T\cL$, then there is some $v'(z)\in \Frob[[z]]$ such that $v(z)=S^{\tau,*}(-z)v'(z)$, 
by the symplectic condition of $S$-matrix~\eqref{eqn:S-symp}, we have $v'(z)=S^{\tau}(z)v(z)=[S^{\tau}(z)v(z)]_{+}$, 
and hence $v(z)=S^{\tau,*}(-z)[S^{\tau}(z)v(z)]_{+}$. 
The converse is straightforward.

Now we consider the derivatives of $\mathscr L_{0,m}({\bf t})$ with respect to $t_n^a$, by equation~\eqref{eqn:diff-W-S}, we have
$$
\pd_{t^a_n}(\mathscr L_{0,m}({\bf t}))
=-\omega(D^{m+1}_{\rho,z}z^{-1}J({\bf t},-z),S^{\tau_{\rm DW},*}(-z)[S^{\tau_{\rm DW}}(z)\phi_az^n]_{+}).
$$
For an arbitrary element $f(z)\in\cH$, it is straightforward to see 
$$\textstyle
\sum_{n,a}\omega(f(z),\phi_az^n)\phi^a(-z)^{-n-1}=[f(z)]_{-}.
$$
Then by using $S^{\tau,*}(-z)[S^{\tau}(z)\phi_az^n]_{+}=\phi_az^n-S^{\tau,*}(-z)[S^{\tau}(z)\phi_az^n]_{-}$
and 
$$
\omega(f(z),S^{\tau,*}(-z)[S^{\tau}(z)\phi_az^n]_{-})=\omega(S^{\tau,*}(-z)[S^{\tau}(z)f(z)]_{+},\phi_az^n),
$$
one can see
$$\textstyle
\sum_{n,a}\omega(f(z),S^{\tau,*}(-z)[S^{\tau}(z)\phi_az^n]_{+})\phi^a(-z)^{-n-1}=f(z)-S^{\tau,*}(-z)[S^{\tau}(z)f(z)]_{+}.
$$
In particular, take $f(z)=-D^{m+1}_{\rho,z}z^{-1}J({\bf t},-z)$, we get
	$$\textstyle
	\sum_{n,a}\pd_{t^a_n}(\mathscr L_{0,m}({\bf t}))(-z)^{-n-1}\phi^a
	=S^{*}(-z)[S(z)D^{m+1}_{\rho,z}z^{-1}J({\bf t},-z)]_{+}-D^{m+1}_{\rho,z}z^{-1}J({\bf t},-z).
	$$
Therefore, $\mathscr L_{0,m}({\bf t})$ is a constant 
if and only if $D^{m+1}_{\rho,z}z^{-1}J({\bf t},-z)\in T\cL$
and by Proposition~\ref{prop:Lagrangian}, this is equivalent to $D^{m+1}_{\rho,z}z^{-1}J^{\tau}(-z)\in T\cL$.
	The proof is finished.
\end{proof}
\begin{remark}
 If $\dvac^{\tau}(z)\in \Frob[z]$, then we have $z^{-1}J({\bf t},-z)\in T\cL$ and $\mathscr L_{0,m}(\bf t)=0$ for all $m\geq -1$ according to the formula~\eqref{eqn:L0m}. 
In particular, we have the string equation $\mathscr L_{0,-1}(\bf t)=0$. 
Similarly, we have dilaton equation 
	$$\textstyle
	\sum_{k,a}\tilde t_k^a\frac{\pd \cF_0({\bf t})}{\pd t_k^a}=2\cF_0({\bf t}),
	$$
	which is equivalent to the geometric formulation: $J({\bf t},-z)\in T\cL$.
	Givental's discussion~\cite[Theorem 1]{Giv04} then applies to this case: $\cL$ is a Lagrangian cone with the vertex at the origin and such
	that its tangent spaces $T\cL$ are tangent to $\cL$ exactly along $z T\cL$.
\end{remark}

\section{Results on Virasoro constraints}
\label{sec:vira-proof}
In this section, we prove Theorem~\ref{thm:vira-anc-intro} and Theorem~\ref{thm:vira-des-intro}, the main results on Virasoro conjectures.

\subsection{Ancestor Virasoro constraints}
In this subsection, we prove the following theorem:
\begin{theorem}\label{thm:anc-vira}
(1).\, The genus-0 ancestor Virasoro conjecture always holds.\\
(2).\, For each $m\geq 0$, the genus-1 $L^{\anc}_{m}$-constraint is equivalent to the following equation: 
\beq\label{eqn:vira-genusone}
\<E^{m+1}\>^{\tau}_{1,1}
= -\frac{1}{4}\sum_{a+b=m}\str(\E^a\mu \E^b\mu)
-\frac{1}{24}\sum_{a+b=m}\str((\E^a\mu\E^b{\bf 1})*) 
+\frac{1}{24}\str((\E^{m}((\mu+\tfrac{\delta}{2}){\bf 1})*).
\eeq
(3).\, The ancestor Virasoro constraints hold for semisimple homogeneous CohFTs. 
\end{theorem}
\begin{remark}\label{rmk:vacuum-vira}
As noted in Remark~\ref{rmk:semicohft-vacuum}, for semisimple cases, we need not assume the vacuum axiom. Given a semisimple homogeneous CohFT \(\Omega^{\tau}\), \(\Omega_{g,n}^{\tau}=\widetilde\Omega_{g,n}^{\tau}\) for \(n > 0\). Here, \(\widetilde\Omega^{\tau}=R^{\tau}\cdot T^{\tau}\cdot (\oplus_{\alpha = 1}^{N}\Omega^{\rm KW_\alpha})\) is constructed through \(R\)- and \(T\)-actions and satisfies the vacuum axiom.
This indicates a relation between the corresponding total ancestor potentials \(\cA^{\tau}({\bf s};\hbar)\) and \(\widetilde\cA^{\tau}({\bf s};\hbar)\): \(\cA^{\tau}({\bf s};\hbar)=e^{F(\tau;\hbar)}\cdot\widetilde\cA^{\tau}({\bf s};\hbar)\) for a function \(F(\tau;\hbar)\) which is independent of \({\bf s}\).
Consequently, the Virasoro constraints for \(\cA^{\tau}({\bf s};\hbar)\) are equivalent to those for \(\widetilde\cA^{\tau}({\bf s};\hbar)\).
\end{remark}
\begin{proof}[Proof of Theorem~\ref{thm:anc-vira}]
The first two parts of Theorem~\ref{thm:anc-vira} are indeed derived from the topological recursion relations of genus-0 and genus-1 and we put the proof in Appendix~\ref{sec:app-anc-vira}.

Now we prove the last part of Theorem~\ref{thm:anc-vira} by using the Givental--Teleman reconstruction theorem~\cite{Giv01a,Tel12}.
We just need to prove the results for CohFT with vacuum and then the results for general cases follow.
By semisimplicity, we have a decomposition $\Frob=\oplus_{\alpha=1}^{N}{\mathbb C}\bar e_\alpha$ of $H=T_{\tau}U$ with a basis $\{\bar e_\alpha\}$ such that $\eta(\bar e_\alpha,\bar e_\beta)=\delta_{\alpha,\beta}$ and $\bar e_\alpha*_{\tau}\bar e_\beta=\delta_{\alpha,\beta}\Delta_{\alpha}^{\frac{1}{2}}\bar e_\alpha$ for some functions $\Delta_{\alpha}=\Delta_{\alpha}(\tau)$.
We define $e_\alpha=\Delta_{\alpha}^{-\frac{1}{2}}\bar e_\alpha$.
The reconstruction theorem gives
\beq\label{eqn:A-KW}
\cA^{\tau}({\bf s};\hbar)=\widehat{T_{\vac}}\widehat{R^{\tau}}\widehat{\Delta}\cD_N^{\rm KW}({\bf s};\hbar).
\eeq
The formula~\eqref{eqn:A-KW} is explained as follows. 
Firstly, $\cD_N^{\rm KW}({\bf t};\hbar)=\prod_{\alpha}\cD^{\rm KW}({\bf t}^{\bar\alpha};\hbar)$ is a product of $N$ copies of the Witten--Kontsevich tau-function. 
Secondly, $\widehat\Delta\cD_N^{\rm KW}({\bf t};\hbar)=\prod_{\alpha}\cD^{\rm KW}(\Delta_{\alpha}^{\frac{1}{2}}{\bf t}^{\bar\alpha};\Delta_{\alpha}^{\frac{1}{2}}\hbar)$
and we denote this by $\cD^{\rm top}({\bf t},\hbar)$. 
Thirdly, $R^{\tau}(z)=\id+R^{\tau}_1z+\cdots\in \End(\Frob)[[z]]$ 
is the $R$-matrix determined by equation~\eqref{eqn:hom-R}.
It follows that $r(z):=\log(R^{\tau}(z))$ is an infinitesimal transformation on the loop space $\cH$ and thus $\widehat{r(z)}$ is defined by ~\eqref{eqn:quantization-A} via identification ${\bf q}(z)={\bf t}(z)-z{\bf 1}$ with ${\bf 1}=\sum_\alpha e_\alpha$.
The operator $\widehat{R^{\tau}}$ is defined by $e^{\widehat{r(z)}}$ 
and $\widehat{R^{\tau}}\cD^{\rm top}({\bf s};\hbar)$ means we replaces the coordinate ${\bf t}$ in the function $e^{\widehat{r(z)}}\cD^{\rm top}$ with ${\bf s}$.
Lastly, the action of $\widehat{T_{\vac}}$ is just a shift on the coordinate ${\bf s}$ which trans ${\bf s}(z)-z{\bf 1}$ to ${\bf s}(z)-z\vac^{\tau}(z)$.

The idea of the proof of $L^{\anc}_m\cA^{\tau}=0$ for a semisimple CohFT is as follows: we first consider how the Virasoro operator $L^{\anc}_m$ go across the operators $\widehat{T_{\vac}}$, $\widehat{R^{\tau}}$ and $\widehat{\Delta}$, then we prove the result constraints for $\cD_{N}^{\rm KW}$ by the Witten--Kontsevich theorem.
For the results, we have
\begin{itemize}
	\item $L^{\anc}_m\widehat{T_{\vac}}=\widehat{T_{\vac}}\bar L^{\anc}_{m}$, where $\bar L^{\anc}_{m}=L^{\anc}_{m}|_{{\bf s}(z)\to {\bf s}(z)+z\vac^{\tau}(z)-z{\bf 1}}$;
	\item $\bar L^{\anc}_{m}\widehat{R^{\tau}}=\widehat{R^{\tau}}L^{\rm top,\tau}_{m}$, where $L^{\rm top,\tau}_{m}=\widehat{\ell^{\rm top,\tau}_{m}}+\frac{m+1}{16}\tr(\E^m)$ with
	$$
	\ell^{\rm top,\tau}_{m}=-z^{-1/2}(z\pd_z z+\E)^{m+1}z^{-1/2}.
	$$
	\item $L^{\rm top,\tau}_{m}\widehat\Delta=\widehat\Delta  L_{m}^{N\rm KW,\tau}$ where $L_{m}^{N\rm KW,\tau}=\sum_{\alpha}L_{m}^{\rm KW_{\alpha},\tau}$ with $L_{m}^{\rm KW_{\alpha},\tau}$ the ancestor Virasoro operator for the trivial CohFT $\Omega^{\rm KW_{\alpha}}_{g,n}(\bar e_\alpha,\cdots,\bar e_\alpha)=1$.
	\item $L_{m}^{\rm KW_{\alpha},\tau}\cD_N^{\rm KW}({\bf t};\hbar)=0$ for each $\alpha$.
\end{itemize}
Since the first and third results follow immediately from the definition of $\widehat{T_{\vac}}$ and $\widehat{\Delta}$, 
the last result follows from the Witten--Kontsevich theorem~\footnote{More strictly, $L_{m}^{\rm KW_{\alpha},\tau}\cD_N^{\rm KW}({\bf t};\hbar)=\prod_{\beta\ne \alpha}\cD^{\rm KW}({\bf t}^{\bar\beta};\hbar)\cdot L_{m}^{\rm KW_{\alpha},\tau}\cD^{\rm KW}({\bf t}^{\bar\alpha};\hbar)$, and the Witten--Kontsevich theorem, together with Dijkgraaf--Verlinde--Verlinde's work~\cite{DVV91}, gives us 
$L_{m}^{\rm KW_{\alpha},\tau}\cD^{\rm KW}({\bf t}^{\bar\alpha};\hbar)|_{\tau=0}=0$. It is a simple consequence of the Witten--Kontsevich theorem that the equation also holds for arbitrary $\tau$.},
we show details for the proof of the second equation in the follows.

The central part of the proof is to compute $e^{-\widehat{r(z)}}\widehat{\ell^{\tau}_{m}}e^{\widehat{r(z)}}$ which by definition equals
\beq\label{eqn:Ri-vira-R}
-\tfrac{1}{2}\, e^{-\widehat{r(z)}}\omega(D_{\E,z}^{m+1}z^{-1}f(z),f(z))^{\widehat{}}
\, e^{\widehat{r(z)}},
\eeq
where $f(z)=\sum_{k\geq 0}\sum_\alpha q^\alpha_k e_\alpha z^k+\sum_{k\geq 0}\sum_{\alpha}p_{k,\alpha}e^{\alpha}(-z)^{-k-1}$.
Here $\{e^\alpha\}$ is the dual basis of $\{e_\alpha\}$ with respect to $\eta$.
We compute the expression~\eqref{eqn:Ri-vira-R} in four steps.
Firstly, introduce the quantization on linear functions of ${\bf p,q}$ by
$$
\widehat{q_k^\alpha}:=q_k^\alpha/\hbar,\qquad
\widehat{p_{k,\alpha}}:=\hbar\,\pd_{q_k^\alpha},
$$
we see $\widehat{f}(z)=\widehat{f_{+}}(z)+\widehat{f_{-}}(z)
=\sum_{k\geq 0}\sum_\alpha \frac{q^\alpha_k}{\hbar} e_\alpha z^k
+\sum_{k\geq 0}\sum_{\alpha}\hbar\pd_{q_{k}^\alpha}\, e^{\alpha}(-z)^{-k-1}$,
where by notation $\widehat{f}(z)$ (as well as $\widehat{f_{\pm}}(z)$) we mean we do not quantize $z^{m}$ in $f(z)$.
Then we have
\begin{align*}
	\tfrac{1}{2}\, \omega(D_{\E,z}^{m+1}z^{-1}f(z),f(z))^{\widehat{}}
	=&\, \tfrac{1}{2}\, \omega(D_{\E,z}^{m+1}z^{-1}\widehat{f_{+}}(z),\widehat{f_{+}}(z))
	+\omega(D_{\E,z}^{m+1}z^{-1}\widehat{f_{+}}(z),\widehat{f_{-}}(z)) \\
	&\, +\tfrac{1}{2}\, \omega(D_{\E,z}^{m+1}z^{-1}\widehat{f_{-}}(z),\widehat{f_{-}}(z)).
\end{align*}
Secondly, by the definition equation~\eqref{eqn:quantization-A} and by direct computation, we have
$$
\widehat{r(z)}=\sum_{m\geq 1}\bigg(-\sum_{k\geq 0}(r_{m})^{\beta}_{\alpha}q^\alpha_k\pd_{q^{\beta}_{k+m}}
+\frac{\hbar^2}{2}\sum_{k+l=m-1}(-1)^{k}\eta^{\alpha\sigma}(r_m)^{\beta}_{\sigma}\pd_{q^{\alpha}_{k}}\pd_{q^{\beta}_{l}}\bigg).
$$
Then by this explicit formula, it is straightforward to get
$[\widehat{r(z)},\widehat{f}(z)]=-{\widehat{r\cdot f}}(z)$,
and it follows immediately
\beq\label{eqn:Ri-f-R}
e^{-\widehat{r(z)}}\widehat{f}(z)e^{\widehat{r(z)}}
={\widehat{R^{\tau} f}}(z),
\eeq
where $R^{\tau}f(z)=R^{\tau}(z)f(z)$.
Hence we get the formula:
\begin{align*}
	-e^{-\widehat{r(z)}}\widehat{\ell^{\tau}_{m}}e^{\widehat{r(z)}}
	=&\, \tfrac{1}{2}\, \omega(D_{\E,z}^{m+1}z^{-1}\widehat{(R^{\tau} f)}_{+}(z),\widehat{(R^{\tau} f)}_{+}(z))
	+\omega(D_{\E,z}^{m+1}z^{-1}\widehat{(R^{\tau} f)}_{+}(z), \widehat{(R^{\tau} f)}_{-}(z)) \\
	&\, +\tfrac{1}{2}\, \omega(D_{\E,z}^{m+1}z^{-1}\widehat{(R^{\tau} f)}_{-}(z), \widehat{(R^{\tau} f)}_{-}(z)).
\end{align*}
Thirdly, notice that $[R^{\tau}(z)f(z)]_{-}=[R^{\tau}(z)f_{-}(z)]_{-}$, we always have
$$
\omega(D_{\E,z}^{m+1}z^{-1}\widehat{(R^{\tau} f)}_{\pm}(z), \widehat{(R^{\tau} f)}_{-}(z))
=\omega(D_{\E,z}^{m+1}z^{-1}[R^{\tau}(z) f(z)]_{\pm}, [R^{\tau}(z) f(z)]_{-})^{\widehat{}}.
$$
Therefore 
\begin{align*}
	&\, -e^{-\widehat{r(z)}}\widehat{\ell^{\tau}_{m}}e^{\widehat{r(z)}}
	-\tfrac{1}{2}\,\omega(D_{\E,z}^{m+1}z^{-1}R^{\tau}(z) f(z),R^{\tau}(z) f(z))^{\hat{}}\\
	=&\, 
	\tfrac{1}{2}\, \omega(D_{\E,z}^{m+1}z^{-1}\widehat{(R^{\tau} f)}_{+}(z), \widehat{(R^{\tau} f)}_{+}(z))
	-\tfrac{1}{2}\, \omega(D_{\E,z}^{m+1}z^{-1}[R^{\tau}(z) f(z)]_{+}, [R^{\tau}(z) f(z)]_{+})^{\widehat{}} .
\end{align*}
By the definition of $\omega$, we see the right-hand side of above equation equals
\begin{align*}
&\, \tfrac{1}{2}\eta(\E^{m+1}\widehat{(R^{\tau} f)}_{0},\widehat{(R^{\tau} f)}_{0})
-\tfrac{1}{2}\, \eta(\E^{m+1}[R^{\tau}(z) f(z)]_{0}, [R^{\tau}(z) f(z)]_{0})^{\widehat{}}
=\tfrac{1}{2}\tr(\E^{m+1}R_1).
\end{align*}
Lastly, by equation~\eqref{eqn:hom-R}, we have $[R^{\tau}_1,\E]=\mu$ and $[R^{\tau}_2,\E]=R^{\tau}_1+\mu R^{\tau}_1$, and these equations imply 
$$\textstyle
\tr(\E^{m+1}R^{\tau}_1)=-\frac{1}{2}\sum_{a+b=m}\tr(\E^a\mu\E^b\mu).
$$
Moreover, equation~\eqref{eqn:hom-R} gives $D_{\E,z}R^{\tau}(z)=R^{\tau}(z) D^{\rm top}_{\E,z}$ where $ D^{\rm top}_{\E,z}=\E+\frac{3}{2}z+z^2\pd_z$,
together with the symplectic condition $R^{\tau,*}(-z)R^{\tau}(z)=\id$, we have
$$\textstyle
	e^{-\widehat{r(z)}}\widehat{\ell^{\tau}_{m}}e^{\widehat{r(z)}}
	=-\frac{1}{2}\, \omega( (D^{\rm top}_{\E,z})^{m+1}z^{-1}f,f)
	+\frac{1}{4}\sum_{a+b=m}\tr(\E^a\mu\E^b\mu).
$$
This proves $\bar L^{\anc}_{m}\widehat{R^{\tau}}=\widehat{R^{\tau}}L^{\rm top,\tau}_{m}$.
The proof is finished.
\end{proof}
\begin{remark}
Our proof of the ancestor Virasoro conjecture for semisimple homogeneous CohFT is fundamentally grounded in Givental's work~\cite{Giv01a}. In his approach, Givental used the Virasoro constraints for the Witten–Kontsevich tau-function, the reconstruction procedure, and the Kontsevich–Manin formula to prove the (descendent) Virasoro conjecture. Our work explicitly elaborates on the intermediate step that Givental left unspecified.
\end{remark}
\begin{remark}\label{rmk:anc-mil-alx}
By exchanging $N$-copies of Virasoro operators for the Witten--Kontsevich tau-function with the reconstruction procedure, $N$-copies of formal Virasoro operators for semisimple CohFTs (not necessarily homogeneous) were constructed by Milanov~\cite{Mil14} and Alexandrov~\cite{Ale24}.
For homogeneous cases, our results show that summations of those $N$-copies of formal Virasoro operators admit concise formulae.
Furthermore, to the best of our knowledge, the ancestor Virasoro operators for non-semisimple CohFTs are new in the literature.
\end{remark}

\subsection{Generalized Virasoro constraints I: genus-0 part}\label{sec:Vira-des}
In this subsection, we prove the genus-0 generalized Virasoro constraints
by generalizing Givental's method of proving the genus-0 Virasoro constraints introduced in~\cite{Giv04} to arbitrary homogeneous CohFTs.

\begin{theorem}\label{thm:vira-des-0}
For the homogeneous CohFT with homogeneous calibrations and a Virasoro index $m_\dvac$ as defined in Definition~\ref{defnofm0}, the genus-$0$ generalized Virasoro constraints hold for $m \geq m_\dvac$:
$$\mathscr L_{0,m}({\bf t}) + {\delta_{m,2 \virm}} \cdot  c_\dvac =0 . $$  
\end{theorem}
\begin{proof}
We note firstly by formula~\eqref{eqn:L0m}, if $\dvac^{\tau}(z)\in \Frob[z]$, then by definition $z^{-1}J^{\tau}(-z)\in T\cL$ and thus $z^{-1}J({\bf t},-z)\in T\cL$ according to Proposition~\ref{prop:Lagrangian}.
The genus-$0$ generalized Virasoro constraints follows immediately from the fact $D_{\rho,z}T\cL\subset T\cL$ and $\omega(f,g)=0$ for $f,g\in T\cL$.

Now we consider the case with $\dvac^{\tau}(z)\notin \Frob[z]$, then by assumption, 
$\delta=2\virm+3$ for some non-negative integer $\virm$.
By Proposition~\ref{prop:Lagrangian-vira}, we prove the genus-$0$ $L_m$-constraint, $m\geq \virm$, in two steps:
firstly, $D^{m+1}_{\rho,z}z^{-1}J^{\tau}(-z)\in T\cL$ (then we have $\mathscr L_{0,m}({\bf t})$ is a constant) and
secondly $\mathscr L_{0,m}({\bf t})|_{\bf t=\tau}=-\delta_{m,2\virm}\cdes$. 
We prove these step by step.

For the first step, by Proposition~\ref{prop:Lagrangian} and equations~\eqref{eqn:S-symp} and~\eqref{eqn:D-rho-E}, we just need to prove $D^{m+1}_{\E,z}\dvac^{\tau}(z)\in \Frob[z]$.
Let $\tilde D_{\E,z}=D_{\E,z}+\frac{\delta-3}{2}z=\E+(\mu+\frac{\delta}{2})z+z^2\pd_z$,
then by equation~\eqref{eqn:hom-vacuum-des}, we have
$\tilde D_{\E,z}\dvac^{\tau}(z)=E$ and thus
	\beq\label{eqn:tildeD-J-vac}
	\tilde D^{m+1}_{\E,z}\dvac^{\tau}(z)=\tilde D^{m}_{\E,z}E \in\Frob[z],\qquad \forall m\geq 0.
	\eeq
	Notice that $[\tilde D_{\E,z},z^{l}]=l z^{l+1}$, we have $\tilde D^k_{\E,z}z^{l}=z^{l}(\tilde D_{\E,z}+l z)^k$.
	Hence
	\beq\label{eqn:D-tildeD}\textstyle
	D^{m+1}_{\E,z}=\big(\tilde D_{\E,z}+\frac{3-\delta}{2}z\big)^{m+1}
	=\tilde D^{m+1}_{\E,z}+\sum_{k=0}^{m}a_{m,k}(z)\tilde D^{k}_{\E,z}
	\eeq
	for some polynomials $a_{m,k}(z)\in\mathbb Q[z]$. In particular, by direct computation, we have 
	$$\textstyle
	a_{m,0}(z)=\big(z^2\pd_z+\frac{3-\delta}{2}z\big)^{m+1}(1)=z^{m+1}\prod_{i=0}^{m}\big(\frac{3-\delta}{2}+i\big).
	$$
	Since by definition, $\virm=\frac{\delta-3}{2}\in\mathbb Z_{\geq 0}$, we have $a_{m,0}(z)=0$ for $m\geq \virm$.
Therefore,
$$\textstyle
	D^{m+1}_{\E,z}\dvac^{\tau}(z)
=\tilde D^{m}_{\E,z}E +\sum_{k=1}^{m}a_{m,k}(z)\tilde D^{k-1}_{\E,z}E
\in \Frob[z].
$$

For the second step, by equation~\eqref{eqn:hom-J} and by induction, it is straightforward to see
$$\textstyle
	D_{\rho,z}^{m+1}z^{-1}J^{\tau}(-z)=z^m\big(\prod_{i=0}^m(i-\virm- \nabla_E) \big) J^{\tau}(-z).
$$
Here the multiplications of $\nabla_E$ act on a function $f$ by $\nabla_{E}^{k+1}(f)=E(\nabla_E^{k}(f))$.
	Write 
	$$\textstyle
\prod_{i=0}^m(i-\virm- \nabla_E)
	=\sum_{k=0}^{m+1} b_{m,k}\nabla_{E}^k,$$ 
	then we have $b_{m,0}=0$ for $m\geq \virm$.
	Notice 
	$$
	\omega(z^m EJ^{\tau}(-z),\nabla_{ E}^{k-1} J^{\tau}(-z))=0,\qquad 2\leq k\leq m+1,$$ 
	(since $\omega(z^m\cdot f,g)=\omega(f,(-z)^mg)$ and $z^m\nabla_{ E}^{k-1} J^{\tau}(-z)\in T\cL$ for $2\leq k\leq m+1$), we have
	$$\textstyle
	\mathscr L_{0,m}({\bf t})|_{\bf t=\tau}=-\frac{1}{2}\sum_{k=1}^{m+1}b_{m,k}\nabla_{E}^{k-1}\big(\omega(z^m EJ^{\tau}(-z),J^{\tau}(-z))\big).
	$$
	Furthermore, it follows from the straightforward computations that
	$$
	\pd_{\tau^a}\big(\omega(z^m EJ^{\tau}(-z),J^{\tau}(-z))\big)
	=(m-2\virm)\omega(z^m\pd_{\tau^a}J^{\tau}(-z),J^{\tau}(-z)),
	$$
	thus we have
	$$\textstyle
	\mathscr L_{0,m}({\bf t})|_{\bf t=\tau} =-\frac{1}{2}\sum_{k=1}^{m+1}b_{m,k}(m-2\virm)^{k-1}\cdot\omega(z^m EJ^{\tau}(-z),J^{\tau}(-z)).
	$$
	For $m\geq \virm$ and $m\ne 2\virm$, we have
	$\sum_{k=1}^{m+1}b_{m,k}(m-2\virm)^{k}=\prod_{i=0}^{m}(i-\virm-(m-2\virm)$,
	thus $\mathscr L_{0,m}({\bf t})|_{\bf t=\tau}=0$.
	For $m=2\virm$, we have $b_{2\virm,1}=(-1)^{\virm+1}(\virm!)^2$ and
	$$\textstyle
	\mathscr L_{0,2\virm}({\bf t})|_{\bf t=\tau}
=-\frac{1}{2}b_{2\virm,1}\cdot\omega(EJ^{\tau}(-z), z^{2\virm}J^{\tau}(-z)).
	$$
Notice that $EJ^{\tau}(-z)=S^{\tau,*}(-z)E$ and $\omega(S^{\tau,*}(-z)f,g)=\omega(f,S^{\tau}(z)g)$, the r.h.s of above equation gives exactly $-\cdes$.
\end{proof}
\begin{remark}\label{rmk:inv-vira-0}
	We can prove that if the term $a_{m,0}(z)$ in the proof of Theorem~\ref{thm:vira-des-0} is not vanished, then the genus-0 $L_m$-constraint fails.
	Otherwise, let $\widetilde\dvac^{\tau}(z)=a_{m,0}(z)\dvac^{\tau}(z)$, the $L_m$-constraint, together with equation~\eqref{eqn:tildeD-J-vac}, gives us $\widetilde\dvac^{\tau}(z)\in\Frob[z]$. 
	By QDE~\eqref{eqn:QDE-vacuum-des} of the $\dvac$-vector, we have
	$$
	z\pd_{\tau^a}\widetilde\dvac^{\tau}(z)=\phi_a*_{\tau}\widetilde\dvac^{\tau}(z)-a_{m,0}(z)\phi_a.
	$$
	Let $\widetilde\dvac^{\tau}(z)=\sum_{i}\widetilde\dvac^{\tau}_i z^i$, then the equation solves $\widetilde\dvac^{\tau}_0=\cdots=\widetilde\dvac^{\tau}_m=0$ and thus
	$$
	\dvac^{\tau}(z)=\widetilde\dvac^{\tau}(z)/a_{m,0}(z)\in\Frob[z],
	$$
which contradicts our assumption.
We see the existence of the Virasoro-index $\virm$ is the necessary (and sufficient for genus-$0$) condition for the generalized Virasoro conjecture.
\end{remark}

\begin{corollary}\label{cor:vira-0-nu-v}
	For a homogeneous CohFT with vacuum $\vac^{\tau}(z)$ and calibrated with a homogeneous $S$-matrix and $\dvac$-vector $\dvac^{\tau}(z)$, the genus-0 $L_m$-constraint holds if and only if 
	$$
	D_{\E,z}^{m+1}\dvac^{\tau}(z)=D_{\E,z}^{m+1}\vac^{\tau}(z).
	$$
\end{corollary}
\begin{proof}
We note first the string equation holds if and only if $J^{\tau}(-z)\in z T\cL$ which is equivalent to the fact that $\dvac^{\tau}(z)$ a polynomial in $z$ and is further equivalent to $\dvac^{\tau}(z)=\vac^{\tau}(z)$.

For the case that $\dvac^{\tau}(z)\notin\Frob[z]$, the validity of the genus-$0$ $L_m$-constraint requires the existence of the Virasoro-index $\virm$ and $m\geq \virm\geq 0$.
Since the $\dvac$-vector and the vacuum vector satisfy the same homogeneity condition~\eqref{eqn:hom-vacuum-des} (or ~\eqref{eqn:hom-vacuum}),  equation~\eqref{eqn:tildeD-J-vac} also applies to the vacuum vector $\vac^{\tau}(z)$. 
By equation~\eqref{eqn:D-tildeD}, $D_{\E,z}^{m+1}\dvac^{\tau}(z)=D_{\E,z}^{m+1}\vac^{\tau}(z)$ follows from $a_{m,0}(z)=0$.
Conversely, the same homogeneity condition for the $\dvac$-vector and the vacuum vector, equation~\eqref{eqn:tildeD-J-vac} and equation~\eqref{eqn:D-tildeD} together give that the equation $D_{\E,z}^{m+1}\dvac^{\tau}(z)=D_{\E,z}^{m+1}\vac^{\tau}(z)$ implies $a_{m,0}(z)=0$
 (since we have assumed $\dvac^{\tau}(z)\notin\Frob[z]$ thus $\dvac^{\tau}(z)\ne \vac^{\tau}(z)$).
Notice that $a_{m,0}(z)=z^{m+1}\prod_{i=0}^{m}(\frac{3-\delta}{2}+i)$,
the vanishing of $a_{m,0}(z)$ ensures $\delta=2\virm+3$ for some integer $0\leq \virm \leq m$ and the genus-$0$ $L_m$-constraint follows.
\end{proof}

\subsection{Generalized Virasoro constraints II: higher genus part}
In this subsection, we consider the generalized Virasoro constraints of higher genus. 
We prove the second and third parts of Theorem~\ref{thm:vira-des-intro} by establishing the equivalence between (generalized) descendent and ancestor Virasoro constraints.

\begin{proposition}\label{prop:anc-des-vira}
	For $m\geq \virm$, the $L_m$-constraint is equivalent to the $L^{\anc}_m$-constraint.
\end{proposition}
\begin{proof}
	By the relation of the total descendent potential and the total ancestor potential (equation~\eqref{eqn:des-anc}), 
	the $L_m$-constraint is equivalent to the following equation for $\cA^{\tau}({\bf s};\hbar)$:
	$$
	\tilde L^{\anc}_{m}\cA^{\tau}({\bf s};\hbar)=0,
	$$
	where the operator $\tilde L^{\anc}_{m}$ is defined by
	$$
	\tilde L^{\anc}_{m}:=e^{-\frac{1}{\hbar^2}\widetilde W^{\tau}({\bf t}-\tau,{\bf t}-\tau)}\cdot (L_m+\tfrac{\delta_{m,2\virm}}{\hbar^2}\cdes)\cdot e^{\frac{1}{\hbar^2}\widetilde W^{\tau}({\bf t}-\tau,{\bf t}-\tau)}.
	$$ 
Here we use the expression~\eqref{eqn:Lm-des} of the operator $L_m$.
By viewing ${\bf t}(z)=[S^{\tau,*}(-z){\bf s}(z)]_{+}+\tau$  and $\phi_a\psi^k=[S^{\tau}(\bar\psi)\phi_a\bar\psi^k]_{+}$,
our goal is to prove $\tilde L^{\anc}_{m}=L^{\anc}_{m}$ (see equation~\eqref{eqn:Lm-anc} for the expression of the operator $L^{\anc}_m$).
By definition, $\tilde L_m^{\anc}$ has form:
$$
L^{\anc}_{m}= \tfrac{1}{2\hbar^2}B^{\tau}_m({\bf t},{\bf t})
+\tilde L_{m,0}^{\tau}+\tilde L_{m,1}^{\tau}+\tfrac{\hbar^2}{2}\tilde L_{m,2}^{\tau},
$$
where $B^{\tau}_{m}({\bf t},{\bf t})$ is a quadratic function of ${\bf t}$ and
\begin{align*}
\tilde L^{\anc}_{m,2}=&\, \sum_{k=1}^{m}(-1)^k \phi_a \psi^{k-1}\circ [D^{m+1}_{\rho,\psi}\phi^a\psi^{-k-1}]_{+},\\
\tilde L^{\anc}_{m,1}=&\,  [D^{m+1}_{\rho,\psi}\psi^{-1}\tilde{\bf t}(\psi)]_{+}
+\sum_{k=1}^{m}(-1)^k\big(\<\phi_a\psi^{k-1}\>^{\tau}_{0,1}+\<\phi_a\psi^{k-1},{\bf t}(\psi)-\tau\>^{\tau}_{0,2}\big) 
[D^{m+1}_{\rho,\psi}\phi^a\psi^{-k-1}]_{+},\\
\tilde L_{m,0}^{\tau}=&\,  -\delta_{m,0}\frac{1}{4}\str\Big(\mu^2-\frac{1}{4}\Big)
+\frac{1}{2}\sum_{k=1}^{m}(-1)^k\<\phi_a\psi^{k-1}, [D^{m+1}_{\rho,\psi}\phi^a\psi^{-k-1}]_{+}\>^{\tau}_{0,2}.
\end{align*}
In the follows, we compute these three terms one by one.
	
Firstly, by equation~\eqref{eqn:D-rho-E}, we have
	$$ \textstyle
	\sum_{k=1}^{m}(-1)^k\phi_a\psi^{k-1}\circ [D^{m+1}_{\rho,\psi}\phi^a\psi^{-k-1}]_{+}
	=\sum_{k=1}^{m}(-1)^k[S^{\tau}(\bar\psi)\phi_a\bar\psi^{k-1}]_{+}\circ [D^{m+1}_{\E,\bar\psi}S^{\tau}(\bar\psi)\phi^a\bar\psi^{-k-1}]_{+}.
	$$
By taking expansion of $S^{\tau}(\bar\psi)=\sum_{j\geq 0}S^{\tau}_j\bar\psi^{-j}$ and using 
$\phi_a\otimes (S^{\tau}_j\phi^a)=(S^{\tau,*}_j\phi_a)\otimes \phi^a$,
above equation can be further simplified as follows:
$$\textstyle
\sum_{k=1}^{m}\sum_{i,j\geq 0}(-1)^k[S^{\tau}_iS^{\tau,*}_j\phi_a\bar\psi^{k-i-1}]_{+}\circ [D^{m+1}_{\E,\bar\psi}\phi^a\bar\psi^{-k-j-1}]_{+}.
$$
Note the summation in above expression is finite, by changing the order of the summation and by using the symplectic condition~\eqref{eqn:S-symp} of the $S$-matrix, we get
$$\textstyle
\tilde L^{\anc}_{m,2}= \sum_{k=1}^{m}(-1)^k \phi_a \bar\psi^{k-1}\circ [D^{m+1}_{\E,\psi}\phi^a\bar\psi^{-k-1}]_{+}.
$$

Secondly, notice that
$$\textstyle
\tilde L^{\anc}_{m,1}= \big[D^{m+1}_{\rho,\psi}\psi^{-1}\big(J^{\tau}(-\psi)+{\bf t}(\psi)-\tau
+\sum_{k\geq 1}\<\phi_a\psi^{k-1},{\bf t}(\psi)-\tau\>^{\tau}_{0,2} \phi^a(-\psi)^{-k}\big)\big]_{+}
$$
by using equation~\eqref{eqn:diff-W-S} and by identifying $\phi_a\psi^k$ with $[S^{\tau}(\bar\psi)\phi_a\bar\psi^k]_{+}$, we see
$$
\tilde L^{\anc}_{m,1}= \big[D_{\E,\bar\psi}^{m+1}\bar\psi^{-1}S^{\tau}(\bar\psi)\big(J^{\tau}(-\bar\psi)+S^{\tau,*}(-\bar\psi)[S^{\tau}(\bar\psi)({\bf t}(\bar\psi)-\tau)]_{+}\big)\big]_{+}.
$$
Now by using equations~\eqref{def:J-function-small}, \eqref{eqn:S-symp}, \eqref{eqn:s-t}
and Corollary~\ref{cor:vira-0-nu-v}, we obtain
	$$
	\tilde L^{\anc}_{m,1}
	=\big[D_{\E,\bar\psi}^{m+1}\bar\psi^{-1}({\bf s}(\bar\psi)-\bar\psi\vac^{\tau}(\bar\psi))\big]_{+}.
	$$
	
Thirdly,  by the definition of $\<-\>^{\tau}_{0,2}$, one can see for $m\geq 1$
$$\textstyle
\tilde L_{m,0}^{\tau}=\frac{1}{2}\sum_{k\geq 1}  (-1)^k  \mathop{\Res}_{z=0} \mathop{\Res}_{w=0}
\eta (\phi_az^{k-1},W^{\tau}(z,w) D_{\rho,w} ^{m+1} \phi^aw^{-k-1})\tfrac{dz}{z} \tfrac{dw}{w}.
$$
Notice that $[D^{m+2}_{\rho,w},w^{-1}]=-(m+2) D_{\rho,w}^{m+1}$, we have
\begin{align*}
\tilde L_{m,0}^{\tau}=&\, \textstyle -\frac{1}{2m+4}\sum_{k\geq 1}  (-1)^k  \mathop{\Res}_{z=0} \mathop{\Res}_{w=0}
\eta (\phi_az^{k-1},W^{\tau}(z,w) D^{m+2}_{\rho,w}\phi^aw^{-k-2})\tfrac{dz}{z} \tfrac{dw}{w}\\
&\, \textstyle +\frac{1}{2m+4}\sum_{k\geq 1}  (-1)^k  \mathop{\Res}_{z=0} \mathop{\Res}_{w=0}
\eta (\phi_az^{k-1},W^{\tau}(z,w) w^{-1}D^{m+2}_{\rho,w}\phi^aw^{-k-1})\tfrac{dz}{z} \tfrac{dw}{w}.
\end{align*}
Recall by definition of $W^{\tau}$ (equation~\eqref{eqn:W-S}), we have
$$W^{\tau}(z,w) w^{-1}=-W^{\tau}(z,w) z^{-1}+S^{\tau,*}(z)S^{\tau}(w)-\id. $$
Substitute this into the second line on the r.h.s of above equation, then one can see the summation containing $-W^{\tau}(z,w) z^{-1}$ cancel with the first line on the r.h.s of above equation, and one gets the result:
$$\textstyle
\tilde L_{m,0}^{\tau}= \frac{1}{2m+4}\sum_{k\geq 1}  (-1)^k  \mathop{\Res}_{z=0} \mathop{\Res}_{w=0}
\eta (\phi_az^{k-1},(S^{\tau,*}(z)S^{\tau}(w)-\id)D^{m+2}_{\rho,w}\phi^aw^{-k-1})\tfrac{dz}{z} \tfrac{dw}{w}.
$$
Now by using similar method that we used to compute $\tilde L_{m,2}^{\tau}$, one can see
\begin{align*}
\tilde L_{m,0}^{\tau}= &\, \textstyle -\frac{1}{2m+4}\sum_{k\geq 1} (-1)^k  \mathop{\Res}_{z=0} \mathop{\Res}_{w=0}
\eta (\phi_az^{k-1},D^{m+2}_{\rho,w}\phi^aw^{-k-1})\tfrac{dz}{z} \tfrac{dw}{w}\\
&\, \textstyle +\frac{1}{2m+4}\sum_{k\geq 1} (-1)^k  \mathop{\Res}_{z=0} \mathop{\Res}_{w=0}
\eta (\phi_az^{k-1},D^{m+2}_{\E,w}\phi^aw^{-k-1})\tfrac{dz}{z} \tfrac{dw}{w}.
\end{align*}
This gives
$$
\tilde L_{m,0}^{\tau}=  \textstyle \frac{1}{2m+4} \str \big([w^0]D^{m+2}_{\rho,w}(w^{-2})\big)
 - \frac{1}{2m+4} \str \big([w^0]D^{m+2}_{\E,w}(w^{-2})\big).
$$
By the definitions of $D_{A,z}$ (equation~\eqref{def:DAz}), we have
$$
\str \big([w^0]D^{m+2}_{A,w}(w^{-2})\big)
=\sum_{i+j+k=m}\str\big(A^i(\mu+\tfrac{1}{2})A^j(\mu-\tfrac{1}{2})A^k\big)
=\frac{m+2}{2}\sum_{i+j=m}\str\big(A^i(\mu+\tfrac{1}{2})A^j(\mu-\tfrac{1}{2})\big)
$$
When taking $A=\rho$, notice that $\rho$ is nilpotent and $\mu$ is diagonal under flat basis, above expression vanishes for $m>0$, thus we have
$$\textstyle
\tilde L_{m,0}^{\tau}=   -\frac{1}{4}\sum_{i+j=m}\str\big(\E^i(\mu+\tfrac{1}{2})\E^j(\mu-\tfrac{1}{2})\big)
$$
and the formula also holds for $m=0$.
	
Now we know $\tilde L_m^{\tau}$ has the following form:
	\begin{align*}
		\tilde L^{\anc}_m=&\, \textstyle \frac{1}{2\hbar^2}B^{\tau}_{m}({\bf t(s)},{\bf t(s)})
		-\frac{1}{4} \sum_{a+b=m} \str(\E^a(\mu+\tfrac{1}{2})\E^b(\mu-\tfrac{1}{2})) \\
		&\, \textstyle +[D^{m+1}_{\E,\bar\psi}\psi^{-1}\tilde{\bf s}(\psi)]_{+}
		+\frac{\hbar^2}{2}\sum_{k=1}^{m}(-1)^{k}\phi_a\bar\psi^{k-1}\circ
		[D^{m+1}_{\E,\bar\psi}\phi^a\bar\psi^{-k-1}]_{+}.
	\end{align*}
	By genus-0 $L_m$-constraint, we have
	$$\textstyle
	\frac{1}{2}B^{\tau}_{m}({\bf t(s)},{\bf t(s)})-\frac{1}{2}\<D^{m+1}_{\E,\bar\psi}\vac^{\tau}(\bar\psi),s_0,s_0\>_{0,3}^{\tau} +O({\bf s}^3)=0.
	$$
	This gives $B^{\tau}_{m}({\bf t(s)},{\bf t(s)})=\eta(\E^{m+1}s_0, s_0)$, and we obtain $\tilde L^{\tau}_m=L^{\anc}_m$.
\end{proof}

Now we can prove the second and third parts of Theorem~\ref{thm:vira-des-intro}.
\begin{theorem}\label{thm:vira-des-g}
We have the following results:

	(1). For each $m\geq \max\{\virm,0\}$, the genus-1 $L_m$-constraint is equivalent to equation~\eqref{eqn:vira-genusone-0}.

	(2). The generalized Virasoro conjecture, as stated in Conjecture \ref{conj:des-vira}, holds for any semi-simple homogeneous CohFT  with  homogeneous calibrations and a Virasoro index.
\end{theorem}
\begin{proof}
This is a direct consequence of Theorem~\ref{thm:anc-vira} and Proposition~\ref{prop:anc-des-vira}.
\end{proof}

\section{Application I: Virasoro constraints of deformed negative $r$-spin theory}
\label{sec:negative-r-spin}
The deformed negative $r$-spin class $\Theta^{r,\epsilon}$, also called the $\epsilon$-deformed $\Theta$-class, was introduced by Norbury~\cite{Nor23} for $r=2$, $\epsilon=0$ and generalized by Chidambaram, Garcia-Failde and Giacchetto~\cite{CGG22}
for arbitrary $r{\geq 2}$ and $\epsilon$.
In~\cite{GJZ23}, Guo, Ji and Zhang introduced the geometric descendent~\footnote{In~\cite{CGG22}, the generating series of the class $\Theta^{r,\epsilon}$ is called the ``descendant potential", it is called the ancestor potential in this paper. When taking $\epsilon=0$, the descendent potential equals the ancestor potential.} 
invariants of the deformed negative $r$-spin class and proved a Kontsevich--Manin type formula that relates the descendent invariants and ancestor invariants of the deformed negative $r$-spin theory.
In this section, we deduce the Virasoro constraints of the deformed negative $r$-spin theory as an application of Theorem~\ref{thm:vira-des-intro}.

\subsection{Definition and descendent potentials} \label{sec:defrspin}
Let $r\geq 2$ be a fixed positive integer, $s$ be an integer (we will focus on $s=-1$ cases), and $0\leq a_1,\cdots, a_n \leq r-1$ be integers satisfying
$$
\tilde{ D}_{g,n}^{r,s}( { \vec{a}}) := (2g-2+n)\cdot s - \textstyle \sum_{i=1}^n a_i \in   r \mathbb Z.
$$
We introduce the proper moduli space of twisted stable $r$-spin curves  	
$$
\Mbar^{r,s}_{g,\vec{a}} = \{(C_g,p_1,\cdots,p_n,   L) :   \    L^{r} \cong \omega_{\log}^{s}(-\textstyle \sum_{i=1}^n  a_i [p_i]) \},
$$
where $\omega_{\log} = \omega( \sum_{i=1}^n  [p_i])$.  
We see $\deg \omega_{\log}^{s}(-\textstyle \sum_{i=1}^n  a_i [p_i])=\tilde  { D}_{g,n}^{r, s}( { \vec{a}})  $
and  the Riemann--Roch theorem gives
$$
{ D}_{g,n}^{r,s}( { \vec{a}}):=\dim H^1(C_g, { L}) -\dim H^0(C_g,{ L})   =  -\deg { L}+g-1 
=  -\tfrac{1}{r} \tilde  { D}_{g,n}^{r,s}( { \vec{a}}) +g-1
$$

Let $\mathcal C_{g,\vec{a}}^{r,s}$ be the universal curve of the moduli space $\Mbar^{r,s}_{g,\vec{a}}$ and let $\mathcal  L_{g,\vec{a}}^{r,s}$ be the universal bundle on it, we have the morphisms
$$
\mathcal C_{g,\vec{a}}^{r,s} \xlongrightarrow{\pi}  \Mbar_{g,\vec{a}}^{r,s} \xlongrightarrow{ p }  \Mbar_{g,n}.
$$
It is known that for $-(r-1)\leq s\leq -1$, $\mathrm R^0 \pi_* \mathcal L_{g,\vec{a}}^{r,s}$ vanishes.   
Following  \cite{CGG22,Ch08},  we consider the vector bundle of rank $ D_{g,n}^{r,s}( { \vec{a}})$ :
\beq
\mathcal V^{r,s}_{g,\vec{a}} :=  \mathrm R^1 \pi_* \mathcal L_{g,\vec{a}}^{r,s}
\eeq
and introduce the following twisted class
$$
c_{\rm top}( \mathcal V^{r,s}_{g,\vec{a}} )  \in H^{  D_{g,n}^{r,s}( { \vec{a}})}(\M_{g,\vec{a}}^{r,s}) .
$$
{where $c_{\rm top}( \mathcal V^{r,s}_{g,\vec{a}} ) $ means the top Chern class of $\mathcal V^{r,s}_{g,\vec{a}} $.}

Now we focus on $s=-1$ cases.
Let $\Frob=\sspan_{\mathbb Q}(\phi_1\cdots,\phi_{r-1})$ be the state space, where $\{\phi_{a}\}_{a=1}^{r-1}$ are vectors associated to the integers $1\leq a\leq r-1$ (we slightly shift the sub-index compared with the general setting since $\phi_0$ has special meaning here). Let $\{\psi_i\}_{i=1}^n$ be the psi-classes, denoting the first Chern class of the universal cotangent line bundle over $\M_{g,\vec{a}}^{r,-1}$ with respect to the $i$-th marked point.
The descendent invariants / correlators for the $\epsilon$-deformed negative $r$-spin theory are defined by
$$
\<  \phi_{a_1}\psi_1^{k_1},\cdots, \phi_{a_n}\psi_n^{k_n}\>_{g,n}^{r,\epsilon}
:=\sum_{m\geq 0}\frac{\epsilon^m}{m!}\frac{(-1)^{D^{r,-1}_{g,n+m}(\vec{a}+\vec{0}_{m})}}{r^{g-1}}
\int_{\M_{g,\vec{a}+\vec{0}_{m}}^{r,-1}}  c_{\rm top}( \mathcal V^{r,-1}_{g,\vec{a}+\vec{0}_{m}} )\prod_i  \psi_i^{k_i},
$$
where $\vec{a}+\vec{0}_m=(a_1,\cdots,a_n,0,\cdots,0)$.
We define the genus-$g$ descendent potential $\cF_g^{r,\epsilon}(\mathbf t)$ and total descendent potential $\cD^{r,\epsilon}(\mathbf t;\hbar) $ for the deformed negative $r$-spin theory as follows:
$$\textstyle
\cF_g^{r,\epsilon}(\mathbf t) :=  \sum_{n\geq 0} \frac{1}{n!} \<\mathbf t(\psi_1),\cdots,\mathbf t(\psi_n)\>_{g,n}^{r,\epsilon},\qquad
\cD^{r,\epsilon }( \mathbf t;\hbar) = e^{\sum_{g \geq 0} \hbar^{2g-2} \cF^{r,\epsilon}_g(\mathbf t)} ,
$$
where $\mathbf t(\psi) = \textstyle \sum_{k\geq 0,1\leq a \leq r-1}  t_k^a\  \phi_{a} \psi^k \in \Frob[[\psi]]$.

\subsection{CohFT and ancestor potentials}
By considering the push-forward of top Chern class of $ \mathcal V^{r,-1}_{g,\vec{a}}$ along the forgetful maps 
$\Mbar_{g,\vec{a}+\vec{0}_m}^{r,-1}\xlongrightarrow{p}\Mbar_{g,n+m}\xlongrightarrow{\pi^m}\Mbar_{g,n}$,
we get the so called $\epsilon$-deformed $\Theta$-class as follows
\beq\label{def:-rspinclass}
\Theta^{r,\epsilon}_{g,n}(\phi_{a_1},\cdots,\phi_{a_n}) := \frac{1}{r^{g-1}}\sum_{m\geq 0}
\frac{\epsilon^m}{m!}\pi^m_{*}p_*  \big((-1)^{\deg}\cdot c_{\rm top}\big( \mathcal V^{r,-1}_{g,\vec{a}+\vec{0}_{m}} \big)\big).
\eeq
It is proved in~\cite{CGG22} that the collection $\{\Theta^{r,\epsilon}_{g,n}\}$ satisfies the axioms of an $(r-1)$-dimensional CohFT on $\Frob$ with symmetric bilinear form $\eta(\phi_a,\phi_b)=\delta_{a+b,r}$.
Moreover, this CohFT is semisimple (for $\epsilon\ne0$) and homogeneous with respect to the Euler vector field
$$\textstyle
E=(r-1)\phi_{r-1}-\sum_{a=1}^{r-1}\frac{a}{r}\tau^a\phi_a.
$$
The conformal dimension $\delta$ of $\Theta^{r,\epsilon}$ is $3$. 
The Euler vector field defines the operator $\mu$ which has formula $\mu(\phi_a)=(\frac{a}{r}-\frac{1}{2})\phi_a$, $a=1,\cdots,r-1$.

By using the Chiodo formula~\cite{Ch08}, the quantum product $*$ of $\Theta^{r,\epsilon}$ has the following explicit formula (see~\cite{CGG22,GJZ23} for more details):
if $a+b=(r-1)m+c$ where $0\leq c\leq r-2$, then
$$
\phi_a*\phi_b=(\tfrac{\epsilon}{r})^{m}\phi_{c+1}.
$$
Especially, notice that $E|_{\tau=0}=(r-1)\phi_{r-1}$, we have
\beq\label{eqn:qpEuler-rspin}
E|_{\tau=0}*\phi_a=\begin{cases}
(r-1)\cdot \frac{\epsilon}{r}\cdot\phi_{a+1} & a=1,\cdots,r-2\\
(r-1)\cdot \frac{\epsilon^2}{r^2} \cdot \phi_1 & a=r-1
\end{cases}.
\eeq
By using the homogeneity condition~\eqref{eqn:hom-vacuum-2}, one can compute the vacuum vector at $\tau=0$:
$$
\vac^{r,\epsilon}(z)=r\sum_{a=1}^{r-2}\phi_a\sum_{m\geq 1}\frac{(mr-1-a)!}{(m-1)!}\frac{(-z)^{m(r-1)-a-1}}{\epsilon^{mr-a-1}}
+r\phi_{r-1}\sum_{m\geq 1}\frac{(mr)!}{m!}\frac{(-z)^{m(r-1)-1}}{\epsilon^{mr}}.
$$
\begin{remark}
For the CohFT $\Theta^{r,\epsilon}$, the vacuum axiom~\eqref{eqn:vacuum-axiom} fails for $n=0$, however, the axiom always holds for $n\geq 1$ as we have explained in Remark~\ref{rmk:semicohft-vacuum} and this will not affect the validity of Virasoro constraints as we have explained in Remark~\ref{rmk:vacuum-vira}.
\end{remark}

Following the general setting, we define the genus-$g$ ancestor potential $\bar\cF_g^{r,\epsilon}(\mathbf s)$ and total ancestor potential $\cA^{r,\epsilon}(\mathbf s;\hbar) $ for the deformed negative $r$-spin theory as follows:
$$\textstyle
\bar\cF_g^{r,\epsilon}(\mathbf s) :=  \sum_{n\geq 0} \frac{1}{n!} \int_{\Mbar_{g,n}}\Theta^{r,\epsilon}_{g,n}({\bf s}(\psi_1),\cdots,{\bf s}(\psi_n)),\qquad
\cA^{r,\epsilon}(\mathbf s;\hbar) = e^{\sum_{g \geq 0} \hbar^{2g-2} \bar\cF^{r,\epsilon}_g(\mathbf s)} ,
$$
where $\mathbf s(\psi) = \textstyle \sum_{k\geq 0,1\leq a \leq r-1}  s_k^a\  \phi_{a} \psi^k \in \Frob[[\psi]]$. 

\subsection{$S$-matrix, $\dvac$-vector and Kontsevich--Manin type formula}
It was proved in~\cite{GJZ23} that the total descendent potential $\cD^{r,\epsilon}({\bf t};\hbar)$ and total ancestor potential $\cA^{r,\epsilon}({\bf s};\hbar)$ are related by the following Kontsevich--Manin type formula:
$$
\cD^{r,\epsilon}({\bf t};\hbar)=e^{\frac{1}{2\hbar}\widetilde W^{r,\epsilon}({\bf t},{\bf t})}\cA^{r,\epsilon}({\bf s(t)};\hbar),
$$
where $\widetilde W^{r,\epsilon} $ is given by equation~\eqref{def:tildeW} at $\tau=0$ with $\Phi(0)=0$.
Moreover, the $S$-matrix $S^{r,\epsilon}(z)$ and the $J$-function (which has form $J^{r,\epsilon}(-z)=-r\phi_{r-1}+\sum_{k\geq 0}J^{r,\epsilon}_{k,a}\phi^a(-z)^{-k-1}$) are explicitly computed in~\cite{GJZ23}, the formulae are
$$
(S^{r,\epsilon}_k)_{a}^b=\left\{\begin{array}{cc}
	(-1)^{k+m}\frac{\Gamma(\frac{a}{r})}{\Gamma(\frac{a}{r}+k-m)}\frac{\epsilon^{m}}{m!},
	& m=\frac{rk+a-b}{r-1}\in \mathbb Z_{+}\\
	0 , & {\rm otherwise}
\end{array}\right. ,
$$
and
$$
J^{r,\epsilon}_{k,a}
=\left\{\begin{array}{cc}
	(-1)^{k+m}\frac{\Gamma(\frac{a}{r})}{\Gamma(\frac{a}{r}+k-m)}\frac{\epsilon^{m+2}}{(m+2)!} & m=\frac{rk+a}{r-1}-1\in\mathbb Z_{+}\\
	0 & {\text{otherwise}}
\end{array}\right..
$$
As noted in Remark~\ref{rmk:J-nu}, given $S$-matrix and $J$-function, one get  the $\dvac$-vector $\dvac^{r,\epsilon}(z)$ by $\dvac^{r,\epsilon}(z)=-z^{-1}S^{r,\epsilon}(z)J^{r,\epsilon}(-z)$. 
Precisely,
$$
\nu^{r,\epsilon}(z)=\sum_{a=1}^{r-2}\phi_a\sum_{m\geq 0}\frac{m!}{(mr+a)!}\frac{\epsilon^{mr+a+1}}{z^{m(r-1)+a+1}}
+r\phi_{r-1}\sum_{m\geq 0}\frac{m!}{(rm)!}\frac{\epsilon^{rm}}{z^{m(r-1)+1}}.
$$
It is straightforward to check the $S$-matrix (resp. the $\dvac$-vector) satisfies the homogeneity condition~\eqref{eqn:hom-S-2} (resp. ~\eqref{eqn:hom-vacuum-des-2}), where $\rho=0$.

\subsection{Virasoro constraints}
We introduce and prove the Virasoro constraints for the deformed negative $r$-spin theory.
For the total ancestor potential $\cA^{r,\epsilon}({\bf s};\hbar)$, since the quantum product of $E$ at the original point $\tau=0$, the grading operator $\mu$ and the vacuum vector $\vac^{r,\epsilon}(z)$ are all explicitly computed, one gets immediately the ancestor Virasoro constraints of $\cA^{r,\epsilon}({\bf s};\hbar)$ by Theorem~\ref{thm:vira-anc-intro}.
Now we consider the descendent Virasoro constraints, which are notable for their concise formulae, as presented in the following Proposition:

\begin{proposition} 
\label{prop:vira-negative-rspin}
	Let $\cD^{r,\epsilon}(\mathbf t;\hbar)$ be the total descendent potential of the deformed negative $r$-spin theory defined in  {\normalfont\S \ref{sec:defrspin}}, then it satisfies the following Virasoro constraints
	$$
	L^{r,\epsilon}_{m}\cD^{r,\epsilon}(\mathbf t;\hbar)=0, \qquad  m\geq 0,
	$$
	where operators $L^{r,\epsilon}_{m}$, $m\geq 0$, satisfy the commutation relation 
$[L^{r,\epsilon}_{m},L^{r,\epsilon}_{n}]=(m-n)L^{r,\epsilon}_{m+n}$ for $m,n\geq 0$ 
and have the following explicit formulae:
\begin{align*}
L^{r,\epsilon}_m=&\, \delta_{m,0}\frac{(r-1)\epsilon^2}{2\, \hbar^2}+\delta_{m,0}\frac{r^2-1}{24\, r}
+\sum_{k\geq 0}\sum_{a=1}^{r-1}\frac{\Gamma(m+1+k+\frac{a}{r})}{\Gamma(k+\frac{a}{r})}\tilde t_k^{a}\frac{\pd }{\pd t_{k+m}^{a}}\\
&\, +\frac{\hbar^2}{2}\sum_{k=1}^{m}(-1)^k\sum_{a=1}^{r-1}
\frac{\Gamma(m+1-k+\frac{a}{r})}{\Gamma(-k+\frac{a}{r})}
\frac{\pd^2 }{\pd t_{k-1}^{r-a}\pd t_{m-k}^{a}}.
\end{align*}
Here the shift on the time variables is given by  $\tilde t_k^a=t_{k}^{a}-\delta_{k,0}\delta_{a,r-1}\cdot r$.
\end{proposition}
\begin{proof}
We first mention that $L^{r,\epsilon}_{m}$ is just the deformed negative $r$-spin case of $L_m+\frac{\delta_{m,0}}{\hbar^2}\cdes$ in Theorem~\ref{thm:vira-des-intro}.
We have seen that the deformed negative r-spin class is homogeneous and semisimple.
Moreover, the $\dvac$-vector is not a polynomial in $z$ and the conformal dimension of the CohFT is $3$, thus the Virasoro-index $\virm=0$.
By Theorem~\ref{thm:vira-des-intro}, the $L_m$-constraint holds for $m\geq 0$.
The explicit formula for $L^{r,\epsilon}_{m}$ follows from 
the following computations: firstly,
$$\textstyle
-\frac{1}{4}\tr(\mu^2-\frac{1}{4})=\frac{1}{r^2}\sum_{a=1}^{r-1}a(r-a)=\frac{r^2-1}{24\, r},
$$
secondly, for $k\in \mathbb Z$, $a=1,\cdots,r-1$, 
$$
D_{\rho,z}^{m+1}\phi_az^k=\tfrac{\Gamma(m+2+k+\frac{a}{r})}{\Gamma(k+1+\frac{a}{r})}\phi_az^{m+1+k},
$$
(recall $\rho=0$), and lastly, $\cdes=\frac{1}{2}[z^{-2}]\eta(E|_{\tau=0},\dvac^{r,\epsilon}(z))=\frac{(r-1)\epsilon^2}{2}$.
\end{proof}
\begin{remark}
By taking  $r=2$, this result recovers the Virasoro constraints of the generalized Br\'ezin-Gross-Witten partition function with parameter $N$ by taking $\epsilon=N\hbar / \sqrt{-2}$, see, e.g.,~\cite{Ale18}
\end{remark}

\section{Application II: CohFT of extended Grothendieck's dessins d'enfants}
\label{sec:vira-dessin}
A dessin d'enfant refers to a bi-colored ribbon graph embedded on a smooth Riemann surface,  Grothendieck~\cite{Gro97} established a one-to-one correspondence between dessins d’enfants and isomorphism classes of Belyi maps~\cite{Bel80},
and thus the theory is now known as Grothendieck’s dessins d’enfants.
Following~\cite{KZ15}, we denote by $N_{k,l}(\lambda_1,\cdots,\lambda_m)$ the counting invariants of dessins d'enfants, where   $k,l,\lambda_1,\cdots,\lambda_m\in\mathbb Z_{+}$, and denote their generating function~\footnote{Here a parameter $s$ used in~\cite{KZ15} is taken to be ${1}$, it can be recovered by taking $p_k\to s^k p_k$.
Also, we add a parameter $\hbar$ whose power traces the genus of the Riemann surface: $2g-2=\sum_{i}\lambda_i-(k+l+m)$.} by
\begin{equation}  \label{FGDD}
F^{\Gdd}(u,v,{\bf p};\hbar)=\sum_{k,l,m\geq 1}\frac{1}{m!}\sum_{\lambda\in\mathbb Z_{+}^m} \hbar^{\sum_{i}\lambda_i-(k+l+m)}N_{k,l}(\lambda)u^kv^lp_{\lambda_1}\cdots p_{\lambda_m}.
\end{equation}
Let $Z^{\Gdd}(u,v,{\bf p};\hbar)=\exp(F^{\Gdd}(u,v,{\bf p};\hbar))$, then according to Kazarian and Zograf~\cite{KZ15}, $Z^{\Gdd}(u,v,{\bf p};\hbar)$ satisfies the following Virasoro constraints:
$$
L^{\Gdd}_mZ^{\Gdd}(u,v,{\bf p};\hbar)=0,\qquad m\geq 0,
$$
where the operators $L^{\Gdd}_m$, $m\geq 0$, are given by
\beq\label{def:LKZm}
 L^{\Gdd}_m= \delta_{m,0}\cdot\frac{uv}{\hbar^2} 
+\sum_{k\geq 1} (p_{k}-\delta_{k,1})(k+m)\frac{\pd }{\pd p_{k+m}}
+(u+v)m\frac{\pd}{\pd p_{m}}
+\hbar^2\sum_{k=1}^{m-1}k(m-k)\frac{\pd^2}{\pd p_k\pd p_{m-k}}.
\eeq
Conversely, $Z^{\Gdd}(u,v,{\bf p};\hbar)$ is uniquely determined by the Virasoro constraints and the initial condition
$Z^{\Gdd}(u,v,{\bf 0};\hbar)=1$.
We refer readers to~\cite{Gro97, KZ15, Zhou19} for details of definitions and various applications of the theory.

 In this section, we introduce a two-dimensional semisimple homogeneous CohFT $\Omega^{\eGdd}$ calibrated by a certain $S$-matrix and $\dvac$-vector.
We will prove that the descendent correlators of $\Omega^{\eGdd}$ extend the counting invariants of dessins d'enfants and thus we call $\Omega^{\eGdd}$ the CohFT of extended Grothendieck’s dessins d’enfants.

\subsection{Construction of CohFT $\Omega^{\eGdd}$} \label{CohFTeGdd}
Consider a generalized Frobenius structure $(\Frob,\eta,*_{\tau})$, where $\Frob$ is a two-dimensional $\mathbb C$ vector space spanned by $\{\phi_0,\phi_1\}$,
$\eta$ is a symmetric bilinear form on $\Frob$ defined by $\eta(\phi_a,\phi_b)=\delta_{a+b,1}$,
$\tau=\tau^0\phi_0+\tau^1\phi_1$ is a coordinate system of $\Frob$,
and the quantum product $*_{\tau}$ is determined by the potential
$$\textstyle
\Phi(\tau)=(\tau^1+\epsilon_1)(\tau^1+\epsilon_2)\log(\frac{1}{1-\tau^0})
+\sum_{i=1}^{2}\frac{(\tau^1+\epsilon_i)^2}{2}\big(\log(1+\frac{\tau^1}{\epsilon_i})-\frac{3}{2}\big).
$$
Here $\epsilon_{1,2}$ are two parameters. 
It is straightforward to see
\begin{align*}
\phi_0*_{\tau}\phi_0=&\,\textstyle \frac{2\tau^1+\epsilon_1+\epsilon_2}{(1-\tau^0)^2} \phi_0 +\frac{2(\tau^1+\epsilon_1)(\tau^1+\epsilon_2)}{(1-\tau^0)^3}\phi_1, \\
\phi_0*_{\tau}\phi_1=&\,\textstyle  \frac{2}{1-\tau^0}\phi_0+\frac{2\tau^1+\epsilon_1+\epsilon_2}{(1-\tau^0)^2}\phi_1,\\
\phi_1*_{\tau}\phi_1=&\,\textstyle \frac{2\tau^1+\epsilon_1+\epsilon_2}{(\tau^1+\epsilon_1)(\tau^1+\epsilon_2)}\phi_0+\frac{2}{1-\tau^0}\phi_1,
\end{align*}
and it follows that $*_{\tau}$ is commutative and associative.
By these explicit formulae, one can see the vector field of the unities of the quantum product $*_{\tau}$ is given by 
$$\textstyle
{\bf 1}=\frac{2(\tilde\tau^0)^2\tilde \tau^1}{(\epsilon_1-\epsilon_2)^2}\phi_0
+\frac{2\tilde\tau^0(\tau^1+\epsilon_1)(\tau^1+\epsilon_2)}{(\epsilon_1-\epsilon_2)^2}\phi_1,
$$
where $\tilde\tau^0=\tau^0-1$, $\tilde\tau^1=\tau^1+\frac{\epsilon_1+\epsilon_2}{2}$.
Clearly, ${\bf 1}$ is not flat.

By checking equations~\eqref{eqn:euler-qp} and~\eqref{eqn:euler-eta}, we can see this generalized Frobenius manifold is homogeneous with the Euler vector field
$$
E=(1-\tau^0)\phi_0,
$$
and its conformal dimension $\delta=3$.
Given the Euler vector field, one has the grading operator $\mu$ by~\eqref{eqn:mu-E}, precisely,
$\mu(\phi_a)=(\frac{1}{2}-a)\phi_a$, $a=0,1$.

Furthermore, this generalized Frobenius manifold is semisimple and the canonical coordinates are given by
$u^{1,2}=u^{\pm}=\frac{1}{1-\tau^0}\big(2\tau^1+\epsilon_1+\epsilon_2\pm2\sqrt{(\tau^1+\epsilon_1)(\tau^1+\epsilon_2)}\big)$.
Let $\widetilde\Psi^\alpha_i=\frac{\pd u^\alpha}{\pd \tau^i}$, then we have the canonical basis $e_\alpha=(\widetilde\Psi^{-1})^i_\alpha\phi_i$
and the normalized canonical basis $\bar e_\alpha=\Delta_{\alpha}^{\frac{1}{2}}e_\alpha$, where $\Delta_{\alpha}=\eta(e_\alpha,e_\alpha)^{-1}$.
We introduce the $\Psi$-matrix: $\Psi:=\Delta^{-\frac{1}{2}}\widetilde\Psi$.

Given the semisimple homogeneous generalized Frobenius manifold, one constructs a semisimple homogeneous CohFT by the Givental--Teleman reconstruction theorem~\cite{Giv01a,Tel12}.
Precisely, the $R$-matrix $R(z)$ is uniquely determined by the homogeneity equation~\eqref{eqn:hom-R}, we have the following formula:
$$
R^{\tau}(z)=\sum_{m\geq 0}\frac{\prod_{i=0}^{m-1}(4i^2-1)}{m! 4^m (u^2-u^1)^m}\cdot 
\left(\begin{array}{cc}
	1 & 2m\sqrt{-1}\\
	(-1)^{m-1}2m\sqrt{-1} & (-1)^m
\end{array}\right) 
\cdot z^m.
$$
Similarly, the vacuum vector $\vac^{\tau}(z)$ is uniquely determined by the formula~\eqref{eqn:formula-vacuum} or by the homogeneity condition~\eqref{eqn:hom-vacuum}, and we have the following formula: 
$$
\vac^{\tau}(z)=\sum_{k\geq 0}\sum_{a=0}^1\sum_{i=0}^{\frac{k+1+a}{2}}(-1)^i\frac{(2k+1-2i)!!(k+1+a)!}{2^{i+a}i!(k+1+a-2i)!}
\frac{(\tilde\tau^0)^{k+2-a}(2\tilde\tau^1)^{k+1+a-2i}}{(\epsilon_1-\epsilon_2)^{2k+2-2i}}\phi_a z^k.
$$
 By Theorem~\ref{thm:reconstruction}, this defines a shifted CohFT, and we denote it by $\Omega^{\eGdd, \tau}$.

Given the CohFT, one defines the ancestor correlator $\<-\>^{\eGdd,\tau}_{g,n}$, the genus-$g$ ancestor potential $\bar \cF_g^{\eGdd,\tau}({\bf s})$ and the total ancestor potential $\cA^{\eGdd, \tau}({\bf s};\hbar)$ by
equations~\eqref{def:bracket-anc}, \eqref{def:anc-Fg-choft} and~\eqref{def:barF-A} respectively.
One can also get $\cA^{\eGdd, \tau}({\bf s};\hbar)$ directly by formula~\eqref{eqn:A-KW}.

\subsection{$S$- and $\dvac$-calibration} \label{SeGdd}
Now we define the descendent potentials by choosing an $S$-matrix and a $\dvac$-vector.
We first compute the $S$-matrix by solving the QDE~\eqref{eqn:QDE-S} and by choosing the integration constants at each step
(at the first step, we require $(S^{\tau}_1)^b_a=\frac{\pd^2 \Phi(\tau)}{\pd \tau^a\pd\tau^{1-b}}$, $a, b=0,1$).
We choose the $S$-matrix  to be the following one:
$$
S^{\tau}(z)=\sum_{n\geqslant 0}\left(
\begin{array}{cc}
	\sum\limits_{i+2j=n}\frac{a_0(\tau)^ia_1(\tau)^j}{i!j!j!} & 
	\sum\limits_{i+2j=n-1}\frac{a_0(\tau)^ia_1(\tau)^j}{i!j!j!}\big(\log(\frac{a_1(\tau)}{\epsilon_1\epsilon_2})-2H_j\big)\\
	\sum\limits_{i+2j=n+1}\frac{a_0(\tau)^ia_1(\tau)^j}{i!j!(j-1)!}& 
	\sum\limits_{i+2j=n}\frac{a_0(\tau)^ia_1(\tau)^j}{i!j!j!}(1+j\log(\frac{a_1(\tau)}{\epsilon_1\epsilon_2})-2j H_j) 
\end{array}
\right)\cdot z^{-n},
$$
where $a_0(\tau)=\frac{2\tau^1+\epsilon_1+\epsilon_2}{1-\tau^0}$ and $a_1(\tau)=\frac{(\tau^1+\epsilon_1)(\tau^1+\epsilon_2)}{(1-\tau^0)^2}$.
Here the detailed process of the computations is omitted because it is rather complex to display.
However, one can check that this $S^{\tau}(z)$ satisfy the QDE~\eqref{eqn:QDE-S} and the symplectic condition~\eqref{eqn:S-symp} straightforwardly.
Moreover, this $S$-matrix satisfies the homogeneity condition~\eqref{eqn:hom-S}. Precisely, we have for $n\geq 1$,
$$
(n+\mu_a-\mu_b)(S^{\tau}_{n})^b_a=(\E S^{\tau}_{n-1}-S^{\tau}_{n-1}\rho)^b_a,
$$
where $\rho\in\End(\Frob)$ is given by $\rho\phi_0=0$, $\rho\phi_1=2\phi_0$.

To compute the $\dvac$-vector, we assume $\dvac^{\tau}(z)$ has form $\dvac^{\tau}_{1}z^{-1}+\dvac^{\tau}_{2}z^{-2}+\cdots$, then by the QDE~\eqref{eqn:QDE-vacuum-des}, we have firstly $\pd_{\tau^a}\dvac^{\tau}_1 =-\phi_a$.
Take $\dvac^{\tau}_1=-\tilde \tau$, we see the $J$-function defined by~\eqref{def:J-function-small} has form
$J^{\tau}(-z)=\tilde\tau+(\dvac^\tau_2+S^\tau_1\tilde\tau)\cdot(-z)^{-1}+\cdots$.
By requiring $\dvac^\tau_2+S^\tau_1\tilde\tau=\frac{\pd \Phi(\tau)}{\pd \tau^a}\phi^a$, we get 
$\dvac^{\tau}_{2}=\frac{1}{2}(\epsilon_1-\epsilon_2)\cdot \big(\log(1+\frac{\tau^1}{\epsilon_1})-\log(1+\frac{\tau^1}{\epsilon_2})\big)\phi_0-\tfrac{(\epsilon_1-\epsilon_2)^2}{2(1-\tau^0)}\phi_1$.
For $k\geq 3$, $\dvac^{\tau}_{k}$ is determined by the homogeneity condition~\eqref{eqn:hom-vacuum-des-2}
$$
\E\dvac^{\tau}_{k-1}=(k-\mu-\tfrac{\delta}{2})\dvac^{\tau}_{k}.
$$
Notice that $k-\mu_a-\tfrac{\delta}{2}=k-2+a>0$ for $k\geq 3$ and $a=0,1$, we get $\dvac^{\tau}_{k}$ from $\dvac^{\tau}_{k-1}$ for $k\geq 3$.

Given the $S$-matrix and the $\dvac$-vector, we have the $J$-function by equation~\eqref{def:J-function-small}.
The total descendent potential $\cD^{\eGdd}(\epsilon_1,\epsilon_2, {\bf t};\hbar)$ is then defined by~\eqref{eqn:des-anc}.

\subsection{Virasoro constraints and their application}
Similarly as the deformed negative $r$-spin theory, by the properties of $\Omega^{\eGdd,\tau}$ and by directly applying the third part of Theorem~\ref{thm:vira-anc-intro},
one gets the ancestor Virasoro constraints for $\cA^{\eGdd,\tau}$ immediately.
Now we consider the descendent Virasoro constraints and deduce explicit formulae for Virasoro operators.

\begin{proposition}
	\label{prop:vira-dessin}
	The total descendent potential $\cD^{\eGdd}(\epsilon_1,\epsilon_2, {\bf t};\hbar)$ satisfies the following Virasoro constraints:
	$$
	L^{\eGdd}_m\cD^{\eGdd}({\bf t};\hbar)=0,\qquad m\geq 0,
	$$
	where $L^{\eGdd}_{m}$, $m\geq 0$, satisfy the commutation relation 
$[L^{\eGdd}_{m},L^{\eGdd}_{n}]=(m-n)L^{\eGdd}_{m+n}$, 
and have the following explicit formulae:
	\begin{align*}
		L^{\eGdd}_m=&\, \delta_{m,0}\cdot\frac{(t_0^1+\epsilon_1)(t_0^1+\epsilon_2)}{\hbar^2} 
+\sum_{k\geq 0}\frac{(k+m+1)!}{k!}\tilde t_{k}^{0}\frac{\pd }{\pd t_{k+m}^0}
+\sum_{k\geq 1}\frac{(k+m)!}{(k-1)!} t_{k}^{1}\frac{\pd }{\pd t_{k+m}^1}\\
&\, +2m!\tilde t_0^1\frac{\pd}{\pd t_{m-1}^0}
+2\sum_{k\geq 1}\sum_{i=0}^{m}\frac{(k+m)!}{(k-1)!(k+i)}t_k^1\frac{\pd}{\pd t_{k+m-1}^0}\\
&\, +\hbar^2\sum_{k=1}^{m-1}k!(m-k)!\frac{\pd^2}{\pd t_{k-1}^0\pd t_{m-k-1}^0}.
	\end{align*}
Here the shifts on the time variables are given by $\tilde t_k^a=t_k^a-\delta_{k,0}\delta_{a,0}+\delta_{k,0}\delta_{a,1}\frac{\epsilon_1+\epsilon_2}{2}$.
\end{proposition}
\begin{proof}
The proof is similar as that for the Virasoro constraints of the negative $r$-spin theory
and we have $L^{\eGdd}_{m}=L_m+\frac{\delta_{m,0}}{\hbar^2}\cdes$.
To see the explicit formula of $L^{\eGdd}_m$, we note firstly $\eta(\rho^{m+1}\tilde t_0,\tilde t_0)=2\delta_{m,0}\cdot (\tilde t_0^1)^2$ and 
$\cdes=\frac{1}{2}\eta(E,\dvac^{\tau}_2)=-\frac{(\epsilon_1-\epsilon_2)^2}{4}$.
Then the explicit formula of  $L^{\Gdd}_m$ follows from 
$$\textstyle
D_{\rho,z}^{m+1}\phi_az^{k}=(k+2-a)_{m+1}\cdot \phi_az^{k+m+1}
+2\delta_{a,1}\cdot \sum_{i=1}^{m+1}\frac{(k+1)_{m+1}}{(k+i)}\phi_0z^{k+m},
$$
where $(x)_{m+1}=(x)(x+1)\cdots(x+m)$ and for $-m-1\leq k\leq -1$, the summation in the formula should be understood as $(-1)^{-k-1}(k+1+m)!(-k-1)!$.
\end{proof}
The following Proposition is a direct consequence of the Virasoro constraints of $\cD^{\eGdd}$:
\begin{proposition} \label{prop:Gdd-eGdd}
Let $Z^{\Gdd}$ be the exponential of the generating function for the Grothendieck's dessins d'enfants \eqref{FGDD}, and let $\cD^{\eGdd}$ be the total descendent potential of the calibrated CohFT introduced in \S \ref{CohFTeGdd} and \S \ref{SeGdd}. 
Then $\cD^{\eGdd}$ extends the function  $Z^{\Gdd}$. Namely, when we take $t^1_k=0$ and $t_k^0=k!\, p_{k+1}$ for $k\geq 0$ in $\cD^{\eGdd}$,
we have
$$
\cD^{\eGdd}(u,v,{\bf t};\hbar)/\cD^{\eGdd}(u,v, {\bf 0};\hbar)=Z^{\Gdd}(u,v,{\bf p};\hbar).
$$
\end{proposition}
\begin{proof}
Under the condition required in the Proposition, the Virasoro operator $L^{\eGdd}_m$ becomes
$L^{\Gdd}_{m}$ defined by equation~\eqref{def:LKZm}.
The conclusion follows from the uniqueness of the solution $Z^{\Gdd}(u,v,{\bf p};\hbar)$ of the Virasoro constraints $ L^{\Gdd}_m Z^{\Gdd}(u,v,{\bf p};\hbar)=0$, $m\geq 0$, with initial condition $Z^{\Gdd}(u,v,{\bf 0};\hbar)=1$.
\end{proof}

\medskip

\addtocontents{toc}{\protect\setcounter{tocdepth}{0}}

\begin{appendices}
\section{Homogeneity conditions on the descendent side}
\subsection{Existence of the homogeneous $S$-matrix}
\label{sec:hom-S}
In this subsection, we prove the existence of the homogeneous $S$-matrix.
To achieve this, we first do some preparations.
Fix a homogeneous basis $\{\phi_a\}_{a=0}^{N-1}$ of the state space $\Frob$, and denote $\mu(\phi_a)=\mu_a\phi_a$, we introduce two sets 
$$
I_1:=\{\mu_a|a=0,\cdots,N-1\},\qquad
I_2:=\{\mu_a-\mu_b| a,b =0,\cdots,N-1\}.
$$
For $\alpha, \beta\in \mathbb R$, we define 
$$
V_{\alpha}:=\{v\in \Frob |\mu(v)=\alpha\cdot v\},
$$
and
$$
M_\beta:=\{A:\Frob\to \Frob| [\mu,A]=\beta\cdot A\}.
$$
It is clear that $V_\alpha\leqslant \Frob$ (resp. $M_{\beta}\leqslant \End(\Frob)$)
and $V_{\alpha}= \{0\}$ (resp. $M_\beta= \{0\}$) unless $\alpha\in I_1$ (resp. $\beta\in I_2$).
Furthermore, we have decomposition:
$$
\Frob=\oplus_{\alpha\in I_1}V_\alpha,\qquad
\End(\Frob)=\oplus_{\beta\in I_2}M_{\beta}.
$$
We introduce the notation
$$
M^{c}_{\beta}:=\oplus_{\beta'\in I_2\setminus \{\beta\}} M_{\beta'}.
$$
\begin{lemma}\label{lem:decom}
(1). If $A\in M_{\beta}$, then $A^{*}\in M_{\beta}$, where $A^*$ is the  adjoint of $A$ with respect to $\eta$.
(2). The endomorphisim $(\ad_{\mu}-\beta)|_{M^{c}_{\beta}}:M^{c}_{\beta}\to M^{c}_{\beta}$ is an automorphism of $M^{c}_{\beta}$.
\end{lemma}
\begin{proof}
The first part follows from $[\mu, A^*]=([\mu,A])^{*}$ (because $\mu^*=-\mu$) and the second part is obvious since the endomorphisim is injective. 
\end{proof}
Now we are ready to prove the existence of the homogeneous $S$-matrix:
\begin{proposition}
For a homogeneous CohFT, there always exits an homogeneous $S$-matrix.
\end{proposition}
\begin{proof}
We first note that given an $S$-matrix $\tilde S^{\tau}(z)$ (i.e., a solution to the QDE~\eqref{eqn:QDE-S} satisfying the symplectic condition),
there is always an operator $\tilde\rho(z)\in\End(\Frob)[[z^{-1}]]$ such that equation~\eqref{eqn:hom-S} holds for $\tilde S^{\tau}(z)$ ($\tilde \rho(z)$ may not satisfy equation~\eqref{eqn:hom-rho}).
In fact, we can view equation~\eqref{eqn:hom-S} as the definition of the operator $\tilde\rho(z)$, more precisely,
\beq\label{def:rho}
\tilde\rho(z):=-z\mu+z \tilde S^{\tau,*}(-z)\mu \tilde S^{\tau}(z)+z\tilde S^{\tau,*}(-z)\cdot (z\pd_z+E)\tilde S^{\tau}(z).
\eeq
Then by using equations~\eqref{eqn:S-symp}, \eqref{eqn:QDE-S} and \eqref{eqn:euler-qp},
it is straightforward to check that $\tilde\rho(z)$ does not depend on $\tau$ and satisfies $\tilde\rho^*(z)=\tilde\rho(-z)$.

Let $S^{\tau}(z)=\tilde S^{\tau}(z)A(z)$, where $A(z)$ is a constant matrix taking form $A(z)=\id+A_1z^{-1}+\cdots$ and satisfies $A^{*}(-z)A(z)=\id$,
then $S^{\tau}(z)$ gives another $S$-matrix and we define ${\rho}(z)$ by $S^{\tau}(z)$ via equation~\eqref{eqn:hom-S}.
Our goal is to find a matrix $A(z)$ such that $\rho(z)$ satisfies equation~\eqref{eqn:hom-rho}.

By substituting $\tilde S^{\tau}(z)=S^{\tau}(z)A(z)^{-1}$ into equation~\eqref{def:rho},  we have
$$
z\pd_zA(z)=A(z)(\mu+\rho(z)/z)-(\mu+\tilde\rho(z)/z)A(z).
$$
By taking expansion of $A(z)$, $\rho(z)$ and $\tilde\rho(z)$, this equation gives
\beq\label{eqn:select-rho}\textstyle
-kA_k=A_k\mu-\mu A_k+\sum_{i=0}^{k-1}(A_{k-1-i} \rho_{i}-\tilde \rho_i A_{k-1-i}).
\eeq
For $k=0$, the equation is trivial. For $k=1$, the equation reads
$$
\rho_0=\tilde\rho_0+([\mu, A_1]-A_1).
$$
By decomposition $\End(\Frob)=M_{1}\oplus M_{1}^{c}$,
we have $\tilde\rho_0=\tilde\rho'_0+\tilde\rho''_0$ where $\tilde\rho'_0\in M_{1}$ and $\tilde\rho''_0\in M_{1}^{c}$.
Similarly, given an operator $A_1$ we can write $A_1=A'_1+A''_1$ where $A'_1\in M_{1}$ and $A''_1\in M_{1}^{c}$.
By Lemma~\eqref{lem:decom}, there exist $A''_1\in M_{1}^{c}$ ($A'_1$ can be arbitrary) such that $\tilde\rho''_0+([\mu, A''_1]-A''_1)=0$, and we can take $\rho_0=\tilde\rho'_0\in M_1$.
Now we consider $k=2$ case of equation~\eqref{eqn:select-rho}:
$$
\rho_1=\tilde\rho_1-(A_1\rho_0-\tilde\rho_0 A_1)+([\mu, A_2]-2A_2).
$$
Similar as the $k=1$ case, we can take a fixed $A''_2\in M^{c}_2$ and an arbitrary $A'_{2}\in M_2$ such that $\rho_1\in M_2$.
Going on the procedure, on each step we can choose a certain $A''_{k}\in M_{k}^{c}$ and an arbitrary $A'_k\in M_k$ such that $\rho_{k-1}\in M_{k}$. Finally we get an matrix $A(z)$ such that $\rho_i\in M_{i+1}$, $i\geq 0$.
Equivalently, this means $\rho(z)$ satisfies equation~\eqref{eqn:hom-rho}.

To see we can require $A^{*}(-z)A(z)=\id$, we note firstly the selected $A(z)$ satisfies
$$
z\pd_zA^{*}(-z)=A^{*}(-z)(\mu+\tilde\rho(z)/z)-(\mu+\rho(z)/z)A^{*}(-z).
$$
where we have used $\rho^{*}(-z)=\rho(z)$ and $\tilde\rho^{*}(-z)=\tilde\rho(z)$.
Then we have
$$
z\pd_z(A^{*}(-z)A(z))=A^{*}(-z)A(z)(\mu+\rho(z)/z)-(\mu+\rho(z)/z)A^{*}(-z)A(z).
$$
Denote $A^{*}(-z)A(z)=\sum_{k\geq 0}B_kz^{-k}=\sum_{k\geq 0}(B'_k+B''_k)z^{-k}$, where $B'_k\in M_{k}$ and $B''_k\in M_{k}^{c}$, then
$$\textstyle
[\mu,B''_k]-kB''_k=\sum_{i=0}^{k-1}[B''_{k-1-i},\rho_i]
+\sum_{i=0}^{k-1}[B'_{k-1-i},\rho_i].
$$
From the recursion procedure, we can take $A(z)$ such that $B'_k=0$ (since we have freedom to choose each $A'_k$), then we have
$$\textstyle
[\mu,B''_k]-kB''_k=\sum_{i=0}^{k-1}[B''_{k-1-i},\rho_i].
$$
Start from $B''_0=0$ (because $B_0=\id$) and by the second part of Lemma~\ref{lem:decom}, this recursion tells us $B''_k=0$ for all $k\geq 0$.
\end{proof}

We introduce a sufficient condition that one can require $\rho(z)=\rho_0$.
\begin{lemma}
If for each $k\in \mathbb Z_{\geq 1}$ and for each $A\in M_{k+1}$, there is an operator $A'\in M_k$ such that $A=[\rho_0,A']$, then we can take $\rho(z)=\rho_0$.
\end{lemma}
\begin{proof}
Given a homogeneous $S$-matrix $\tilde S^{\tau}(z)$, let $S^{\tau}(z)=\tilde S^{\tau}(z)A(z)$, where $A(z)$ is a constant matrix taking form $A(z)=\id+A_1z^{-1}+\cdots$ and satisfying $A^{*}(-z)A(z)=\id$
and $z\pd_zA(z)=[A(z),\mu]$ (i.e., $A_k\in M_k$),
then $S^{\tau}(z)$ gives another homogeneous $S$-matrix.
We define $\tilde{\rho}(z)$ (resp. $\rho$) by $\tilde S^{\tau}(z)$ (resp. $S^{\tau}(z)$) via equation~\eqref{def:rho}, then we have
$$
z\pd_zA(z)=A(z)(\mu+\rho(z)/z)-(\mu+\tilde\rho(z)/z)A(z).
$$
By taking expansion of $A(z)$, $\rho(z)$ and $\tilde\rho(z)$, we get for $k\geq 1$,
$$\textstyle
\sum_{i=0}^{k-1}(A_{k-1-i}  \rho_{i}- \tilde \rho_i A_{k-1-i})=0.
$$

For $k=1$, the equation solves $\rho_0=\tilde\rho_0$. In fact, the homogeneity condition~\eqref{eqn:hom-S} gives us
$\E=S_1+[S_1,\mu]+\rho_0$, restricted on the subspace $M_1$, this determines $\rho_0$ by $\E$.
For $k\geq 2$, the equation can be rewritten as follows:
$$\textstyle
 \rho_{k-1}=\tilde\rho_{k-1}+A_1(\tilde\rho_{k-2}-\rho_{k-2})+\cdots +A_{k-1}(\tilde\rho_0-\rho_0) 
+[\tilde\rho_{k-1}, A_1]+\cdots+[\tilde\rho_0, A_{k-1}].
$$
Suppose we have $\rho_i=\tilde \rho_i$ for $i=0,\cdots,k-2$, then we take $A_1=\cdots=A_{k-2}=0$, the equation becomes
$$\textstyle
 \rho_{k-1}=\tilde\rho_{k-1}+[\rho_0, A_{k-1}].
$$
If $[\rho_0, M_{k-1}]=M_{k}$, then there is an operator $\tilde A_{k-1}\in M_{k-1}$ such that $\tilde\rho_{k-1}+[\rho_0,\tilde A_{k-1}]=0$.
Take $A_{k-1}=\frac{1}{2}(\tilde A_{k-1}+(-1)^k\tilde A^{*}_{k-1})$ and take $A_{l}$ for $l\geq k$ such that $A^{*}(-z)A(z)=\id$, one get $\rho_{k-1}=0$.
Repeat the procedure, if $[\rho_0, M_{k-1}]=M_{k}$ for each $k\geq 2$, then we can take $\rho_{k-1}=0$ for $k\geq 2$, i.e., $\rho(z)=\rho_0$.
\end{proof}

\subsection{Some consequences of homogeneity conditions for $J$-function}
\label{sec:cons-hom-J}
In this subsection, we deduce some consequences of homogeneity conditions for $J$-function.
Before showing explicit formulae, we note firstly the operator $\rho(z)$ is nilpotent.
This can be proved as follows: equation~\eqref{eqn:hom-rho} gives us
\beq\label{eqn:com-mu-rho}
\mu\rho_i=\rho_i(\mu+i+1), \qquad i\geq 0.
\eeq
This can be rewritten as $(\rho_i)_a^b(i+1+\mu_a-\mu_b)=0$ for $a,b=0,\cdots,N-1$, 
where $(\rho_i)_a^b=\eta(\phi^b,\rho_i\phi_a)$. Since $|\mu_a-\mu_b|<\infty$, $\rho_i=0$ for $i$ big enough. 
Similarly, by using the commutation relation~\eqref{eqn:com-mu-rho} repeatedly, one can prove $\rho_{i_1}\cdots\rho_{i_m}=0$ for $m$ big enough, especially, this proves $\rho(z)$ is nilpotent.

Now we prove equations~\eqref{eqn:hom-J-positive}--\eqref{eqn:hom-Phi} one by one.
Firstly, denote $\dil(z)=\sum_{j\geq 0}^{n}\dil_j z^j$,
the coefficient of $z^{k+1}$, $k\geq 0$, in equation~\eqref{eqn:hom-J} gives 
$$\textstyle
(k+\mu+\frac{\delta}{2})\dil_{k}=-\sum_{i=0}^{n-k-1}\rho_i \dil_{k+i+1}.
$$
In particular, for $k=n$, we see $(n+\mu+\frac{\delta}{2})\dil_{n}=0$.
For $k=n-1$, we have
$$\textstyle
(n-1+\mu+\frac{\delta}{2})^2\dil_{n-1}=-(n-1+\mu+\frac{\delta}{2})\rho_0 \dil_{n}=-\rho_0(n+\mu+\frac{\delta}{2}) \dil_{n}=0.
$$
By induction on $k$ from a bigger one to a smaller one and by using the commutation relation~\eqref{eqn:com-mu-rho} , one can prove
$(k+\mu+\frac{\delta}{2})^{n+1-k}\dil_{k}=0$, $k=n,n-1,\cdots,0$.
Notice the operator $\mu$ is diagonal under flat basis, we get
$\big(z\pd_z+\mu+\tfrac{\delta}{2}\big)\dil(z)=0$ and $\rho(z)\dil(z)\in \Frob[z^{-1}]$.
Secondly, by considering the coefficient of $z^{0}$ in equation~\eqref{eqn:hom-J}, one has immediately equation~\eqref{eqn:hom-J-0}.
Lastly, consider the coefficient of $z^{-1}$ in equation~\eqref{eqn:hom-J} and by the relation of $\Phi(\tau)$ with $J$-function:
$J^{\tau}_{0,a}=\pd_{\tau^a}\Phi(\tau)$, we have
$$\textstyle
E J^{\tau}_{0,a}=(2-\tfrac{\delta}{2}+\mu_a)J^{\tau}_{0,a}+\eta(\phi_a,\rho_0\tilde\tau)-\sum_{i\geq 0}\eta(\phi_a,\rho_{i+1}\dil_i),
$$
By taking integration, we get equation~\eqref{eqn:hom-Phi}.

\section{Genus-0 and genus-1 ancestor Virasoro constraints}\label{sec:app-anc-vira}
\subsection{Genus-0 ancestor Virasoro constraints}
In this subsection, we study the genus-0 ancestor Virasoro constraints $\mathscr L^{\anc}_{0,m}({\bf s})=0$, $m\geq -1$, where we have precisely
\begin{align*}
	\mathscr L^{\anc}_{0,m}({\bf s})=&\,\frac{1}{2}\eta(\E^{m+1}s_0,s_0)+\corr{[D^{m+1}_{\E,\bar\psi}\bar\psi^{-1}\tilde{\bf s}(\bar\psi)]_{+}}^{\tau}_{0,1}\\
	&\,  +\frac{1}{2}\sum_{k=1}^{m}(-1)^k\corr{\phi_a\bar\psi^{k-1}}^{\tau}_{0,1}
	\corr{[D^{m+1}_{\E,\bar\psi}\phi^a\bar{\psi}^{-k-1}]_{+}}^{\tau}_{0,1}.
\end{align*}

The following Lemma will be useful for us.
\begin{lemma}
	We have the following formula:
	\beq\label{eqn:anc-Doper-dtau}
	\pd_{\tau^i} D^{m+1}_{\E,z}=[z^{-1}\cdot\phi_i*, D^{m+1}_{\E,z}],\qquad m\geq -1.
	\eeq
\end{lemma}
\begin{proof}
	For $m=-1$, the equation is trivial;
	for $m=0$, the equation is equivalent to equation~\eqref{eqn:diff-E};
	for $m\geq 0$, the equation follows immediately from the case for $m=0$.
\end{proof}

We would like to mention that the genus-0 ancestor Virasoro constraints are nothing but the Frobenius structure and the genus-0 tautological relations. 
Still we give a proof to display the equivalence.
\begin{proof}[Proof of the first part of Theorem~\ref{thm:anc-vira}]
	By considering the coefficients of $\prod_{i=1}^{n}s_{k_i}^{b_i}$, the genus-0 ancestor Virasoro equations are equivalent to the following equations:
	\begin{align}
		&\<D^{m+1}_{\E,\bar\psi}\vac^{\tau}(\bar\psi),\phi_{b_{[n]}}\bar\psi^{k_{[n]}}\>^{\tau}_{0,n+1}\nonumber\\
		=&\,  \delta_{n,2}\delta_{k_1,0}\delta_{k_2,0}\eta(\E^{m+1}\phi_{b_1},\phi_{b_2})
		+\sum_{i=1}^{n}\< D^{m+1}_{\E,\bar\psi}\bar\psi^{k_i-1}\phi_{b_i}, \phi_{b_{[n]-\{i\}}}\bar\psi^{k_{[n]-\{i\}}}\>^{\tau}_{0,n}\nonumber\\
		&\,  +\frac{1}{2}\sum_{k=1}^{m}(-1)^k\sum_{I\sqcup J=[n]}\<\phi_a\bar\psi^{k-1},\phi_{b_I}\bar\psi^{k_I}\>^{\tau}_{0,|I|+1}
		\<D^{m+1}_{\E,\bar\psi}\phi^a\bar{\psi}^{-k-1},\phi_{b_J}\bar\psi^{k_J}\>^{\tau}_{0,|J|+1},\label{eqn:genus-0-anc-vira}
	\end{align}
	where the correlators with insertion $\phi_a\psi^{-k}$, $k>0$, are set to be $0$.
	We call equation~\eqref{eqn:genus-0-anc-vira} the $(n; k_1,\cdots,k_n)$-case of the genus zero $L_{m}^{\tau}$-constraint.
	By dimensional reason, these equations are trivial for $n\leq 1$ or $\sum_{i=1}^{n}k_i\geq n-1$.
	For $n=2$, the remaining equation to be proved is the one with $k_1=k_2=0$, which is
	\begin{align*}
		\<D^{m+1}_{\E,\bar\psi}\vac^{\tau}(\bar\psi),\phi_{b_{1}},\phi_{b_{2}}\>^{\tau}_{0,3}
		=\eta(\E^{m+1}\phi_{b_1},\phi_{b_2}).
	\end{align*}
	One can prove it by noticing $D^{m+1}_{\E,\bar\psi}\vac^{\tau}(\bar\psi)=E^{m+1}+O(\bar\psi)$.
	
	We apply induction on $n$.
	Suppose the $(n;k_1,\cdots,k_n)$-case genus zero $L_m^{\anc}$-constraint holds for a fixed integer $n\geq 2$ and arbitrary integers $k_1,\cdots,k_n$.
	We take derivative with respect to $\tau^{b_{n+1}}$ on such equation.
	By using equations~\eqref{eqn:anc-Doper-dtau}, \eqref{eqn:diff-anc-corr} and \eqref{eqn:QDE-vacuum}, the derivative of the left-hand side of equation~\eqref{eqn:genus-0-anc-vira} gives
	\begin{align}
		&\<D^{m+1}_{\E,\bar\psi}\vac^{\tau}(\bar\psi),\phi_{b_{[n]}}\bar\psi^{k_{[n]}},\phi_{b_{n+1}}\>^{\tau}_{0,n+2}
		-\<D^{m+1}_{\E,\bar\psi}\phi_{b_{n+1}}\bar\psi^{-1},\phi_{b_{[n]}}\bar\psi^{k_{[n]}}\>^{\tau}_{0,n+1}\nonumber\\
		&\,\textstyle  -\sum_{i=1}^{n}\<D^{m+1}_{\E,\bar\psi}\vac^{\tau}(\bar\psi),\phi_{b_{1}}\bar\psi^{k_{1}},\cdots,\phi_{b_{n+1}}*\phi_{b_i}\bar\psi^{k_i-1} ,\cdots,\phi_{b_{n}}\bar\psi^{k_{n}}\>^{\tau}_{0,n+1}.\label{eqn:lhsofdgenus-0-anc-vira}
	\end{align}
	and the derivative of the right-hand side of equation~\eqref{eqn:genus-0-anc-vira} is
	$$\textstyle
	({\rm I})+\sum_{i=1}^{n}({\rm II})_i+\frac{1}{2}\sum_{k=1}^{m}(-1)^k\sum_{I\sqcup J=[n]}\big(({\rm III})^{l}_{I,J}+({\rm III})_{I,J}^{r}\big),
	$$ where 
	\begin{align*}
		({\rm I}) =&\,\delta_{n,2}\delta_{k_1,0}\delta_{k_2,0}\big(\<\phi_{b_3},\phi_{b_2},D_{\E,\bar\psi}^{m+1}\phi_{b_1}\bar\psi^{-1}\>^{\tau}_{0,3}
		+\<\phi_{b_3},\phi_{b_1},D_{\E,\bar\psi}^{m+1}\phi_{b_2}\bar\psi^{-1}\>^{\tau}_{0,3}\big),
	\end{align*}
\begin{align*}
		({\rm II})_{i}=&\,\< D^{m+1}_{\E,\bar\psi}\bar\psi^{k_i-1}\phi_{b_i}, \phi_{b_{[n]-\{i\}}}\bar\psi^{k_{[n]-\{i\}}},\phi_{b_{n+1}}\>^{\tau}_{0,n+1}
		\,-\< D^{m+1}_{\E,\bar\psi}\bar\psi^{k_i-2}\phi_{b_{n+1}}\phi_{b_i}, \phi_{b_{[n]-\{i\}}}\bar\psi^{k_{[n]-\{i\}}}\>^{\tau}_{0,n}\\
		&\,\textstyle -\sum_{1\leq j\leq n; j\ne i}\< D^{m+1}_{\E,\bar\psi}\bar\psi^{k_i-1}\phi_{b_i}, \phi_{b_{[n]-\{i,j\}}}\bar\psi^{k_{[n]-\{i,j\}}},\phi_{b_{n+1}}*\phi_{b_j}\bar\psi^{k_j-1}\>^{\tau}_{0,n},
	\end{align*}
	\begin{align*}
		({\rm III})^{l}_{I,J}=&\,\<\phi_{b_{n+1}},\phi_a\bar\psi^{k-1},\phi_{b_I}\bar\psi^{k_I}\>^{\tau}_{0,|I|+2}
		\<D^{m+1}_{\E,\bar\psi}\phi^a\bar{\psi}^{-k-1},\phi_{b_J}\bar\psi^{k_J}\>^{\tau}_{0,|J|+1}\\
		&\,-\<\phi_{b_{n+1}}*\phi_a\bar\psi^{k-2},\phi_{b_I}\bar\psi^{k_I}\>^{\tau}_{0,|I|+1}
		\<D^{m+1}_{\E,\bar\psi}\phi^a\bar{\psi}^{-k-1},\phi_{b_J}\bar\psi^{k_J}\>^{\tau}_{0,|J|+1}\\
		&\,\textstyle -\sum_{i\in I}\<\phi_a\bar\psi^{k-1},\phi_{b_{I-\{i\}}}\bar\psi^{k_{I-\{i\}}},\phi_{b_{n+1}}*\phi_{b_{k_i}}\bar\psi^{k_i-1}\>^{\tau}_{0,|I|+1}
		\<D^{m+1}_{\E,\bar\psi}\phi^a\bar{\psi}^{-k-1},\phi_{b_J}\bar\psi^{k_J}\>^{\tau}_{0,|J|+1},
	\end{align*}
	\begin{align*}
		({\rm III})^{r}_{I,J}=&\,\<\phi_a\bar\psi^{k-1},\phi_{b_I}\bar\psi^{k_I}\>^{\tau}_{0,|I|+2}
		\<D^{m+1}_{\E,\bar\psi}\phi^a\bar{\psi}^{-k-1},\phi_{b_J}\bar\psi^{k_J},\phi_{b_{n+1}}\>^{\tau}_{0,|J|+2}\\
		&\,-\<\phi_a\bar\psi^{k-1},\phi_{b_I}\bar\psi^{k_I}\>^{\tau}_{0,|I|+1}
		\<D^{m+1}_{\E,\bar\psi}(\phi_{b_{n+1}}*\phi^a)\bar{\psi}^{-k-2},\phi_{b_J}\bar\psi^{k_J}\>^{\tau}_{0,|J|+1}\\
		&\,\textstyle -\sum_{j\in J}\<\phi_a\bar\psi^{k-1},\phi_{b_{I}}\bar\psi^{k_{I}}\>^{\tau}_{0,|I|+1}
		\<D^{m+1}_{\E,\bar\psi}\phi^a\bar{\psi}^{-k-1},\phi_{b_{J-\{j\}}}\bar\psi^{k_{J-\{j\}}},\phi_{b_{n+1}}*\phi_{b_{k_j}}\bar\psi^{k_j-1}\>^{\tau}_{0,|J|+1}.
	\end{align*}
	By applying the $(n;k_1,\cdots,k_i-1,\cdots,k_n)$-case genus zero $L_m^{\anc}$-constraint on the second line of equation~\eqref{eqn:lhsofdgenus-0-anc-vira}, 
	and by some direct computations (note that the second lines of the expression for $({\rm III})^{l}_{I,J}$ and $({\rm III})^{r}_{I,J}$ cancel each other when taking summation over $k$ with a factor $(-1)^k$), one gets the $(n+1;k_1,\cdots,k_n,0)$-case genus zero $L_m^{\anc}$-constraint.
	By the symmetry of insertions, the $(n+1;k_1,\cdots,k_n,k_{n+1})$-case genus zero $L_m^{\anc}$-constraint holds if there is at least one of the $k_i$ equals to $0$. 
	If all $k_i\geq 1$ for all $i=1,\cdots,n+1$, then $\sum_{i=1}^{n+1}k_i\geq n+1\geq n$.
	The equations also trivially hold due to the dimensional reason.
	Thus, by induction, all the $(n; k_1,\cdots,k_n)$-cases of the genus zero $L_{m}^{\tau}$-constraints are proved.
\end{proof}
\subsection{Genus-1 ancestor Virasoro constraints}
We have precisely
\begin{align}
	\mathscr L^{\anc}_{1,m}({\bf s})=&\, -\frac{1}{4} \sum_{a+b=m} \tr(\E^a(\mu+\tfrac{1}{2})\E^b(\mu-\tfrac{1}{2}))
	+\corr{D_{\E,\bar\psi}^{m+1}\bar\psi^{-1}\tilde\bs(\bar\psi)}^{\tau}_{1,1}\nonumber\\
	&\, +\frac{1}{2}\sum_{k=1}^{m}(-1)^k\corr{\phi_a\bar\psi^{k-1}, [D^{m+1}_{\E,\bar\psi}\phi^a\bar{\psi}^{-k-1}]_{+}}^{\tau}_{0,2}\nonumber\\
	&\, +\sum_{k=1}^{m}(-1)^k\corr{\phi_a\bar\psi^{k-1}}^{\tau}_{0,1}\corr{[D_{\E,\bar\psi}^{m+1}\phi^a\bar\psi^{-k-1}]_{+}}^{\tau}_{1,1}. \label{eqn:genus-one-virasoro-1}
\end{align}
Similarly as the genus-0 case, we would like to mention that the genus-1 ancestor Virasoro constraints 
can be deduced from the constraints at ${\bf s}=0$ and the genus-1 tautological relations.
Still we give a proof here.
\begin{proof}[Proof of the second part of Theorem~\ref{thm:anc-vira}]
	Notice that at ${\bf s=0}$, equation~\eqref{eqn:genus-one-virasoro-1} reads
	$$\textstyle
		\<D_{\E,\bar\psi}^{m+1}\vac^{\tau}(\bar\psi)\>^{\tau}_{1,1}=-\frac{1}{4} \sum_{a+b=m} \tr(\E^a(\mu+\tfrac{1}{2})\E^b(\mu-\tfrac{1}{2})).
	$$
	For $m\geq 0$, notice that $D_{\E,z}^{m+1}\vac^{\tau}(z)=[\E^{m+1}\vac^{\tau}(z)]_{\leq 1}+\sum_{a+b=m}\E^a(\mu+\frac{3}{2})\E^b{\bf 1}z+O(z^2)$,
	and by QDE and the homogeneous condition of the vacuum vector, $$\E^{m+1}\vac^{\tau}(z)=E^{m+1}-E^{m}*((\mu+\tfrac{\delta}{2}){\bf 1})z+O(z^2).$$
	By topological recursion relation, and by noticing $(\mu+\frac{\delta}{2}){\bf 1}={\bf 1}-\nabla_{\bf 1}E$, we get equation~\eqref{eqn:vira-genusone}.
	
	Now we consider the coefficient of $\prod_{i=1}^n s_{k_i}^{b_i}$, $n\geq 1$, in equation~\eqref{eqn:genus-one-virasoro-1}, the equation reads
	\begin{align}
		&\,\<\phi_{b_{[n]}}\bar\psi^{k_{[n]}},D_{\E,\bar\psi}^{m+1}\vac^{\tau}(\bar\psi)\>^{\tau}_{1,n+1}\nonumber\\
		=&\, \sum_{i=1}^{n}\<\phi_{b_{[n]-\{i\}}}\bar\psi^{k_{[n]-\{i\}}},[D_{\E,\bar\psi}^{m+1}\phi_{b_i}\bar\psi^{k_i-1}]_{+}\>^{\tau}_{1,n}\nonumber\\
		&\, +\frac{1}{2}\sum_{k=1}^{m}(-1)^k\<\phi_a\bar\psi^{k-1}, [D^{m+1}_{\E,\bar\psi}\phi^a\bar{\psi}^{-k-1}]_{+},\phi_{b_{[n]}}\bar\psi^{k_{[n]}}\>^{\tau}_{0,2+n} \nonumber\\
		&\,  +\sum_{k=1}^{m}(-1)^k\sum_{I\sqcup J=[n]}\<\phi_a\bar\psi^{k-1},\phi_{b_I}\bar\psi^{k_I}\>^{\tau}_{0,|I|+1} \<[D_{\E,\bar\psi}\phi^a\bar\psi^{-k-1}]_{+},\phi_{b_J}\bar\psi^{k_J}\>^{\tau}_{1,|J|+1}.\label{eqn:vira-genus1-trr}
	\end{align}
	For the case with $k_i\geq 1$, $i=1,\dots,n$, by dimensional reason, the equation becomes
	$$\textstyle
	\<\phi_{b_{[n]}}\bar\psi^{k_{[n]}},D_{\E,\bar\psi}^{m+1}\vac^{\tau}(\bar\psi)\>^{\tau}_{1,n+1}
	=\sum_{i=1}^{n}\<\phi_{b_{[n]-\{i\}}}\bar\psi^{k_{[n]-\{i\}}},[D_{\E,\bar\psi}^{m+1}\phi_{b_i}\bar\psi^{k_i-1}]_{+}\>^{\tau}_{1,n}.
	$$
	The only non-trivial cases are $k_i=1$, $i=1,\cdots,n$ or there is a $k_j=2$ and $k_i=1$ for $i\ne j$.
	These equations can be proved by the following genus one topological recursion relations:
	$$\textstyle
	\<\phi_{a_1}\bar\psi,\cdots,\phi_{a_n}\bar\psi\>^{\tau}_{1,n}=\frac{(n-1)!}{24}\sum_{\sigma_1,\cdots,\sigma_n}\<\phi_{\sigma_1},\phi_{a_1},\phi^{\sigma_2}\>^{\tau}_{0,3}\cdots \<\phi_{\sigma_{n}},\phi_{a_n},\phi^{\sigma_1}\>^{\tau}_{0,3},
	$$
	and 
	$$\textstyle
	\<\phi_{a_1}\bar\psi,\cdots,\phi_{a_{n-1}}\bar\psi^2, \phi_{a_n}\>^{\tau}_{1,n}=\frac{(n-1)!}{24}\sum_{\sigma_1,\cdots,\sigma_n}\<\phi_{\sigma_1},\phi_{a_1},\phi^{\sigma_2}\>^{\tau}_{0,3}\cdots \<\phi_{\sigma_{n}},\phi_{a_n},\phi^{\sigma_1}\>^{\tau}_{0,3}.
	$$
	If there is some $k_i=0$, then by similar method as the proof of the first part of Theorem~\ref{thm:anc-vira}, equation~\eqref{eqn:vira-genus1-trr} can be deduced from the ones with less insertions.
\end{proof} \end{appendices}

\end{document}